\pdfoutput=1

\documentclass{article}

\usepackage[english]{babel}

\usepackage[letterpaper,top=1in,bottom=1in,left=1in,right=1in]{geometry}

\usepackage[dvipsnames]{xcolor}
\usepackage{tikz}
\usetikzlibrary{arrows.meta,positioning}
\usepackage{amsmath,amsthm,amssymb}
\usepackage{graphicx}
\usepackage[colorlinks=true, allcolors=blue]{hyperref}
\usepackage[nameinlink]{cleveref}
\usepackage[linesnumbered,ruled,vlined,algo2e]{algorithm2e}
\usepackage{algorithm}
\usepackage{algpseudocode}
\usepackage{enumitem}
\usepackage[sort&compress,numbers]{natbib}
\usepackage{svg}
\usepackage{thm-restate}
\usepackage[skins,many]{tcolorbox}
\usepackage{bm}

\newtheorem{theorem}{Theorem}
\newtheorem{corollary}[theorem]{Corollary}
\newtheorem{lemma}[theorem]{Lemma}
\newtheorem{definition}[theorem]{Definition}
\newtheorem{claim}[theorem]{Claim}
\newtheorem{observation}[theorem]{Observation}

\usepackage{hyperref}
\hypersetup{
    colorlinks = true,
    allcolors = black
}

\usepackage[nameinlink]{cleveref}

\newcommand{\A}{\mathcal{A}}
\newcommand{\N}{\mathbb{N}}
\newcommand{\Z}{\mathbb{Z}}
\newcommand{\T}{\mathcal{T}}
\renewcommand{\P}{\mathsf{P}}

\newcommand{\maxID}{\mathcal{ID}_{\text{max}}}

\newcommand{\DS}{\mathcal{DG}}
\newcommand{\maxIDi}{\maxID^{(i)}}
\newcommand{\Lmax}{L_{\text{max}}}

\newcommand{\Sin}{\Sigma_{\text{in}}}
\newcommand{\C}{\mathcal{C}}
\newcommand{\Sout}{\Sigma_{\text{out}}}

\newcommand{\set}[1]{\left\{#1\right\}}

\newcommand{\fA}{\mathcal A}

\newtcolorbox{myframe}[2][]{%
	breakable,enhanced,colback=white,colframe=black,coltitle=black,
	sharp corners,boxrule=0.4pt,
	fonttitle=\itshape,
	attach boxed title to top left={yshift=-0.3\baselineskip-0.4pt,xshift=2mm},
	boxed title style={tile,size=minimal,left=0.5mm,right=0.5mm,
		colback=white,before upper=\strut},
	title=#2,#1
}

\usepackage{authblk}

\title{The Distributed Complexity Landscape on Trees\\ Depends on the Knowledge About the Network Size}

\author[1]{Alkida Balliu}
\author[2]{Sebastian Brandt}
\author[3]{Fabian Kuhn}
\author[1]{Dennis Olivetti}
\author[4]{Timothé Picavet}
\author[3]{Gustav Schmid}

\affil[1]{Gran Sasso Science Institute, L'Aquila, Italy}
\affil[2]{CISPA Helmholtz Center for Information Security, Saarbrücken, Germany}
\affil[3]{University of Freiburg, Germany}
\affil[4]{LaBRI, Université de Bordeaux, France}

\date{}

\begin{document}
\maketitle


\begin{abstract}

One of the most successful theoretical models in distributed computing is the LOCAL one, introduced in a seminal work by Linial [SIAM J.\ Comp.\ 1992]. Over the years, when studying distributed graph problems in the LOCAL model, researchers made different assumptions on the exact details of this model. For example, sometimes it is assumed that all machines know the exact size of the network, other times machines are assumed to only know a polynomial upper bound on the size of the network, while sometimes no prior knowledge is assumed. Are these small differences irrelevant details or do they actually heavily affect the obtained results? We investigate how robust our current understanding of the LOCAL model truly is, by focusing on one of the most studied classes of problems, called Locally Checkable Labelings (LCLs).

LCLs are graph problems for which correct solutions can be described by listing a finite set of valid constant-radius neighborhoods. 
Since Naor and Stockmeyer introduced LCLs [FOCS 1995], understanding them has been in the center of attention, and, in the last 10 years, researchers were able to make a lot of progress. For example, Chang, Kopelowitz, and Pettie [FOCS 2016] showed that the randomized complexity of any LCL problem on $n$-node graphs is at least its deterministic complexity on $\sqrt{\log n}$-node graphs. Later, Chang and Pettie [FOCS 2017], showed that, on bounded-degree trees, any randomized algorithm solving an LCL in $n^{o(1)}$ rounds can be automatically transformed into a deterministic algorithm with runtime $O(\log n)$. Then, Balliu, Hirvonen, Korhonen, Lempiäinen, Olivetti, and Suomela [STOC 2018] showed that these kind of automatic speedups are no longer possible for general bounded-degree graphs.
The above-mentioned results make use of the assumption that the nodes have, for free, prior knowledge of $n$. How much does this assumption affect the beautiful theory of LCLs as we know it nowadays?

It turns out that, perhaps surprisingly, if we were to consider a setting where nodes are oblivious of $n$, or if we relax the setting such that nodes know a polynomial upper bound of $n$, already on trees, the theory of LCLs looks quite different from the one we currently know. In fact, while the fundamental classification of problems seems to remain the same, our results show that the picture becomes much more complex: for example, there are \emph{more} cases in which randomness helps in solving LCLs faster; there are problems with \emph{very unnatural} complexities; and for some problems the exact lower bound even depends on which definition of $\Omega$ we use!

\end{abstract}
\thispagestyle{empty}

\clearpage
\thispagestyle{empty}

\tableofcontents
\clearpage
\setcounter{page}{1}

\section{Introduction}

Over the last 40 years, researchers have made remarkable progress in advancing our knowledge in the field of theoretical distributed computation. In the process of exploring such a vast area of research, it is only natural to study various meaningful models of distributed computing, each focusing on different aspects. 
One of the most successful theoretical models in distributed computing is the LOCAL one, introduced in a seminal work by Linial~\cite{Linial92}. This is a synchronous message passing model, where a network is modeled as a graph, in which nodes represent machines, and edges represent communication links.

When striving to understand the unknown, it helps significantly to simplify the context or, sometimes, the task at hand. 
For instance, when trying to design fast algorithms for distributed graph problems in the LOCAL model, it is very often assumed that the nodes of the graph are given, for free, \emph{prior knowledge of the total number $n$ of nodes} in the graph. In fact, in the literature, some algorithms are designed to work when the exact value of $n$ is known a priori, others require to know only a polynomial approximation of $n$, while still for others this global knowledge of $n$ is not needed at all.

Assuming prior knowledge of the total number $n$ of nodes is indeed quite useful: it allows us to design distributed algorithms that state operations like ``first, each node $v$ checks whether some property is satisfied in its $O(\log^2 n)$-radius neighborhood, then ...''. Such an operation requires that the algorithm is given a polynomial upper bound on $n$. However, while in the centralized setting knowing the size of the input is a natural assumption, in the distributed setting this is not quite realistic: think, e.g., of networks that change over time, or huge networks that span the whole world. Moreover, for many interesting distributed graph problems, learning the total number of nodes in a graph potentially requires much more time than solving the problem with the best known algorithm that receives this knowledge for free.

Therefore, it is natural to ask the following questions. What can be solved efficiently, in the distributed setting, with algorithms that are not given as input the value of $n$ or if nodes are only given an imprecise estimate of $n$? Are current algorithms heavily relying on the knowledge of $n$, and if they are, can we turn them into algorithms that do not have access to this information while maintaining the same runtime? How much would our knowledge change in a setting in which an algorithm is not given any prior knowledge of global parameters? Ultimately, how robust is our current knowledge in the theory of distributed computing? 

\paragraph{Some algorithms can be turned into uniform ones.}
The question of whether we can obtain \emph{uniform algorithms}~\cite{KSV11}, that is algorithms which are oblivious to global parameters, has already been discussed in a few cases. In some cases, researchers have made an effort and asked whether the algorithms they were presenting could be made uniform. For example, in \cite{BarenboimE10}, Barenboim and Elkin first gave an algorithm for solving MIS when knowing $n$, and then, by paying some overhead on the runtime, they show how to turn their algorithm into one in which no knowledge about $n$ is necessary. Moreover, there are some generic approaches that one could try for obtaining uniform algorithms. For instance, one could try to \emph{guess} the value $n$ of the total number of nodes: one can first run a non-uniform algorithm with some value $N=N_0$, then, if it fails, one can try again with some larger $N$, until one succeeds. By how much one increases the guess depends on the runtime of the given algorithm. Korman, Sereni, and Viennot refined this idea and designed a technique that proved useful for making existing algorithms uniform \cite{KSV11}. On a high level, their technique works as follows. They first run the algorithm with a guessed value of $n$: after this, some nodes are going to succeed, while some others are going to fail (think of the $(\Delta + 1)$-vertex coloring problem; nodes that fail have got a color that conflicts with some neighbor). The nodes that succeed keep their output (their color), while the others run the algorithm with a suitably increased guess of $n$. The authors used this technique on many existing algorithms that solve variants of maximal independent set and vertex/edge colorings, obtaining uniform algorithms that asymptotically have the same runtime as the non-uniform ones. However, this technique can be applied only to a specific kind of algorithms. In fact, in order to be able to apply it, the algorithm must be able to provide an output that is compatible with the outputs of nodes that have already terminated.

\paragraph{Parts of the distributed theory inherently can't be made uniform.}
Despite the efforts of understanding distributed computing in the uniform setting, that is in the setting where algorithms may not leverage global parameters, many important results in the literature strongly rely on prior knowledge of global parameters, and previous techniques do not help in obtaining uniform and efficient algorithms. One striking example is the successful line of research that studies Locally Checkable Labeling (LCL) problems. The research community put a lot of effort into understanding LCLs and yet a major part of the theory that researchers have built crumbles in an oblivious-to-$n$ setting.
In this paper, we make a first step forward in understanding this extensively-studied class of problems in the uniform setting and in settings where the knowledge about $n$ is less accurate. Perhaps surprisingly, our results show that even when restricted to bounded-degree trees, when $n$ is not known a priori, the theory of LCLs looks vastly different than the one we have known so far, and it gets more complex as well.

LCLs were introduced in 1995 in the seminal work of Naor and Stockmeyer \cite{NaorStockmeyer95}, and, informally, they are a class of problems for which correct solutions can be specified by listing a finite set of valid labeled constant-radius neighborhoods (plus some additional restrictions). Observe that this definition implies that the correctness of a solution for an LCL problem can be checked, distributedly, in constant communication rounds. Many well-studied problems fall into this class: maximal matching, maximal independent set, vertex coloring, etc. Since its introduction, the LCL class of problems has been extensively studied over the years (see e.g., \cite{lcls_on_paths_and_cycles,balliu19lcl-decidability,B0COSS22_LCLregularTrees,BHOS19HomogeneousLCL,CP19timeHierarchy,LCLs_in_rooted_trees,binary_lcls,brandt21trees,bcmos21,B0FLMOU22,chang20,lcls_on_grids,BBOS18almostGlobal,BBOS20paddedLCL,BHKLOS18lclComplexity,CKP19exponential}), and one very well-studied aspect is the complexity landscape of LCLs: what are possible complexities of LCLs? For example, nowadays we know that there cannot be any LCL with deterministic complexity $\Theta(\sqrt{\log n})$. Moreover, it is worth mentioning that investigating LCLs has given insights and has helped in understanding problems outside the LCL class. For example, while LCLs are problems defined on graphs of bounded degree, the \emph{round elimination} technique, which has been developed to understand the complexity of LCLs \cite{brandtlower16}, has then been used successfully to show many different lower bounds in the unbounded-degree setting \cite{Balliu2019, trulytight,rs-siam,outdegree-ds,Balliu0KO22,mm-hypergraphs}.

An interesting graph class where LCLs have been studied is \emph{trees}. The graph class of trees is very important, not only because it is well-studied and natural, but also because, in the distributed setting, for many interesting natural problems, trees appear to be hard instances. For example, most lower bounds proved via the round elimination technique hold already on trees. Nowadays, the general feeling is that we know everything regarding the landscape of the deterministic and randomized complexities of LCLs on trees. In fact, we know that the deterministic complexity of any LCL problem on trees is one of the following: $O(1)$, $\Theta(\log^* n)$, $\Theta(\log n)$, or $\Theta(n^{1/k})$ for all integers $k>0$~\cite{CP19timeHierarchy}. The complexity gaps are constructive: if we design an algorithm that solves an LCL on trees in e.g., $O(\log^2 n)$ rounds, then we can automatically speed it up and reduce its complexity to $O(\log n)$. We also know that randomness may \emph{only} help for problems with deterministic complexity $\Theta(\log n)$, and if it does help, it helps exponentially \cite{CKP19exponential}. 

The case of trees already highlights a big issue suffered by the beautiful distributed complexity theory of LCLs that has been built over the years. For the so called lower regime $O(\log n)$ much is already known, while for the polynomial regime complexities $\Theta(n^{1/k})$ the picture is unclear. 
LCLs that have complexity $O(\log n)$ when $n$ is known to the algorithm, can be solved with normal-form algorithms that do not require any knowledge of $n$ \cite{CP19timeHierarchy, CKP19exponential}. See also the work of Brandt et al.~\cite{uniform_algorithms}, in which randomized uniform algorithms for LCLs in regular trees are considered. In that work, they consider algorithms such that for any $\varepsilon>0$ they have the guarantee that any node running for at least $T(\varepsilon)$ rounds has terminated with probability at least $1-\varepsilon$. They show the following equivalences between the standard randomized non-uniform LOCAL complexities and the uniform LOCAL complexities: $O(1) = O(1)$, $\Theta(\log^* n)=\Theta(\log^* 1/\varepsilon)$, $\Theta(\log \log n) =\Theta(\log \log 1/\varepsilon)$ (where we should think of $\varepsilon = \frac{1}{n^c}$).

It seems reasonable to hope that similar results can be proven about the polynomial regime complexities $\Theta(n^{1/k})$, but it is unclear upon closer inspection. In fact, all \emph{upper bounds} of $O(n^{1/k})$ are shown by assuming that some value $N = \Theta(n)$ is given to the nodes and the algorithms heavily rely on this. If we instead give some value $N = \omega(n)$ to the nodes, then the upper bounds get asymptotically worse, and it is not known whether a tight $\Theta(n^{1/k})$ complexity (for all integers $k>1$) can be achieved without relying on such knowledge.
In other words, some known facts about the distributed complexity theory of LCLs may \emph{collapse}, and hence the following questions arise naturally. When relaxing the assumption that nodes have knowledge about the network size: Is it really true that randomness may help only for problems with deterministic complexity $\Theta(\log n)$? Is it really true that there are no LCLs with deterministic complexity in the range $\omega(\sqrt{n})$ to $o(n)$? How robust is our current knowledge of LCLs? In this work, we try to bring some clarity to these issues.

\subsection{Some Useful Background}
In a cornerstone work, Chang and Pettie \cite{CP19timeHierarchy} proved that, on trees, there can be no LCL with a complexity that is in the range $\omega(\log n)\cap n^{o(1)}$. After that, the techniques used in \cite{CP19timeHierarchy} have been used to prove many other different \emph{complexity gaps} \cite{chang20,BBOS18almostGlobal}. In essence, in the $\omega(\log n)$ region, there can only be problems with complexity $\Theta(n^{1/k})$, for all integers $k>0$, and, as shown in \cite{CP19timeHierarchy,chang20}. We consider two natural examples of problems with such complexities: $k$-hierarchical $2 \frac{1}{2}$-coloring and $k$-rake-and-compress decomposition. 
The latter one of these problems is the theoretically more interesting problem since it is \emph{complete} for the class of problems solvable in $\Theta(n^{1/k})$ rounds. That is, if we have a fast algorithm for $k$-rake-and-compress decomposition, we can also solve any other LCL with complexity $O(n^{1/k})$ in the same asymptotic time \cite{chang20,BBOS18almostGlobal}. However, in the spirit of both brevity and presenting our ideas in a more understandable manner, we restrict our more involved proofs to the easier to deal with problem of $k$-hierarchical $2 \frac{1}{2}$-coloring. We note that $k$-hierarchical $2 \frac{1}{2}$-coloring is a natural candidate, since (1) it was the first family of problems revealing the $\Theta(n^{1/k})$ classes, and (2) its study eventually lead to the classification of LCLs into these classes.

\paragraph{Hierarchical coloring.}
The problem of $k$-hierarchical $2 \frac{1}{2}$-coloring requires us to produce a proper $2$-coloring, but only on some parts of a tree. We note that this is an artificial problem and its description may sound a bit convoluted.

\begin{figure}
    \centering
    \includegraphics[width=\linewidth]{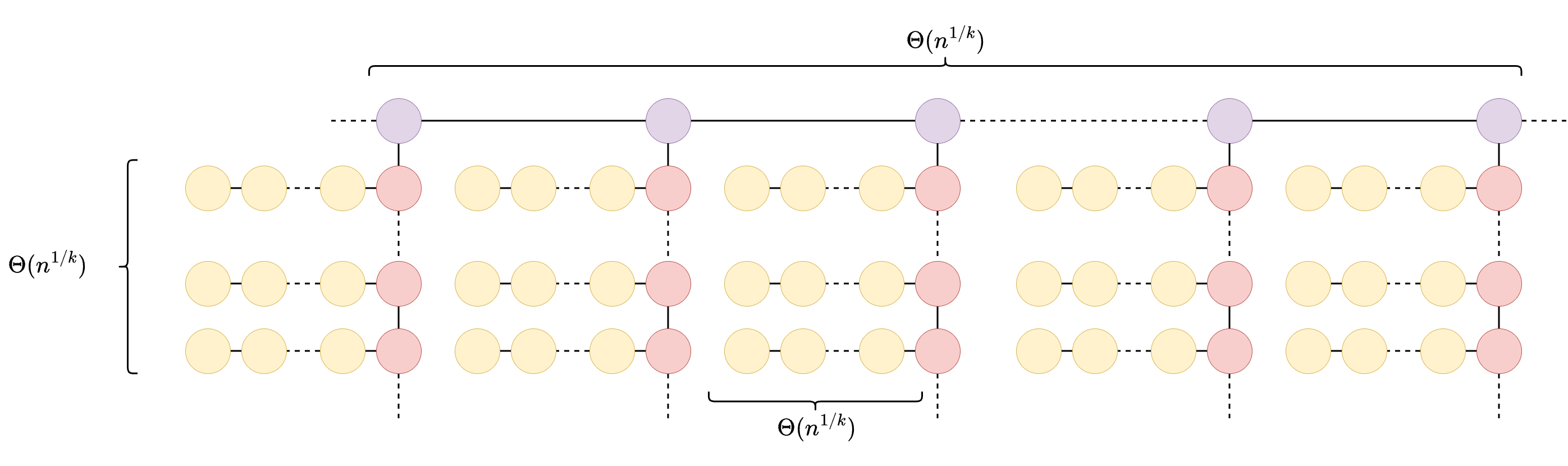}
    \caption{The lower bound instance for $k$-hierarchical $2 \frac{1}{2}$-coloring, for $k = 3$. Level-3 nodes are purple and form a path at the top. To each level-3 node, we attach a path of red level-2 nodes, and to each level-2 node we attach a path of yellow level-1 nodes. All paths have length $\Theta(n^{1/3})$.}
    \label{fig:halfColoringGraph}
\end{figure}

We explain the problem using the canonical lower bound instance shown in \Cref{fig:halfColoringGraph}. The graph is organized into hierarchical levels obtained by iteratively removing nodes of degree at most~$2$.

Formally, level~1 consists of all nodes of degree at most~$2$ (shown in yellow). After removing these nodes, level~2 consists of all remaining nodes of degree at most~$2$ (shown in red). More generally, level~$i$ consists of the nodes that have degree at most~$2$ after removing all nodes of levels $< i$. The parameter $k$ specifies the number of relevant levels; nodes of level larger than $k$ may output arbitrary labels. Since $k$ is a constant, membership in levels $\le k$ can be determined within a constant-radius neighborhood.

Each level induces a collection of pairwise disjoint paths. In the construction of \Cref{fig:halfColoringGraph}, with the particular case of $k=3$, all such paths have length $\Theta(n^{1/3})$.

Each level~1 path must either be properly $2$-colored, or its nodes must unanimously be labeled \emph{decline}. If the nodes of a level~1 path output a $2$-coloring, then any adjacent level~2 node (at the two endpoints of the path) may output the label \emph{exempt}. The \emph{exempt} label is compatible with any other label, effectively removing all constraints and thus splitting a level~2 path into two independent subpaths.

Level~2 paths (after splitting at \emph{exempt} nodes) face the same choice: they must either produce a proper $2$-coloring or unanimously output \emph{decline}. This process continues up the hierarchy. At the final level~$k$ (level~3 in \Cref{fig:halfColoringGraph}), nodes are not allowed to output \emph{decline}, and therefore every level~3 path, excluding exempt level~3 nodes, must be properly $2$-colored.

\begin{itemize}
\item \textbf{Lower bound intuition (case \bm{$k=3$}):} In the instance of \Cref{fig:halfColoringGraph}, all level~1 paths have length $\Theta(n^{1/3})$. Producing a proper $2$-coloring of such a path requires $\Omega(n^{1/3})$ rounds. Hence, if any level~1 path outputs a $2$-coloring, the desired lower bound already follows.

Otherwise, all level~1 paths output \emph{decline}, and thus no level~2 node may output \emph{exempt}. They therefore remain of length $\Theta(n^{1/3})$. So, the same argument applies here: if any level~2 path is $2$-colored, this requires $\Omega(n^{1/3})$ rounds.

If both level~1 and level~2 paths output \emph{decline}, then the responsibility is pushed to level~3, where again no node is allowed to output \emph{exempt}. However, level~3 nodes are not allowed to output \emph{decline}, and thus must $2$-color their path, which again has length $\Theta(n^{1/3})$. Consequently, in all cases, some path of length $\Theta(n^{1/3})$ must be $2$-colored, implying an $\Omega(n^{1/3})$ lower bound.

\item \textbf{Algorithmic idea:} A matching upper bound of $O(n^{1/3})$ rounds is obtained as follows. Each node explores the path it belongs to up to distance $O(n^{1/3})$. If the path has length at most $O(n^{1/3})$, the nodes produce a proper $2$-coloring. Otherwise, the entire path outputs \emph{decline}.

At higher levels, nodes adjacent to already colored lower-level paths may output \emph{exempt}, thereby splitting longer paths into smaller ones. Nodes that cannot output \emph{exempt} must have attached a lower level path that outputs \emph{decline}, such a declined path must have length $\Omega(n^{1/3})$ and so, by a charging argument, there can be at most $O(n^{1-1/3})$ level-2 nodes that do not output \emph{exempt}. Repeating this idea and exploring the level-2 paths until length $O(n^{1/3})$, we ensure that the number of level-3 nodes that do not output \emph{exempt} is at most $O(n^{2/3- 1/3}) = O(n^{1/3})$. We then produce a proper two coloring of all level-3 paths in $O(n^{1/3})$ rounds.
\item \textbf{Role of knowledge of \bm{$n$}:} The algorithm requires nodes to know (an upper bound on) $n$ in order to determine when to stop exploring their path. If a node stops too early, too few nodes will output \emph{exempt}, leaving a long path at the final level. If it explores for too long, the runtime may exceed the $O(n^{1/k})$ bound.
\end{itemize}

\paragraph{Rake-and-compress decomposition.}
While \cite{CP19timeHierarchy} gave $k$-hierarchical $2 \frac{1}{2}$-coloring as a first example of problems with complexity $\Theta(n^{1/k})$ in trees, in \cite{chang20} it has been shown that, for each integer $k$, there is a problem that is \emph{complete} for the class of problems that can be solved in time $\Theta(n^{1/k})$. That is, any $\Theta(n^{1/k})$-round problem can be solved by using an algorithm for the so-called $k$-rake-and-compress decomposition.\footnote{Note that the $k$ in the names of both problems will be the same $k$ as in the complexities, that is $2$-hierarchical $2 \frac{1}{2}$-coloring has complexity $\Theta(\sqrt{n})$ and computing a $3$-rake-and-compress decomposition has complexity $\Theta(n^{1/3})$.} Informally, a rake operation is the process of removing all nodes of degree exactly 1 and a compress operation is the process of removing all nodes of degree exactly 2. The goal is to repeatedly use these two operations to remove the entire graph, while using the compress operation only $k-1$ times. 

More formally, we are required to partition the nodes of a given tree $T$ into $2k-1$ subsets $R_1, C_1, \ldots$, $R_{k-1}$, $C_{k-1}$, $R_{k}$, where the sets $R_1, R_2, \ldots$ are called \emph{rake layers} and the sets $C_1,C_2,\ldots$ are called \emph{compress layers}. The components of each rake layer have an edge orientation that forms a rooted tree, and the nodes of each compress layer form disjoint paths, whose endpoints again have a single outgoing edge towards a Rake node. Note that performing multiple rake operations in a row naturally leads to the removed nodes forming a rooted tree. Similarly a compress operation makes disjoint paths and the endpoints of the removed paths had only one neighbor that was not removed. So this problem is supposed to capture the process of repeatedly using these two operations to remove the entire graph.

\medskip

It is shown, in \cite{chang20}, that once such a decomposition has been computed, any LCL of complexity $\Theta(n^{1/k})$ can be solved in a time proportional to the diameter of the largest connected component induced by nodes belonging to the same $R_i$ layer. Moreover, by performing $O(n^{1/k})$ rakes (which produces the rake layer $R_1$) followed by 1 compress (which produces the compress layer $C_1$) and repeating this process $k$ times (without performing the $k$-th compress), it is possible to compute a decomposition where the diameter of the largest connected component of every $R_i$ is upper bounded by $O(n^{1/k})$~\cite{chang20}. 

The formal definition of this problem is given in \Cref{sec:decAndLandscape}. The known algorithms for computing such decompositions have the same issues suffered by the known algorithms for $k$-hierarchical $2 \frac{1}{2}$-coloring: they strictly rely on having prior knowledge of $n$.

\paragraph{Why we restrict to these problem families.}
In our work, we prove lower and upper bounds for these families of problems, and show that minimal changes on the assumptions that are usually made in the LOCAL model can drastically affect the complexity of these problems.
As $k$-rake-and-compress is a \emph{complete} problem for LCLs of complexity $\Theta(n^{1/k})$, we believe that all LCL problems of this class behave similarly. In order to avoid a long and technical analysis, we decide to only work on $k$ rake-and-compress decompositions and the hierarchical coloring problems. Generalizing our theorems would significantly obfuscate the interesting aspects of our lower bound and upper bound results, specifically, how we optimize in these new settings and the concepts behind lower bounds in these new settings.

\paragraph{Many possible assumptions.}
The assumptions made over the years when studying graph problems in the LOCAL model and similar models are many and of different flavors\footnote{For instance, the dynamic-LOCAL model~\cite{dynamic_local} is defined such that the algorithm knows the value of $n$. If we define the dynamic-LOCAL model in perhaps the most intuitive way, where there is no assumption on the future number of nodes, the algorithms and known results do not follow through. However, the algorithms that we present in this paper would directly work in this version of dynamic-LOCAL. This shows that, when studying a model of computation, these assumptions really matter, even outside the standard LOCAL model.}. For example, the above-described algorithms need to know a linear upper bound on $n$, and if we relax this requirement a bit by providing a polynomial upper bound on $n$, the complexity of these algorithms would get worse. However, some other algorithms in the literature work perfectly fine when given a polynomial upper bound on $n$, and others do not need to know $n$ at all.

Moreover, there are other types of assumptions made in the literature, that are not about the knowledge of $n$. For example, it is typically assumed that nodes are assigned unique IDs, and some algorithms rely on the fact that the range of possible IDs is \emph{small} (e.g., polynomial in the number of nodes). Moreover, some algorithms assume and rely on the fact that all nodes know the range of the possible IDs, while some algorithms work fine even without this knowledge.

\subsection{Our Contributions}
In the following, we summarize our results. As we will show, not only is the complexity landscape of LCLs much more diverse than what was previously known, but also, the complexities even depend on which definition of $\Omega$ we use!

\paragraph{No assumptions at all.} We start our investigation in the most restrictive version of the LOCAL model: nodes do not have access to randomness and are not given any prior knowledge about $n$. To avoid trivial impossibility results, we assume that each node is assigned a unique ID that can be an arbitrary natural number. In fact, without IDs, even in a graph of just two nodes connected by an edge, it is impossible to solve basic problems like $2$-coloring because the nodes are not able to break symmetry. In \Cref{sec:lb-no-input}, we first illustrate the difficulties that we face in this model by giving an $\Omega(n)$ lower bound for $k$-hierarchical $2\frac{1}{2}$-coloring, for any integer $k$. 

\begin{restatable}{theorem}{HalfColNoBounds}\label{thm:2half-no-bounds}
For all integers $k > 0$, $k$-hierarchical $2\frac{1}{2}$-coloring requires $\Omega(n)$ rounds, in the LOCAL model, where no bound on $n$, nor on the size of the ID space, is provided to the nodes.
\end{restatable}

Then, in \Cref{thm:decomposition-no-bounds}, we show that this lower bound also holds for computing $k$-rake-and-compress decompositions, which is the key component at the heart of the theory of LCLs in trees. 
\begin{restatable}{theorem}{RcNoBounds}\label{thm:decomposition-no-bounds}
For all integers $k > 0$, solving $k$-rake-and-compress requires $\Omega(n)$ rounds, in the LOCAL model, where no bound on $n$, nor on the size of the ID space, is provided to the nodes.
\end{restatable}

These results convey the following message.

\begin{myframe}{}
    Having no knowledge at all about the size of the network, and no guarantees on the ID space, makes all the polynomial classes collapse. Hence, in order to find any non-trivial solution to these problems, \emph{some} additional guarantees are required.
\end{myframe}

\paragraph{Polynomial upper bound on $n$ given.} 
As a next step, we investigate the commonly considered setting in which nodes are provided a polynomial upper bound on $n$. More formally, we consider the setting in which nodes are given two inputs $N$ and $c$, and are promised that $n \le N\le n^c$. On a high level, we show that, problems that have complexity $\Theta(n^{1/k})$ when $n$ is known, still have a polynomial complexity in this setting, but the exponent of this polynomial depends in a complicated way on $c$ and $k$.

For this setting, we will provide upper bounds on the complexity of computing $k$-rake-and-compress decompositions in \Cref{sec:NAlgo}. We will show that the complexities of the algorithms that we provide can be derived by solving a non-trivial optimization problem. This will result in highly unnatural complexities. For example, for $k = 3$ and $c = 3$, we will obtain an algorithm with complexity $O(n^{\frac{7 -\sqrt{13}}{6}})\approx O(n^{0.566})$ (which should be compared with the complexity $O(n^{1/3})$ when $n$ is known). Since computing such a decomposition is a \emph{complete} problem for problems with complexity $O(n^{1/k})$ in the standard setting, we obtain the result stated in \Cref{thm:rc-upper-bound}.

\begin{restatable}{theorem}{RcUpperBound}\label{thm:rc-upper-bound}
Let $\Pi$ be an LCL problem that, on trees, can be solved in $O(n^{\frac{1}{k}})$ rounds in the LOCAL model, when nodes are provided with a linear upper bound on $n$ and IDs are from some polynomial range.\\
Consider the LOCAL model where nodes are provided $N$ and $c$, such that $N$ is guaranteed to satisfy $n \le N \le n^c$.
Let $\alpha_1 < \frac{1}{c}$ be the unique real solution to the equation $1 = \left(\frac{1}{1-c\alpha_1}\right)^{k - i_0} i_0 \alpha_1$, where $i_0 := \left\lfloor\frac{1}{c \alpha_1}\right\rfloor$.\\
The LCL problem $\Pi$ can be solved in $O(n^{c\alpha_1} + \log^* I)$ rounds on trees, where $I$ is the largest ID assigned to any node.
\end{restatable}

In \Cref{sec:LBpolyPromise}, we will prove a matching lower bound. However, we will show that the actual complexities of these problems depend on which definition of $\Omega$ we use. In fact, there exist two incompatible definitions of $\Omega$: the one by Hardy-Littlewood and the one by Knuth. We recommend Knuths original article \cite{KnuthONotation} as a reference on the different notations.

In the Hardy-Littlewood version, $f(n)\in\Omega(g(n))$ states that there is a constant $c>0$ such that $f(n) \geq c\cdot g(n)$ for infinitely many, large enough, $n$. On the other hand, Knuth's version is stronger and states that there is a constant $c>0$, such that after some starting point $n_0$, we have $f(n)\geq c\cdot g(n)$ for all $n\ge n_0$.

Our lower bounds use the weaker Hardy-Littlewood definition of $\Omega$. We will show in \Cref{sec:Lowerbound} and \Cref{sec:Lowerbound2} that the upper bounds that we provide for $k$-hierarchical $2\frac{1}{2}$-coloring and for $k$-rake-and-compress decomposition are tight\footnote{Technically there is still a $\log^*(I)$ term in \Cref{thm:rc-upper-bound}, but this is negligible as long as the largest ID $I$ is bounded by something like $2^{2^{2^n}}$.}. More precisely, we prove \Cref{thm:lb-known-N} and \Cref{thm:lb-known-N-rc}.

\begin{restatable}{theorem}{HalfColLBPoly}\label{thm:lb-known-N}
    Consider the LOCAL model where nodes are provided $N$ and $c$, such that $N$ is guaranteed to satisfy $n \le N \le n^c$, and nodes are assigned unique IDs in $\{1, \ldots, n^c\}$.
    Let $\alpha_1 < \frac{1}{c}$ be the unique value satisfying $1 = \left(\frac{1}{1-c\alpha_1}\right)^{k - i_0} i_0 \alpha_1$, where $i_0 := \left\lfloor\frac{1}{c \alpha_1}\right\rfloor$. \\
    Then, for the Hardy-Littlewood definition of $\Omega$, the $k$-hierarchical $2\frac{1}{2}$-coloring problem requires $\Omega(n^{c\alpha_1})$ deterministic rounds.
\end{restatable}

\begin{restatable}{theorem}{RcLBPoly}\label{thm:lb-known-N-rc}
    Consider the LOCAL model where nodes are provided $N$ and $c$, such that $N$ is guaranteed to satisfy $n \le N \le n^c$, and nodes are assigned unique IDs in $\{1, \ldots, n^c\}$.
    Let $\alpha_1 < \frac{1}{c}$ be the unique value satisfying $1 = \left(\frac{1}{1-c\alpha_1}\right)^{k - i_0} i_0 \alpha_1$, where $i_0 := \left\lfloor\frac{1}{c \alpha_1}\right\rfloor$. \\
    Then, for the Hardy-Littlewood's definition of $\Omega$, the $k$-rake-and-compress decomposition problem requires $\Omega(n^{c\alpha_1})$ deterministic rounds.
\end{restatable}

While it is a bit unsatisfying to use this weaker notion of $\Omega$, it is in fact unavoidable. For the more commonly used definition of $\Omega$ due to Knuth, we will show, in \Cref{ssec:ubknuth}, that proving $\omega(n^{1/k})$ lower bounds is not possible. For this purpose, we design an algorithm that, for infinitely many values of $n$, is able to solve these problems in $O(n^{1/k})$ rounds. More precisely, we prove the following.

\begin{restatable}{theorem}{NoKnuthLB}\label{thm:no-kuth-LB}
Let $\Pi$ be an LCL problem that, on trees, has complexity $\Theta(n^{\frac{1}{k}})$ in the LOCAL model, when nodes are provided with a linear upper bound on $n$ and IDs are from some polynomial range.\\
Consider the LOCAL model where nodes are provided $N$ and $c$, such that $N$ is guaranteed to satisfy $n \le N \le n^c$.\\
Then, for Knuth's definition of $\Omega$, the problem $\Pi$ is not solvable in $\Omega(T + \log^* I)$ rounds, for any $T$ that is asymptotically strictly larger than $n^{1/k}$, where $I$ is the largest ID assigned to any node.
\end{restatable}

Our results for this setting can be summarized as follows.
\begin{myframe}{}
    By slightly changing the usual assumptions, that is, by assuming that nodes know a polynomial upper bound on $n$, rather than a linear upper bound on $n$, we obtain that LCL problems now exhibit very different complexities. However, the added complexity can be dealt with in a structured and constructive way.
\end{myframe}

\paragraph{Promise on the ID space.}
In \Cref{sec:promiseIds}, we consider the setting in which nodes are not given any upper bound $N$ on $n$, but they are only given a parameter $c$, and are promised that the largest ID assigned to the nodes of the graph is at most $n^c$. Observe that, while before nodes where given $N$ and $c$, in this case nodes are given \emph{just} $c$.

Perhaps surprisingly, we show that for some problems this knowledge is \emph{sufficient}. To be more precise, for $k$-hierarchical $2\frac{1}{2}$-coloring, we are able to obtain an algorithm that matches the lower bound (for the stronger setting in which $N$ is given) that we presented in \Cref{thm:lb-known-N}. Note that while the bounds of \Cref{thm:lb-known-N,thm:lb-known-N-rc} work for all LCLs, we present here an algorithm only for $2\frac{1}{2}$-coloring, since the technical details are already complicated enough for this problem.

\begin{restatable}{theorem}{IdAlgo}\label{thm:2.5ColIdsUB}
Consider the LOCAL model where nodes have no knowledge of $n$, but are assigned unique IDs in $\{1, \ldots, n^c\}$, and $c$ is known by all nodes.
Let $\alpha_1 < \frac{1}{c}$ be the unique value satisfying $1 = \left(\frac{1}{1-c\alpha_1}\right)^{k - i_0} i_0 \alpha_1$, where $i_0 := \left\lfloor\frac{1}{c \alpha_1}\right\rfloor$. 
Then, for all $k>1$, $k$-hierarchical $2\frac{1}{2}$-coloring can be solved in $O(n^{c\alpha_1})$ deterministic rounds.
\end{restatable}

This implies the following exciting possibility.
\begin{myframe}{}
    A promise on the size of the ID space may be just as good as having an actual upper bound on the number of nodes.
\end{myframe}

\paragraph{Randomized algorithms.}
The final setting examined in this work is the randomized setting, where nodes do not have any knowledge about $n$ (see \Cref{sec:rand}). In the randomized LOCAL model, nodes are given access to an arbitrary number of random bits and they typically do not have access to unique IDs. If each node knows a polynomial upper bound on the number of nodes, the nodes can use the randomness to generate unique $O(\log n)$-bit identifiers w.h.p. Even without any knowledge about $n$, the nodes can still use the randomness to generate unique (infinitely long) names, thus we consider nodes receive no additional input (in particular no form of IDs).

In the version of the randomized LOCAL model where a linear bound on $n$ is known, randomization does not help at all when solving the problem of $k$-hierarchical $2\frac{1}{2}$-coloring \cite{chang20}. However, in the setting without any bound on $n$, randomness actually provides an advantage over the corresponding deterministic case (the case of \Cref{thm:2half-no-bounds}).

We first show that the complexity of computing a $2$-hierarchical $2\frac{1}{2}$-coloring in the oblivious-to-$n$ setting is $\Theta(\frac{n}{\log n})$, in contrast to the deterministic complexity $\Theta(n)$. 

We then go on to show that our ideas can be pushed further to the $3$-hierarchical $2\frac{1}{2}$-coloring problem, where we can observe some interesting behavior. We give an algorithm that solves $3$-hierarchical $2\frac{1}{2}$-coloring in time $O(\frac{n}{f(n)})$, where $f(n)$ is a functional square root of $\log n$. That means that $f(n)$ is a function that satisfies $f(f(n)) = \log n$. Note that $f(n)$ is a function that grows much slower than any polynomial function, but also much faster than e.g., any polynomial in $\log n$ (cf.\ \Cref{subsec:halflog} for a more detailed discussion). 

We additionally prove that this is not just an artifact of our techniques, by also proving a lower bound of $\Omega\big(\frac{n}{f(n) \log n}\big)$, which makes the upper bound result nearly tight. We believe that this suffices to illustrate the possibly unexpected behavior that is emerging in this setting and we leave the full generalization as a open question for future research\footnote{Some back-of-the-envelope calculation suggest that the following generalization might be true. For $k\geq 2$, the complexity of the $k$-hierarchical $2\frac{1}{2}$-coloring problem seems to be $O\big(\frac{n}{f(n)}\big)$, where $f(\cdot)$ is a function such that $f^{(k-1)}(n)=\log n$. That is, for $k=2$, we have $f(n)=\log n$, for $k=3$, we have $f(f(n))=\log n$, for $k=4$, we have $f(f(f(n)))=\log n$, and so on.}.
Our results about $2$ and $3$-hierarchical $2\frac{1}{2}$-coloring are formally stated and proven in \Cref{thm:randomized_k_2_upper,thm:randomized_k_2_lower,thm:randomized_k_3_upper,thm:randomized_k_3_lower}.

\begin{myframe}{}
    In the standard LOCAL model, it is known that randomness helps on trees only for LCLs with deterministic complexity $\Theta(\log n)$. Our results imply that this statement is an artifact of the assumptions that are typically made. In more restricted settings, randomness can be used in non-trivial ways to get faster algorithms. 
\end{myframe}

\paragraph{Open questions.}
While our work implies many new open questions regarding the complexity of graph problems when we deviate from common assumptions, we highlight the following two open questions which are most related to this work.
\begin{itemize}
    \item Let $\Pi$ be an arbitrary LCL problem which, when $n$ is known, has complexity $\Theta(n^{1/k})$. What is the complexity of $\Pi$ in the setting in which an upper bound $N\le n^c$ on $n$ is given, and $c$ is known by all nodes? What about the setting in which IDs are bounded by $n^c$ and $c$ is known?
    \item Let $\Pi$ be an arbitrary LCL problem which, when $n$ is known, has complexity $\Theta(n^{1/k})$. What is its randomized complexity if $n$ is not known?
\end{itemize}
Note that, in our work, we answer these questions only for the problems of $k$-hierarchical $2 \frac{1}{2}$-coloring and $k$-rake-and-compress decompositions, and in the randomized case, only for $k=1$ and $k=2$.

\subsection{High-level Ideas}
We describe the challenges and key ideas in each of the discussed settings.

\paragraph{No assumptions at all.}
In this setting, where nodes do not have any knowledge about $n$ and IDs are unbounded, we prove $\Omega(n)$ lower bounds.

The key difficulty in this setting is that IDs carry no useful information. We illustrate the lower bound idea with the 2-hierarchical $2\frac{1}{2}$-coloring problem. Imagine a node $v$ located in the middle of a path with monotonically increasing IDs in both directions (see \Cref{fig:lB} on the left). From $v$'s perspective, every additional round only reveals one more hop of the path.

\begin{figure}[h!]
    \centering
    \includegraphics[width=0.8\linewidth]{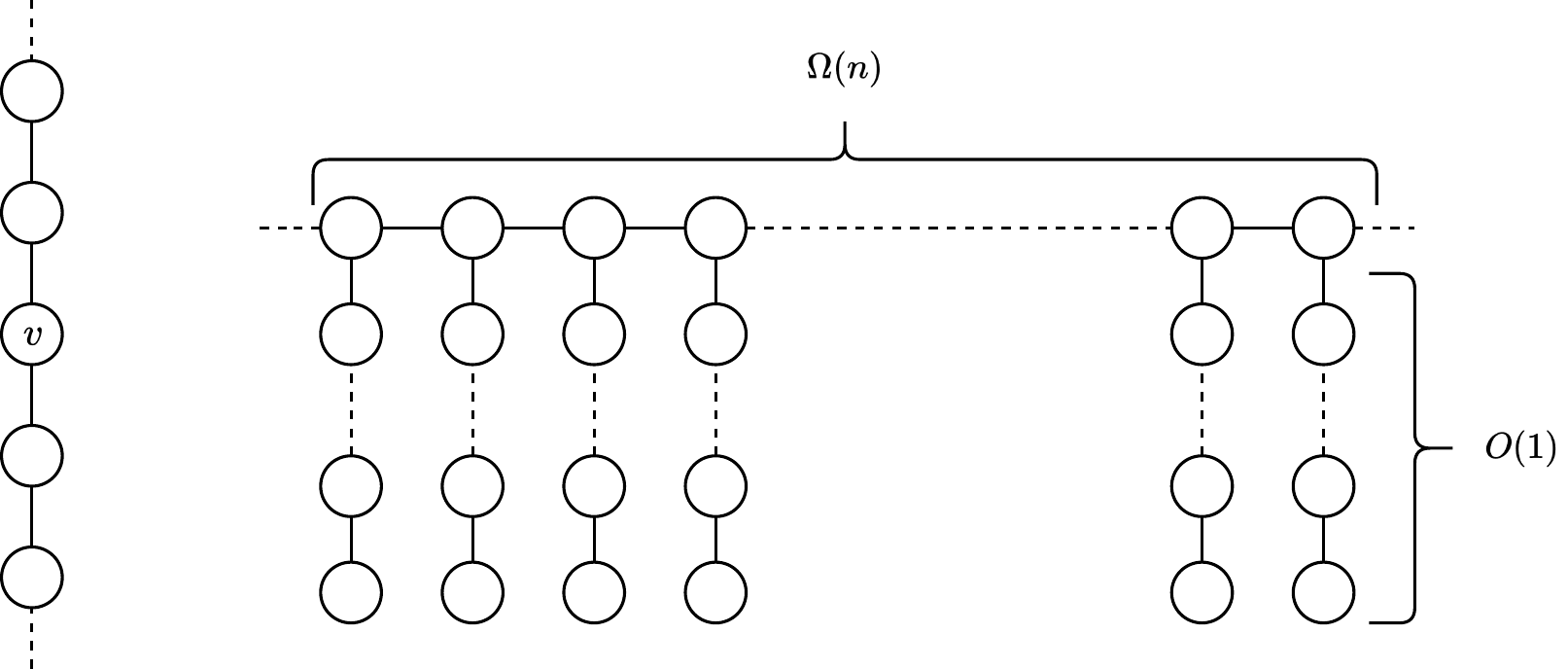}
    \caption{We depict the impossible trade off that an algorithm is faced with when solving 2-hierarchical $2\frac{1}{2}$-coloring in the oblivious setting. From the perspective of a level-1 node (on the left), the node can tell it is in a path, but it has no information of how long this path is compared to the remaining instance. If it is not willing to spend $\Theta(n)$ rounds, it will have to terminate in a constant number of rounds. But this results in the construction on the right, where all level-1 paths have constant length and decline too quickly, leaving a linear-in-$n$ length level-2 path on top that must be 2-colored.}
    \label{fig:lB}
\end{figure}

If the node keeps exploring the path for too many rounds, an adversary can choose $n$ so that the explored portion is in fact the entire graph, implying that the algorithm already spent $\Omega(n)$ rounds. Hence, to achieve sublinear runtime, nodes must terminate after a bounded number of rounds, without seeing the full extent of the path.

However, a node that has only seen a constant-radius neighborhood cannot safely produce a consistent $2$-coloring of the path. Therefore, level~1 nodes must output \emph{Decline} after a constant number of rounds. Since this decision is local, it propagates to all level~1 paths, which thus unanimously output \emph{Decline}.

But if we can force very short level-1 paths to decline, we can construct a hard instance by attaching many such level~1 paths to a long path of level~2 nodes (see \Cref{fig:lB} on the right). As all adjacent level~1 paths output \emph{Decline}, none of the level~2 nodes may output \emph{exempt}, and must produce a consistent $2$-coloring of the top path. Since this path has length $\Omega(n)$, this instance requires $\Omega(n)$ rounds to be solved.

The lower bound for computing a $k$-rake-and-compress decomposition essentially follows the same idea, but the technical details are more involved.

\paragraph{Polynomial upper bound on $n$ given.}
In the setting where nodes are provided with a constant $c$ and a polynomial upper bound $N \le n^c$, the previous trade off of the oblivious setting above becomes manageable.

For our upper bound, we adapt the original $O(n^{1/k})$ round $k$-rake-and-compress decomposition algorithm that is known to solve all LCLs~\cite{chang20}. That algorithm uses the two subroutines \emph{rake} and \emph{compress}.

\begin{figure}[htpb]
    \centering
    \begin{tikzpicture}
    [
        scale=0.75,
        every node/.style={circle,thick,draw=black},
        rake1/.style={fill=YellowGreen},
        compress1/.style={fill=CornflowerBlue},
        rake2/.style={fill=OrangeRed},
        every path/.style={very thick},
        oriented/.style={-{Latex[round,width=2mm,length=2mm]}},
    ]
		\node [style=rake1] (0) at (2.25, -1) {};
		\node [style=rake1] (1) at (3.75, -1) {};
		\node [style=rake1] (3) at (3, 0) {};
		\node [style=compress1] (4) at (3, 1.5) {};
		\node [style=rake1] (5) at (5.75, 0) {};
		\node [style=rake1] (6) at (7.75, 0) {};
		\node [style=rake1] (7) at (7.75, -1) {};
		\node [style=rake1] (8) at (5.75, -1) {};
		\node [style=rake1] (10) at (13, 0) {};
		\node [style=rake1] (11) at (16, 0) {};
		\node [style=rake1] (12) at (16, -1) {};
		\node [style=rake1] (13) at (13, -1) {};
		\node [style=compress1] (19) at (5.75, 1.5) {};
		\node [style=compress1] (20) at (7.75, 1.5) {};
		\node [style=compress1] (22) at (1.25, 1.5) {};
		\node [style=rake1] (23) at (1.25, 0) {};
		\node [style=rake1] (24) at (1.25, -1) {};
		\node [style=compress1] (25) at (0, 1.5) {};
		\node [style=compress1] (27) at (13, 1.5) {};
		\node [style=compress1] (28) at (16, 1.5) {};
		\node [style=compress1] (34) at (6.75, 1.5) {};
		\node [style=rake1] (36) at (17.5, 0) {};
		\node [style=rake1] (37) at (17.5, -1) {};
		\node [style=rake2] (38) at (17.5, 3) {};
		\node [style=rake1] (39) at (14.5, 0) {};
		\node [style=rake1] (40) at (14.5, -1) {};
		\node [style=compress1] (41) at (14.5, 1.5) {};
		\node [style=rake2] (42) at (12.5, 3) {};
		\node [style=rake1] (46) at (10.75, 0) {};
		\node [style=rake1] (47) at (11.5, -1) {};
		\node [style=compress1] (48) at (10.75, 1.5) {};
		\node [style=compress1] (49) at (12, 1.5) {};
		\node [style=rake2] (50) at (5.75, 4.5) {};
		\node [style=rake1] (51) at (10, -1) {};
		\node [style=rake1] (52) at (4.75, -1) {};
		\node [style=rake1] (53) at (6.75, -1) {};
		\node [style=rake2] (54) at (9, 3) {};
		\node [style=rake2] (55) at (-0.75, 3) {};
		\node [style=rake1] (56) at (-0.75, 0) {};
		\node [style=rake1] (57) at (-0.75, -1) {};
		\node [style=compress1] (58) at (5.75, 1.5) {};
		\node [style=rake1] (59) at (9, 0) {};
		\node [style=rake1] (60) at (9, -1) {};
		\node [draw=none,text=YellowGreen!80!Green] (61) at (-3, -0.5) {\large \bm{$R_1$}};
		\node [draw=none,text=CornflowerBlue!85!cyan] (62) at (-3, 1.5) {\large \bm{$C_1$}};
		\node [draw=none,text=OrangeRed] (63) at (-3, 4) {\large \bm{$R_2$}};
		\draw [style=oriented] (0) to (3);
		\draw [style=oriented] (8) to (5);
		\draw [style=oriented] (24) to (23);
		\draw (22) to (25);
		\draw [style=oriented] (12) to (11);
		\draw [style=oriented] (13) to (10);
		\draw [style=oriented] (7) to (6);
		\draw [style=oriented] (1) to (3);
		\draw [style=oriented] (3) to (4);
		\draw [style=oriented] (23) to (22);
		\draw [style=oriented] (5) to (19);
		\draw [style=oriented] (6) to (20);
		\draw [style=oriented] (11) to (28);
		\draw [style=oriented] (10) to (27);
		\draw [style=oriented] (37) to (36);
		\draw [style=oriented] (36) to (38);
		\draw [style=oriented] (40) to (39);
		\draw [style=oriented] (39) to (41);
		\draw (27) to (41);
		\draw (41) to (28);
		\draw [style=oriented] (47) to (46);
		\draw [style=oriented] (27) to (42);
		\draw [style=oriented] (46) to (48);
		\draw [style=oriented] (49) to (42);
		\draw (48) to (49);
		\draw [style=oriented] (4) to (50);
		\draw [style=oriented] (19) to (50);
		\draw [style=oriented] (42) to (50);
		\draw [style=oriented] (52) to (5);
		\draw [style=oriented] (53) to (5);
		\draw [style=oriented] (51) to (46);
		\draw [style=oriented] (57) to (56);
		\draw [style=oriented] (56) to (55);
		\draw [style=oriented] (25) to (55);
		\draw (34) to (20);
		\draw (34) to (58);
		\draw (22) to (4);
		\draw [style=oriented] (60) to (59);
		\draw [style=oriented] (59) to (54);
		\draw [style=oriented] (48) to (54);
		\draw [style=oriented] (28) to (38);
    \end{tikzpicture}
    \caption{An example of a rake-and-compress decomposition, obtained by applying, in order, 2 rake operations, a single compress operation, and 2 rake operations. The red and green rake layers $R_1,R_2$ form rooted trees, while the (in this case single blue) compress layer $C_1$ form disjoint paths.
    Notice that every edge incident to a rake vertex is oriented in accordance to the order induced by the layers.}
    \label{fig:r&c}
\end{figure}
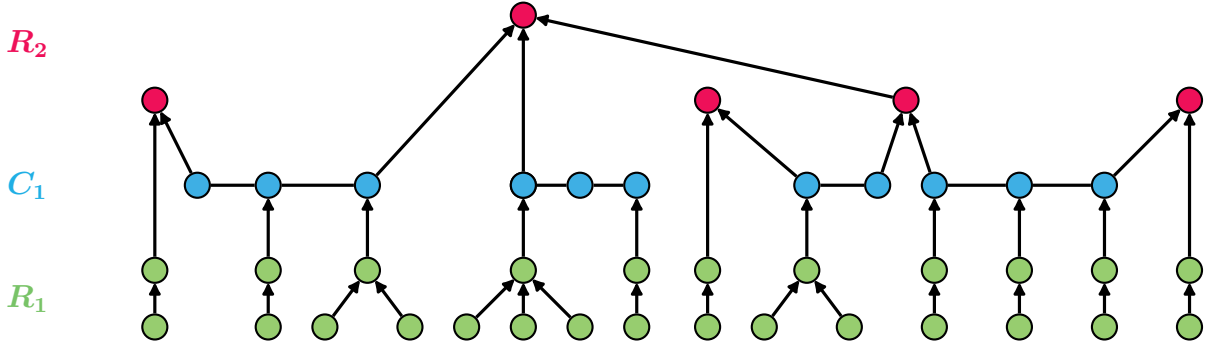

A rake operation is the removal\footnote{In this context, ``remove'' does not mean that the node physically disappears from the graph. Rather, it means that the node is marked as processed (or assigned to the current layer of the decomposition), and the algorithm proceeds on the residual graph induced by unprocessed nodes.} of all degree $\le 1$ nodes and a compress operation is the removal of all nodes of degree exactly 2. 

A rake-and-compress decomposition is then obtained by repeatedly performing $O(n^{1/k})$ rakes, followed by a single compress (see \Cref{fig:r&c} for an example). It can be shown that only one application of this is guaranteed to remove a $\frac{1}{n^{1/k}}$-fraction of all nodes. So, after repeating it $k-1$ times, only $O(n^{1/k})$ nodes remain. Those nodes can now be ``raked away'' with the final $O(n^{1/k})$ rake operations.

The trivial way to extend this approach to the new setting is to just do $O(N^{1/k})$ rakes followed by a compress, but this can be improved. 

Informally, the improvement over the $O(N^{1/k})$-round adaptation of the original algorithm comes from parameterizing the phases of the algorithm (see \Cref{sec:NAlgo}). Instead of using the same budget $N^{1/k}$ in each phase, we perform $N^{\alpha_i}$ rake steps in phase~$i$, for carefully chosen exponents $\alpha_i$.

There is a sweet spot in choosing the value of these exponents. To ensure a good runtime when $N$ is close to $n^c$, the $\alpha_i$'s should be small. However, if $N$ is large, even with a small $\alpha_1$, we have performed so many rakes that we have already removed a large portion of the graph. Therefore, we can afford to increase the other $\alpha_i$, because if $N$ is too large, we will already have removed the entire graph well before having performed $N^{\alpha_i}$ many rakes.

This allows us to increase the $\alpha_i$ across phases: if the algorithm does not terminate early, it implicitly indicates that $N$ is closer to $n$, and we can afford to spend more rounds. Balancing these effects leads to an optimization problem that determines the optimal choice of the $\alpha_i$ (see \Cref{ssec:optimization-problem}).

\paragraph{Why we can't use Knuths definition of $\Omega$:}
Before talking about our lower bounds for the polynomial setting, we discuss here why we use the Hard-Littlewood definition of $\Omega$ for our lower bounds. As already mentioned above, it is, in fact, impossible to show a lower bound of $\omega(n^{1/k})$ for this setting with a polynomial bound on $n$. To prove this, we give an algorithm that is slow for most values of $n$, but runs in $O(n^{1/k})$ for infinitely many values of $n$.

On a high level, we are able to achieve such a result as follows. For any given $N$, we will fix a strategy to guess what the true $n$ is. We then use this guess $N'$ as if it were the actual $n$. This guess will be correct for infinitely many $n$ (no matter which $N$ is given as input, we will correctly guess $n$). Whenever we have guessed the correct $n$, our algorithm will run in time $O(n^{1/k})$, which is sufficient to violate Knuth's definition of $\Omega$.

We have the promise that all nodes in our graph are given the same upper bound $n\le N\le n^c$ and all nodes know $c$. So, the nodes know that the true $n$ must be somewhere in the range $\set{N^{1/c}, \ldots, N}$.

To guess the correct $n$, we define a sequence of values that are very far apart, so that for any given $N$, only one element of the sequence falls in the range $\set{N^{1/c}, \ldots, N}$. For example, the recursive sequence defined below is one of them.
\begin{align*}
    s_1 &= 2, \\
    s_{i+1} &= s_i^{2c} ~~\text{for}~i\geq 1
\end{align*}
When $n$ is itself an element $s_i$ of the sequence, then no matter what $N$ is given to the nodes, we will know $s_i=n \in \set{N^{1/c}, \ldots, N}$. 
With this in mind, our guessing strategy is simply the following: if there exists an $s_i \in \set{N^{1/c}, \ldots, N}$, then act as if the actual number of nodes in the graph is $s_i$. This guess will be correct whenever $n=s_i$, no matter which $N$ is given as input. Hence, the algorithm will run in $O(n^{1/k})$ for all $n\in \{s_i\}_i$.

\paragraph{Lower bounds with a polynomial upper bound.}
The parameters $\alpha_i$ in our upper bound are derived from an optimization problem that captures the trade-off described above: choosing the $\alpha_i$ too large leads to excessive runtime when $N$ is close to $n^c$, while choosing them too small limits the progress made in early phases and does not make the remaining instance significantly easier. 

We then proceed to prove a lower bound of $n^{c\alpha_1}$. 
By showing a lower bound matching our upper bound, we show that this optimization problem perfectly represents this tradeoff: it exactly captures the difficulty of this setting.

Let us think of the conceptually easier $2\frac{1}{2}$-coloring problem. 
Remember that each node is assigned a level $1,\ldots, k$, where each level consists of disjoint paths. Conceptually, a level-$i$ node $v$ has at least one path of level-$(i-1)$ nodes attached to it, and if this attached path outputs a 2-coloring, then the constraints on $v$ are removed and $v$ may output \emph{exempt}. If instead, all level $i$ paths adjacent to $v$ output the \emph{decline} label, the constraints on $v$ remain. 

This behavior represents the tradeoff for the problem of $2\frac{1}{2}$-coloring; if all level-$i$ paths that output the decline label have length at least $L$, then if there are $n_i$ nodes of level $i$ in total, it follows that at most $n_{i}/L$ nodes of level $i+1$ are still constrained. This means that the remaining instance is small and hence easier. 

The proof extends the idea from the setting without knowledge of $N$ by playing a sort of adversarial game between the algorithm and the instance. There, we essentially asked the algorithm how long it is willing to explore a path of level-1 nodes until it commits to outputting a \emph{decline} label on such a path.
If it spends more than $N^{\alpha_1}$ rounds, then, by setting $N = n^c$, the algorithm runs for too long. 
If instead, the nodes decide to output decline when the path is much shorter than $N^{\alpha_1}$, the adversary sets $N = n$, which means the number of higher level nodes without properly 2-colored neighbors is very large. Then, even if the perfect trade-offs are made for all higher levels, we still have a long path that eventually needs to be properly 2-colored. 

Carefully repeating this argument over $k-1$ phases enforces the same trade-off throughout the execution. As a result, if an algorithm deviates from the behavior described in the optimization problem, it necessarily has to 2-color a path that is too long at some point. Thus, despite its complex appearance, the optimization problem exactly captures the aforementioned trade-off, yielding tight lower bounds.

\paragraph{Promise on the ID space.}
In our next result, we consider only a promise on the size of the largest ID instead of a promise on the concrete number of nodes. We have to deviate from the upper bound strategy described above, but we will still rely on the same tradeoff described by our optimization problem.

Since the nodes do not have any estimate of $n$ given as input, we instead have to rely on an indirect measure: the identifiers that they see. To this end, consider the node $v$ in the left part of \Cref{fig:lB}. Node $v$ will now decide to \emph{Decline} in the following way. Let $i$ be the current round number and let $\maxID$ be the largest ID that node $v$ has seen so far. Then $v$ will output \emph{Decline} if 
\[
i > (\maxID)^{\alpha_1},
\]
where $\alpha_1$ is the same constant as in the optimization problem. This will result in some nodes deciding to \emph{Decline} very quickly. However, we are able to show that this is not the case for too many nodes.
For this, we introduce Decline Gadgets. They are a small graph structure equipped with a node ID assignment, that make our algorithm output \emph{Decline} on one of their nodes. What makes Decline Gadgets special is that they are a sort of necessary condition to force our algorithm to output \emph{Decline}. That is, whenever a node in a real instance outputs \emph{Decline}, it is precisely because such a gadget is present.
Using the fact that the number of available IDs is limited, we can bound the maximum number of such gadgets. This allows us to prove that our algorithm makes sufficient progress and obtains the same runtime as the algorithm providing the polynomial upper bound.

\paragraph{Randomized algorithms.}
We now focus on the randomized setting, where nodes are given unlimited random bits.
The key idea here is that nodes can obtain a crude estimate of $\log n$ despite having no knowledge of $n$. Again, consider the setting in the left of \Cref{fig:lB}, where a level-1 node $v$ sees only its own level-1 path and needs to figure out when to decide to output \emph{decline}. We now apply the following strategy: Each node independently marks itself with probability $1/2$. Nodes then explore their level-$1$ path until either the path ends, or a marked node is encountered. Upon seeing a marked node, they immediately output \emph{Decline}.

This simple mechanism induces a probabilistic length threshold. Any level-$1$ path of length $\Omega(\log n)$ contains a marked node with high probability and therefore declines. Conversely, if a path does not encounter any marked node, then with high probability, its length is at most $O(\log n)$, thus giving the level-1 nodes some confidence to keep exploring their path. Therefore, randomization allows nodes to implicitly recover a notion of the $\log n$ scale without explicit knowledge of $n$.

The consequences for the level-2 nodes are roughly as follows. A level-$2$ path only needs to be fully $2$-colored if \emph{all} its attached level-$1$ paths decline. If the level $2$ path has at least polynomial in $n$ length, then there must also be polynomially many level-1 paths that have all decided to decline. W.h.p., this only happens if a majority of all those level-$1$ paths are $\Omega(\log n)$ nodes long. Thus, such a level-2 path can only have length at most $O(n/\log n)$ w.h.p., which is why this is the complexity of our algorithm.

For the lower bound, we observe that a node in the middle of a long level-$1$ path must decide to decline within constant time with at least constant probability (otherwise, we can force the algorithm to require $\Omega(n)$ rounds on some path graphs). Any level-$1$ path of length $\Theta(\log n)$ therefore has to decline w.h.p.\ and we can thus create a level-$2$ path of length $\Theta(n/\log n)$ that has to be $2$-colored.

When generalizing to $3$ levels, we treat level $1$ exactly as before. 
If we encounter long level-$2$ paths whose attached level-$1$ paths all decline, we obtain coarse information about $n$: namely, such level-$1$ paths must have length $\Omega(\log n)$. So when relating the lengths of these level-1 paths to the length of our level-2 paths, the question becomes: at what length should we decide to start declining level-2 paths as well?
From a high-level perspective, this again leads to an optimization problem analogous to the deterministic case, where a polynomial bound on $n$ was given.

However, the nature of this optimization problem changes significantly. In the polynomial setting, it suffices to optimize over functions of the form $n^{\alpha}$ for small constants $\alpha$. Here, in contrast, the information available is only at the scale of $\log n$, and it is no longer clear that restricting attention to polynomial functions yields the correct trade-off. Instead, one must optimize over a much broader class of functions.

In fact, solving the resulting optimization problem reveals that the correct threshold is given by a function $f(n)$ satisfying:
\[
f(f(n)) = \log n.
\]
This yields an algorithm with complexity $O(\frac{n}{f(n)})$.
Thus, unlike in the polynomial regime, here the optimal choice is characterized implicitly over the whole space of functions via a functional equation, reflecting the fundamentally different scaling behavior in the randomized setting.

Compared to the other settings, the lower bound is established using a similar, albeit significantly more technical argument.
Still, we are able to show that our optimization problem is the correct one, by proving an almost tight lower bound of $\Omega\Large(\frac{n}{f(n) \log n}\Large)$.




\section{Preliminaries}
We rely on classical graph theoretical notation, and refer to an undirected graph as $G=(V,E)$. In addition to edges, we also refer to halfedges, that is the set
\[
    \bar{E} := \set{(e,v) \mid e \in E \land v \in e},
\]
which essentially splits every edge into two parts. For our formal definition of LCLs, we will require labelings of the halfedges $\phi: \bar{E} \rightarrow \Sigma$, that assign every halfedge a label from some set $\Sigma$.\\
We denote by $B_r(u)$ the radius $r$ ball around a node $u \in V$, and sometimes we also refer to this as the $r$-hop neighborhood. \\
Throughout this work, we will restrict ourselves to the setting where the maximum degree $\Delta$ is a constant and where all graphs are trees.

\subsection{The LOCAL Model}
The LOCAL model is a model of distributed computing. This means that we are given a network represented as a graph $G = (V,E)$, where each node represents a computational unit, and edges represent communication links. Nodes are assumed to be entities with unlimited computational power and communication links are unbounded in the amount of information they can transmit. Time is measured in synchronous rounds of communication. Nodes all start at the same time (round 0), then in each round any node can send an arbitrary amount of information (e.g. its entire state) to all of its direct neighbors and then perform arbitrary computation. At the end of the computation, each node must produce an output, such that all outputs together solve the computational task. For example, each node is required to output a color, and the solution is globally correct if there does not exist a monochromatic edge.

Additionally, nodes may start the computation with some predetermined inputs in round 0. For example, if we restrict the nodes to only perform deterministic computation it is common to give each node a unique identifier from the set $[1, n^c]$, for some constant $c\ge 1$. This is then called the deterministic LOCAL model. Most often, nodes are also provided with knowledge about some of the graph parameters, like the number of nodes $n$ and the maximum degree $\Delta$. What exactly is given as an input to the nodes in round 0 varies between different works in the field. 

In \Cref{sec:lb-no-input} we assume that nodes are restricted to deterministic computation and are given unique natural numbers as IDs, nothing else, so in particular no upper bound on $n$ is known.

In \Cref{sec:NAlgo,sec:LBpolyPromise}, we assume nodes are restricted to deterministic computation and are given two natural numbers $N,c\ge1 $ as inputs, together with the promise that $n \le N \le n^c$. Additionally, nodes are given a unique ID from the set $[1, n^c]$.

In \Cref{sec:promiseIds}, we assume nodes are restricted to deterministic computation and are given only one natural number $c \ge 1$ as input. Again nodes are given unique identifiers from the set $[1,n^c]$.

In \Cref{sec:rand}, nodes are given access to an unbounded number of random coins. Nodes are not given any additional information, so they know nothing about $n$ and are not provided with any IDs.

\subsection{LCL Problems}
An LCL problem $\Pi = (\Sin, \Sout, r, \C)$ is a quadruple where 
\begin{itemize}
    \item $\Sin$ is a finite set of input labels
    \item $\Sout$ is a finite set of output labels
    \item $r \in \N$ is the checkability radius of $\Pi$
    \item $\C$ is a finite set of input-/output- labeled centered $r$-hop neighborhoods $C \in \C$. The labeling assigns every halfedge in $C$ an input label from $\Sin$ and an output label from $\Sout$.
\end{itemize}

A solution to $\Pi$ on a $\Sin$ labeled graph $(G, \phi_{in})$ is an output labeling $\sigma_{out}$ that assigns every halfedge in $G$ an output label from $\Sout$, such that for every node $v$, the input-/output- labeled $r$ hop ball around $v$ is isomorphic to a member of $\C$.

\paragraph{LCLs with labels on nodes.}
For our definition of the LCL $k$-hierarchical $2\frac{1}{2}$-coloring, we will refer to labels on nodes instead. Note that we can encode labels on nodes, by using labels on edges, by requiring all halfedges adjacent to some node $v$ to output the same label.

\subsection{\texorpdfstring{$k$}{k}-hierarchical \texorpdfstring{$2\frac{1}{2}$}{2½}-coloring} \label{sec:2.5Col}
Many of our results are about $k$-hierarchical $2\frac{1}{2}$-coloring. These problems are the first example of problems with complexities $\Theta(n^{1/k})$ \cite{CP19timeHierarchy}. From a high level view, $k$-hierarchical $2\frac{1}{2}$-coloring exactly captures what makes a problem with complexity $\Theta(n^{1/k})$ hard: we might have to solve a global problem in some long path. The value of $k$ determines how much freedom we have in choosing which path to solve this hard problem in.

The $k$-hierarchical $2\frac{1}{2}$-coloring problem is defined in the following way.
Given a tree $G=(V,E)$, we define $k$ sets $L_1, \ldots, L_{k}$, called respectively level $1, \ldots, k$.

We define the levels inductively on $1 \le i \le k$.
First, 
\[
L_1 = \{v\in V \mid \deg(v)\leq 2\}.
\]
Let $G^{(i)}$ be the subgraph induced by $V^{(i)} = V \setminus \bigcup_{1 \le j <i}L_j$. 
We then have  
\[
L_i = \{v\in V^{(i)} \mid \deg_{G^{(i)}}(v)\leq 2\}.
\]
We call $G^{(k+1)}$ the remainder.

Each node is either in a level $1 ,\ldots, k$, or in the remainder. Which set a node $v$ belongs to can be determined in $k$ rounds. There are no input labels, and the set of output labels is $\set{B,W, E,D}$, which stand for \emph{Black}, \emph{White}, \emph{Exempt} and \emph{Decline} respectively. The constraints are as follows:
\begin{enumerate}
    \item All nodes in the remainder output $D$.
    \item No node in level $k$ may output $D$ and no node in level 1 may output $E$.
    \item No node that outputs $W$ can be adjacent to a node of the same level that outputs $W$, or $D$. Similarly no node that outputs $B$ can be adjacent to a node of the same level that outputs $B$, or $D$.
    \item A node may only output $E$, if it is adjacent to a lower level node that outputs one of $\set{B,W,E}$.
\end{enumerate}
This concludes the problem description.
We make the following observations. 
\begin{itemize}
    \item Each level consists only of isolated nodes and paths.
    \item Each level-1 path is either properly 2-colored, using $W,B$, or all nodes output $D$.
    \item Each level $\ge 1$ path consists of continuous subpaths that are either properly 2-colored, or where all nodes output $D$. These subpaths must be separated by nodes that output $E$.
    \item If $G^{(k)}$ is non-empty, then at least one path must be properly 2-colored. If no level $<k$ node outputs $B$, or $W$, then no node is allowed to ever output $E$. Importantly, no node of level $k$ can output $E$. Since level $k$ nodes cannot output $D$, they are forced to output a consistent 2-coloring.
\end{itemize}

Essentially, the class of the $2\frac{1}{2}$-coloring problems are considered the canonical representatives of the complexity classes $\Theta(n^{1/k})$. They are first introduced in \cite{CP19timeHierarchy} as the first examples of LCLs that have the polynomial complexities $\Theta(n^{1/k})$. As a result, all of these complexity classes are non-empty.
\begin{lemma}[\cite{CP19timeHierarchy}]
For any $k \in \N$, $k$-hierarchical $2\frac{1}{2}$-coloring has complexity $\Theta(n^{1/k})$.
\end{lemma}

\subsection{Tree Decompositions and Complexity Classes}\label{sec:decAndLandscape}
When restricting the input graph to a tree, the distributed complexities that exist for LCL problems are very well understood.  This setting is therefore an excellent case study for our question of how giving different initial inputs to nodes changes the model.

To this end we give an overview about which polynomial complexity classes exist on trees in the LOCAL model, when nodes know a linear upper bound on $n$. At the heart of the study of these complexity classes are rake-and-compress decompositions. 

\begin{definition}[$(\gamma, \ell, L)$-decomposition \cite{CP19timeHierarchy}]\label{def:decomposition}
A $(\gamma, \ell, L)$-decomposition of a tree $\T$ is a decomposition of the nodes in $2L-1$ rake and compress layers. The compress layers $V^C_1, \ldots, V^C_{L-1}$ consist of paths and the rake layers $V^R_1, \ldots, V^R_L$ each consist of $\gamma$ sublayers $V^R_i = (V^R_{i,1}, \ldots, V^R_{i,\gamma} )$. The layers satisfy the following properties, based on the layer ordering $V^R_i < V^C_i < V^R_{i+1}$ and $V^R_{i,j} < V^R_{i,j+1}$ for all $1 \leq i < L-1, 1\leq j < \gamma$. 
\begin{enumerate}
	\item The components of each compress layer $V^C_i$ are isolated paths of length in $[\ell, 2\ell]$. Furthermore, the endpoints of each such path have exactly one neighbor in a higher layer. All other nodes have no neighbors in higher layers.
	\item The components of each rake sublayer $V^R_{i,j}$ are isolated nodes with at most one neighbor of a higher layer.
\end{enumerate}
\end{definition}

With the decomposition formally defined, we can give formal descriptions of the two operations. Assume we already have a (partial) $(\gamma, \ell, L)$-decomposition of some subtree $\T \subset G$. Let the remaining graph be $\bar{G} = G\setminus \T$. 

\paragraph{The Rake Operation:}
Only nodes that have degree $\le 1$  in $\bar{G}$ participate in a rake operation. If there are two degree-1 nodes connected by an edge, only one of them participates, chosen arbitrarily. Any node $v$ that participates computes the minimum $V^R_{i,j}$, such that $V^R_{i,j}$ is larger than the layer of all of $v$'s neighbors in $\T$. Then $v$ outputs $V^R_{i,j}$. Note that $1\le i\le L$ and $1 \le j \le \gamma$.\\

Inside of a compress procedure we will need to compute a $(\ell, 2\ell)$-ruling set\footnote{An $(a,b)$-ruling set is a subset $S$ of the nodes of the input graph such that the distance between any two nodes from $S$ is at least $a$ and for each node $u$ that is not contained in $S$, there is a node from $S$ that is in distance at most $b$ from $u$.}. To do this efficiently, we first precompute a distance-$\ell$ $O(1)$-coloring at the beginning of the algorithm. Since $\Delta \in O(1)$ this can be done deterministically by using e.g.\ Linial's coloring-reduction algorithm on $G^{\ell}$ in $O(\log^* C)$ rounds. Here, $C$ is the size of some initial coloring, which is typically given by the initial unique IDs assignment. 
Given such a coloring we can compute a $(\ell, 2\ell)$-ruling set in $O(1)$ rounds, by simply iterating through the colors. 
Having clarified this, we can describe the compress procedure.

\paragraph{The Compress Operation:}
The compress operation takes as input two values $\ell \in \N$ and $1 \le j\le L-1$.
Let $P$ be any maximal subpath of nodes of degree exactly 2 in $\bar{G}$, of length at least $\ell$. We first compute a $(\ell, 2\ell)$-ruling set on $P$. Then all ruling set nodes join layer $V^R_{j+1,1}$ and all of the remaining nodes join layer $V^C_j$.\\

Clearly a Rake operation can be performed in just one round of the LOCAL model. For $\ell \in O(1)$ and using the already discussed idea of precomputing a distance coloring, a Compress operation can be performed in $O(1)$ rounds, after an initial $O(\log^* \maxID)$ rounds of precomputation.

We then compute a $(\gamma, \ell, L)$-decomposition, by performing $\gamma$ rakes, followed by a Compress$(\ell, j)$ and then repeatng these two steps for $1 \le j \le L -1$. We then finish with another round of $\gamma$ rakes. Doing this we get the following results.
\begin{lemma}[\cite{CP19timeHierarchy}]\label{lem:polyDecomp}
For $\ell \in O(1)$ and any positive integer $k$, set $\gamma = n^{1/k}(\ell/2)^{1-1/k}$ then a $(\gamma, \ell , k)$-decomposition can be computed in $O(kn^{1/k})$ rounds.
\end{lemma}  

\begin{lemma}[\cite{CP19timeHierarchy}]\label{lem:logDecomp}
For $\ell \in O(1)$ and $\gamma \in O(1)$ then by setting $L \in O(\log n)$ a $(\gamma, \ell , L)$-decomposition can be computed in $O(\log n)$ rounds.
\end{lemma}  

However, \Cref{lem:polyDecomp} assumes that nodes know a linear upper bound on the number of nodes. For \Cref{lem:logDecomp} the only implicit assumption is that $\log^*(\maxID) \in O(\log n)$, where $\maxID$ denotes the largest ID.
We will soon see that by varying these assumptions we significantly change the complexity of computing a $(\gamma, \ell, L)$-decomposition.\\
This is significant, because $(\gamma, \ell, k)$-decompositions are in some sense $\Theta(n^{1/k})$-complete. This immediately implies that the complexity landscape of LCLs significantly depends on these assumptions. The completeness of these decompositions comes from the following result.

\begin{lemma}[\cite{chang20}]\label{lem:all-problems}
Let $k$ be a positive integer. Assume a $(\gamma, \ell , k)$-decomposition can be computed in $T_{dec}$. If an LCL $\Pi$ admits an $o(n^{1/(k-1)})$-round algorithm, then $\Pi$ can be solved in time $T_{dec} + k\gamma$. 
\end{lemma} 

Together with this result from \cite{CP19timeHierarchy} we completely settle the types of polynomial complexities that exist in the LOCAL model, where a linear bound on $n$ is given.
\begin{lemma}[\cite{CP19timeHierarchy}]\label{lem:all-problems2}
Let $k$ be a positive integer. Assume a $(\gamma, \ell , O(\log n))$-decomposition can be computed in $T_{dec}$. If an LCL $\Pi$ admits an $n^{o(1)}$ algorithm, then $\Pi$ can be solved in time $T_{dec} + \gamma \log n$. 
\end{lemma} 

We get as a corollary.
\begin{corollary}[\cite{CP19timeHierarchy,ChangHLPU20_treeLLL}]\label{cor:complexityClasses}
    Let $\Pi$ be an LCL on trees in the LOCAL model, where nodes know a linear upper bound on $n$, then one of the following is true. 
    \begin{itemize}
        \item There exists a $k \in \N$, such that $\Pi$ has complexity $\Theta(n^{1/k})$.
        \item $\Pi$ has complexity $O(\log n)$.
    \end{itemize}
\end{corollary}

Importantly, \Cref{cor:complexityClasses}, does not guarantee that these complexity classes actually exist. However, the complexity of $k$-hierarchical $2\frac{1}{2}$-coloring implies that such classes are non-empty. 

Given the importance of rake and compress decompositions, in our work we study the complexity of these problems. Unfortunately, the problem of computing a $(\gamma,\ell,k)$-decomposition cannot be expressed as an LCL problem. For this reason, we introduce the $k$-rake-and-compress family of LCL problems, which requires us to compute a decomposition where the value of $\gamma$ does not matter.
This idea was first introduced in \cite{fullversion-two}.

\begin{definition}[$k$-rake-and-compress]\label{def:decompLCL}
For any integer $k$, the output of the $k$-rake-and-compress problem is a partial orientation of the edges and an assignment of one label out of $\Sout = \{R_1, \dots, R_k, C_1, \dots, C_{k-1}\}$ to each node. The labels $R_1, \dots, R_{k}$ are called \emph{rake labels} and $C_1,\dots,C_{k-1}$ compress labels. Any legal labeling must satisfy the following rules, based on the ordering of the labels $R_1 < C_1 <R_2<C_2<R_3 <\dots<C_{k-1}< R_k$:
\begin{enumerate}
    \item All edges adjacent to at least one node labeled rake must be oriented, while the other edges must not be oriented. \label{rule2:orientRakes}
    \item Each node $v$ has at most one edge $e=(v,u)$ oriented outgoing, except for compress nodes that have two compress neighbors. Such compress nodes must not have any outgoing edge. \label{rule2:oneOutgoing}
    \item For all oriented edges $(u,v)$ the label of $v$ is larger than or equal to the label of $u$. \label{rule2:orderedOrientation}
    \item For all compress labels, the subgraph induced by the nodes of that label consists only of disjoint paths. \label{rule2:compressPaths}
    \item Two nodes that have a different compress label must not be adjacent.\label{rule2:pathsDisjoint}
\end{enumerate}
\end{definition}

Note that any $(\gamma, \ell, k)$-decomposition automatically also gives a solution to the $k$-rake-and-compress problem: simply have nodes output their respective rake or compress layer. However, the converse may not hold, since in a solution for $k$-rake-and-compress there is no bound on the diameter of the connected components induced by nodes belonging to the same rake layer.

We now provide some intuition on the rules used to define $k$-rake-and-compress.
\begin{itemize}
    \item Because of \cref{rule2:orientRakes,rule2:oneOutgoing} all components of nodes that output rake labels are consistently oriented.
    \item Because of \cref{rule2:compressPaths,rule2:pathsDisjoint} compress labels are only used in paths and two different compress paths are separated by at least one rake label.
    \item \Cref{rule2:orderedOrientation} means that we have to keep track of the number of compresses that we have already performed. To handle long paths we will simply let the nodes output a compress label and then have both endpoints pick a strictly larger rake label. Then we orient the edges connecting the endpoints towards the endpoints. Clearly this works only if we still have an available compress label.
\end{itemize}

\subsection{Background for the Unusual Complexities}
In computer science, the most common definition of the set $\Omega(f(n))$ is due to Knuth. However, there exists an alternative, but incompatible, definition due to Hardy and Littlewood. In \Cref{sec:LBpolyPromise} we prove a lower bound for the Hardy-Littlewood definition of $\Omega$, diverging from common notation. Surprisingly, we prove that this deviation to a strictly weaker version of this lower bound is necessary.

We have another case of unusual behavior that we encounter in this work. In \Cref{sec:rand}, we prove tight upper and lower bounds for a complexity that involves a function $f(n)$, such that $f(f(n)) = \log n$. Essentially, this function is the functional square root of the logarithm. In other words, it is obtained by taking half a logarithm. 

\subsubsection{Different Definitions of \texorpdfstring{$\Omega$}{Ω}}\label{subsec:omega_definition}
Two different definitions of $\Omega$ exist in the literature.
Let $f:\mathbb{N}\to\mathbb{N}$ and $g:\mathbb{N}\to\mathbb{N}$ be two functions.
The first is the Hardy-Littlewood definition:
\[
   f(n) \in \Omega(g(n)) \iff \exists k>0, \forall n_{0}, \exists n>n_{0}, |f(n)| \geq k\cdot |g(n)|.
\]
The second is the Knuth definition:
\[
    f(n) \in \Omega(g(n)) \iff \exists k>0, \exists n_{0}, \forall n>n_{0}, |f(n)| \geq k\cdot |g(n)|.
\]

The Hardy-Littlewood version indicates that $f$ is larger than $g$ infinitely often. On the other hand, Knuth's version is stronger and indicates that, after some starting point $n_0$, $f$ is \emph{always} bounded below by $g$.
For instance, for $f(n)=n^{(-1)^n}$ and $g(n)=n$, we have $f(n)=\Omega(g(n))$ according to Hardy-Littlewood's definition, but not according to Knuth's definition.

This difference will be key in some of our results, and we seem to have stumbled upon one of the rare cases where we cannot just rely on Knuth's definition.
Indeed, \Cref{thm:no-kuth-LB} proves that without using the Hardy-Littlewood definition, we cannot improve the lower bounds of any LCL with complexity $\Theta(n^{1/{k}})$ in the standard LOCAL model.

The only part in which the Hardy-Littlewood definition is used in this work is in \Cref{sec:LBpolyPromise}, in which we are very explicit about what definition we use. In all other parts of this work, we use the Knuth's definition.

\subsubsection{Taking Half a Logarithm}\label{subsec:halflog}
In \Cref{sec:rand}, where we study the randomized complexity of $2\frac{1}{2}$-coloring in the case where nothing is known about $n$, the complexity of the $3$-hierarchical $2\frac{1}{2}$-coloring problem turns out to not be easily expressible by analytic functions or generally by functions that we usually use to analyze the complexity of algorithms. The complexity of $3$-hierarchical $2\frac{1}{2}$-coloring turns out to essentially be of the order $n/f(n)$, where $f(x)$ is a function for which $f(f(x))=\ln x$. Such a function $f(x)$ is also known as a \emph{functional square root} of $\log x$~\cite{sqrtoflog}. The inverse function of $f(x)$ is a so-called half-exponential function $h(x)$ for which $h(h(x))=\exp(x)$. It was shown in \cite{halfexponential}, that there are infinitely many functions $h(x)$ that satisfy this identify and that are strictly monotonically increasing, continuous and continuously differentiable. The same is therefore also true for functions $f(x)$ satisfying $f(f(x))=\ln x$. In \cite{sqrtoflog}, it is shown that if $f(x)$ is restricted to a family of functions that can be expressed by some specific power series, then $f(x)$ is indeed unique. To the best of our knowledge such half-logarithmic and half-exponential functions have not been widely studied. They have however in particular appeared in the context of circuit complexity~\cite{halfexp_in_complexity}.

In the context of this paper, we assume for every $a>1$ a half-logarithmic function $f_a(x):\mathbb{R}\to\mathbb{R}$ that satisfies $f_a(f_a(x))=\log_a(x)$ (for $x\geq a$) and that has the following properties.
\begin{itemize}
    \item $f_a(x)$ is continuous and strictly monotonically increasing in $x$
    \item For all $1<a<b$ and all sufficiently large $x$, we have $f_b(x)< f_a(x)$.
\end{itemize}
Such a family of functions $f_a(x)$ might not be efficiently or explicitly constructible, but it definitely exists. Under those assumptions, we can show that $f_a(x)$ satisfies some useful additional properties. One can in particular show that for any fixed $1<a<b$, we have $f_a(x)=\Theta(f_b(x))$. To see this, let $c:=\log_a b$ and define $g_a(x)$ as $g_a(x):=(1/c)\cdot f_a(x)$. We have
\[
g_a(f_a(x)) \leq \frac{f_a(f_a(x))}{c} = \frac{\log_a x}{c} = \log_b x.
\]
Since $f_b(x)<f_a(x)$, we have $\log_b x = f_b(f_b(x)) < f_b(f_a(x))$ and thus $g_a(x) = O(f_b(x))$. As the base of the logarithm does not affect the asymptotic behavior of $f_a(x)$, we will use $f(x)$ instead of $f_a(x)$, wherever this is appropriate.

While we cannot specify the asymptotic behavior of the half-logarithmic function $f(x)$, we can sandwich $f(x)$ between some easily representable functions. For any constant $\epsilon>0$, one can for example verify that for $h(x)=\exp(\ln^\epsilon x)$, we have $h(h(x))=\exp(\ln^{\epsilon^2}x)=\omega(\log x)$. Moreover, if we use the notation $g^{(k)}(x)$ to denote the $k$-fold application of a function $g(x)$, we can define the following family of functions $h_{k,\epsilon}$ for every integer $k\geq 1$ and $\epsilon>0$:
\[
h_{k,\epsilon}(x) := \exp^{(k)}\left(\big[\ln^{(k)} x\big]^\epsilon\right).
\]
For all fixed $k$ and $\epsilon$, one can verify that $h_{k,\epsilon}^{(2)}(x)=\omega(\log n)$ and therefore $h_{k,\epsilon}(x)=\omega(f(x))$ for the half-logarithm function $f(x)$.
Note that for every $k\geq 1$, for every $\epsilon_1\in(0,1)$ and every $\epsilon_2>0$, we have $h_{k+1,\epsilon_1}(x)=o(h_{k,\epsilon_2}(x))$ even if $\epsilon_2$ is much smaller than $\epsilon_1$. For integers $k\geq 1$ and constant $c>0$, one can similarly build a family of functions $\ell_{k,c}$ as
\[
\ell_{k,c}(x) := \exp^{(k)}\left(\big[\ln^{(k+1)} x\big]^c\right).
\]
Here, one can verify that for every $k\geq 1$ and every $c>0$, $\ell_{k,c}^{(2)}(x)=o(\log x)$ and thus $\ell_{k,c}(x)=o(f(x))$. In this case, as long as $c>1$, for every $k$, we have $\ell_{k+1,c}(x)=\omega(\ell_{k,c'}(x))$ even if the constant $c'$ is chosen much larger than $c$.

\section{LOCAL Without Extra Knowledge}\label{sec:lb-no-input}
To answer our question about how the LOCAL model is influenced by the initial knowledge provided to the nodes, we first investigate the most restrictive setting. In this section nodes are restricted to deterministic computation and are given unique IDs that are arbitrary natural numbers. Nodes are not provided with any other input, so they do not know anything about $n$ and the IDs can be arbitrarily large. Note that without IDs we run into some trivial impossibilities and we do not want to restrict the model to the extent where problems become trivially unsolvable.

Even though we kept our model strong enough to be able to still solve all considered problems, we will prove that this model is nevertheless too restrictive. We first prove that $k$-hierarchical $2\frac{1}{2}$-coloring requires $\Omega(n)$ rounds. Note that in the LOCAL model $O(n)$ rounds is a trivial upper bound for any solvable problem, since in $O(n)$ rounds every node can learn the entire network with all inputs and brute force a solution through local computation.

We prove lower bounds for the two most extensively studied families of problems with polynomial complexities, namely $k$-hierarchical $2 \frac{1}{2}$-coloring and $k$-rake-and-compress.
We start by defining a family of graphs, that we call \emph{$k$-hierarchical lower bound graphs}.
\begin{definition}[$k$-hierarchical lower bound graph]\label{def:LowerBoundGraph}
    Let $\ell_1,\ldots,\ell_k$ be positive integers.
    A $k$-hierarchical lower bound graph with parameters $\ell_1,\ldots,\ell_k$ is defined recursively as follows.

    Start from a path $P_k$ of $\ell_k$ nodes. The path $P_k$ is called \emph{path of level $k$}, and its nodes are called \emph{nodes of level $k$}. Let $T_k = P_k$.
    Then, recursively, for $i = k-1,\ldots,1$ do the following. For each path $P_{i+1}$ of level $i+1$, for each node $v$ of $P_{i+1}$:
    \begin{itemize}
        \item if $v$ has degree $2$ in $T_{i+1}$, create a path $P_i$ of $\ell_{i}$ nodes, and connect one endpoint of $P_i$ to $v$;
        \item if $v$ has degree $1$ in $T_{i+1}$, create two paths $P_i$, both of $\ell_{i}$ nodes, and for each $P_i$ connect one endpoint of $P_i$ to $v$.
    \end{itemize}
    The paths $P_i$ are \emph{paths of level $i$} and their nodes are \emph{nodes of level $i$}. Let $T_i$ be the obtained graph.
\end{definition}

\begin{observation}\label{obs:nodes-k-hierarchical}
Let $L_i = \prod_{1 \le j \le i} \ell_j$.
For some constants $c_1 \ge 1$ and $c_2$ that depend solely on $k$, in a $k$-hierarchical lower bound graph with parameters $\ell_1,\ldots,\ell_k$, each component of the subgraph induced by nodes of level at most $i$ contains at least $c_1 \cdot L_i$ and at most $c_2 \cdot L_i$ nodes.
\end{observation}

\begin{observation}\label{obs:one-2-colored-path}
    Let $G$ be a $k$-hierarchical lower bound graph. Then, any solution for $k$-hierarchical $2\frac{1}{2}$-coloring on $G$ must satisfy that there exists an index $i \in \{1,\ldots,k\}$ and a path $P_i$ of level $i$ that is properly $2$-colored.
\end{observation}
\begin{proof}    
    By induction on $i$, suppose that all nodes on all paths $P_j$ for $j < i$ are labeled $D$. By the definition of $k$-hierarchical $2\frac{1}{2}$-coloring, either there exists a path $P_i$ that is properly $2$-colored, or all nodes on all paths $P_i$ are also labeled $D$. 

    Since the definition of $k$-hierarchical $2\frac{1}{2}$-coloring does not allow nodes of $P_k$ to be labeled $D$, we get that, if there is no path $P_j$ for $j < k$ that is properly $2$-colored, then $P_k$ is properly $2$-colored.
\end{proof}

With this we can prove our first result.
\HalfColNoBounds*
\begin{proof}
Let $\mathcal{A}$ be an algorithm that solves $k$-hierarchical $2\frac{1}{2}$-coloring. Assume for a contradiction that for all $\epsilon > 0$, for all integers $n_0$, there exists some $n = f(\epsilon,n_0) > n_0$, such that the algorithm $\mathcal{A}$ terminates in at most $\epsilon n$ rounds on all $n$-node graphs. 

In the following, let $c_1$ and $c_2$ be the constants guaranteed to exist by \Cref{obs:nodes-k-hierarchical}.
Let $L_0 = 1$ and $L_i = \prod_{1 \le j \le i}\ell_j$. 
Let $N_{i+1} = f(\frac{1}{20 c_2 L_i}, 100 c_2 L_i )$
and $\ell_{i+1} = \lfloor N_{i+1} / (c_2 L_i) \rfloor$ for all $i \in \{0,\ldots,k\}$.  Observe that $\ell_i \ge 100$ for all $i \in \{1,\ldots,k\}$. By assumption, in instances of size $N_i$ the algorithm must terminate in at most $N_i / (20 c_2 L_{i-1})$.

Let $G$ be the $k$-hierarchical lower bound graph with parameters $\ell_1,\ldots,\ell_k$, labeled with an arbitrary ID assignment. Let $T_i = \ell_i / 10$. We prove by induction on $i$ that, by running $\mathcal{A}$ on $G$, for each path $P_i$ of level $i$ it holds that there are two nodes $u,v \in P_i$ satisfying the following properties:
\begin{enumerate}
    \item both $u$ and $v$ are at distance at least $T_i + 1$ from the endpoints of $P_i$;
    \item $u$ and $v$ are at distance at least $2 T_i + 1$ from each other;
    \item both $u$ and $v$ terminate in at most $T_i$ rounds;
    \item both $u$ and $v$ output $D$;
\end{enumerate}
The above properties, combined with the definition of $k$-hierarchical $2\frac{1}{2}$-coloring, imply that the whole graph $G$ is labeled $D$, which contradicts \Cref{obs:one-2-colored-path} and hence proves the theorem.

The base case $i=0$ trivially holds (by the definition of the $k$-hierarchical lower bound graph, paths $P_0$ do not exist).
Let $P_i$ be an arbitrary path of level $i$. Let $u$ and $v$ be the two nodes at distance exactly $T_i + 1$ from (at least) one of the endpoints of $P_i$. 
Nodes $u$ and $v$ clearly satisfy property 1. Since $P_i$ has length $\ell_i$, since $\ell_i \ge 100$, and since $u,v$ are at distance $\ell_i/10+1$ from the endpoints, property 2 is also clearly satisfied.
Assume for a contradiction that property 3 does not hold, and hence that $u$ or $v$ runs for strictly more than $T_i$ rounds. W.l.o.g., let this node be $v$. We consider the graph $G'$ obtained by taking the connected component containing $P_i$ in the subgraph of $G$ induced by all nodes of level at most $i$. By \Cref{obs:nodes-k-hierarchical}, $G'$ has at most $c_2 \cdot L_i$ nodes. Since $\ell_{i} \le N_{i} / (c_2 L_{i-1})$, we get that $G'$ has at most $N_i$ nodes. We modify $G'$ by taking an arbitrary path $P_1$ at distance strictly larger than $T_i$ from $v$ and adding nodes in order to get a graph $G''$ of exactly $N_i$ nodes. By assumption, the algorithm $\mathcal{A}$, when run on $G''$, must terminate in at most $N_i / (20 c_2 L_{i-1})$ rounds. Since  $\ell_{i} \ge  N_{i} / (c_2 L_{i-1}) - 1 \ge N_{i} / (2 c_2 L_{i-1})$, we get that $N_{i} \le 2 c_2 L_{i-1} \ell_{i}$. Hence, the algorithm must terminate in at most $N_i / (20 c_2 L_{i-1}) \le 2 c_2 L_{i-1} \ell_{i} / (20 c_2 L_{i-1}) = \ell_i / 10 = T_i$ rounds. Since within distance $T_i$ the view of node $v$ is the same on $G$ and $G''$, and no additional input is provided to $v$, then the runtime of $v$ must be at most $T_i$ also on $G$, contradicting the assumption that on $G$ node $v$ runs for strictly more than $T_i$ rounds, and hence proving property 3.
We now prove that property 4 holds. Suppose for a contradiction that $u$ or $v$ output a label different from $D$. By the definition of the problem, since all nodes of lower layers are outputting $D$, it cannot be $E$, and hence it must be a color. By the constraints of the problem we get that $P_i$ must be properly $2$-colored. Let $z$ be a node of $P_i$ that is at distance strictly larger than $T_i$ from both $u$ and $v$ and that lies on the path connecting $u$ to $v$ (which, by properties 1 and 2, must exist). Let $z_1$ and $z_2$ be the two neighbors of $z$ in $P_i$. Consider the graph $\hat{G}$ obtained by removing from $G$ the node $z$ and all the nodes of lower layers reachable from it by not passing from $z_1$ nor $z_2$, and then adding an edge connecting $z_1$ and $z_2$. Observe that the view of $u$ and $v$ in $\hat{G}$ is the same as their view on $G$, and hence they run for the same time and they produce the same output. However, in $\hat{G}$, the parity of the distance between $u$ and $v$ is different from the parity in $G$. Thus, the produced $2$-coloring cannot be proper, reaching a contradiction and hence proving property $4$.
\end{proof}

In the following, by \emph{rake node of level $i$}  (resp.\ compress node of level $i$) we denote a node with output label $R_i$ (resp.\ $C_i$). By \emph{type} of a node we denote its output label. 
Before proving a result about $k$-rake-and-compress, we observe some useful properties about rake-and-compress decompositions.
\begin{lemma}\label{obs:rc-path}
    Let $v$ be a rake node of level $i$. Let $t(u)$ be the type of node $u$ in the $k$-rake-and-compress decomposition.
    Then, there must exist a path $v_1,\ldots, v_l$ satisfying the following:
    \begin{itemize}
        \item $v_1$ is a node of degree $\le 2$;
        \item $v_l = v$;
        \item $t(v_j) \le t(v_{j+1})$ for all $j$, where the ordering is the one from \Cref{def:decompLCL};
        \item if $v_j$ is a rake node, it points to $v_{j+1}$.
    \end{itemize}
\end{lemma}
\begin{proof}
    We prove that each node $u$ of degree $\ge 3$ must have at least one incoming edge. Since the graph is acyclic, and edges cannot be oriented from a layer to a strictly lower one, by recursively applying this idea starting from $v$, we obtain the lemma.

    If $u$ is a rake node of degree $\ge 3$, since all edges incident to $u$ are oriented and at most one is outgoing, then the claim follows.

    If $u$ is a compress node of degree $\ge 3$, since at most $2$ edges incident to $u$ can be unoriented (because incident to other compress nodes), and all the other cannot be outgoing, then $u$ has at least one incoming edge.
\end{proof}

By the fact that a rake node cannot have two outgoing edges, and the fact that an edge cannot be oriented from a larger level to a lower level, we observe the following.
\begin{observation}\label{obs:no-2-outgoing}
    Let $G$ be a graph. Any $k$-rake-and-compress decomposition of $G$ must satisfy the following. Let $P$ be a path that is a subgraph of $G$. Assume that all nodes of $P$ are rake nodes (of possibly different levels). Then, there exists one node $v$ in $P$ such that all edges of $P$ are oriented towards $v$. Moreover, for all nodes $u \in P$ it must hold that $t(u) \le t(v)$.
\end{observation}

This lets us prove our lower bound for computing a rake and compress decomposition.
\RcNoBounds*
\begin{proof}
The proof of this theorem is similar to the proof of \Cref{thm:2half-no-bounds}. 
    Let $\mathcal{A}$ be an algorithm that solves $k$-rake-and-compress.  Assume for a contradiction that for all $\epsilon > 0$, for all integers $n_0$, there exists some $n = f(\epsilon,n_0) > n_0$, such that the algorithm $\mathcal{A}$ terminates in at most $\epsilon n$ rounds on all $n$-node graphs. 
    Let $L_i$, $\ell_i$, and $T_i$ be defined as in the proof of \Cref{thm:2half-no-bounds}. 
    Let $G$ be the $k$-hierarchical lower bound graph with parameters $\ell_1,\ldots,\ell_k$, labeled with an arbitrary ID assignment. We prove by induction on $i$ that, by running $\mathcal{A}$ on $G$, for each path $P_i$ of level $i$ the following holds.
\begin{enumerate}
    \item There is at least one node with a compress label of level $\ge i$.
    \item If $1 \le i < k$, let $v$ be the endpoint of $P_i$ that is connected to a node of $P_{i+1}$. Then $v$ is a rake node of level $\ge i+1$, or a compress node of level $\ge i$.
\end{enumerate}
Observe that property 1 would imply that $P_k$ contains at least one node labeled compress of level $k$, but by the definition of the problem, in a $k$-rake-and-compress decomposition, there are no compress nodes of level $k$, reaching a contradiction. 

The base case $i=0$ of property $1$ and $2$ trivially holds (by the definition of the $k$-hierarchical lower bound graph, paths $P_0$ do not exist).
We start by proving that property 1 holds for $i \ge 1$, assuming that property 2 holds for $i-1$.
Suppose for a contradiction, that there is no node labeled compress of level $\ge i$. 
By property $2$, all nodes connected to $P_i$ from lower layers are rake nodes of layer $\ge i$ or compress nodes of layer $\ge i-1$.
Since the only nodes of degree $\le 2$ are nodes of some $P_1$, we get that:
\begin{itemize}
    \item Nodes of $P_i$ cannot be compress nodes of layer $i-1$;
    \item By \Cref{obs:rc-path}, nodes of $P_i$ are rake nodes of level $\ge i$ or compress nodes of level $\ge i$.
\end{itemize}
Note that the above statement is trivial for $i=1$.
By \Cref{obs:no-2-outgoing}, there must exist exactly one node $v$ of $P_i$ satisfying that all edges of $P_i$ are oriented towards $v$. Similarly as in the proof of \Cref{thm:2half-no-bounds}, we can modify $G$ and create a different graph $\hat{G}$ such that we can pick some nodes in $G$ that: they run in $T_i$ rounds; they have the same view in $G$ and in $\hat{G}$ and hence they produce the same output in $G$ and in $\hat{G}$; their output cannot be completed into a valid solution in $\hat{G}$. This contradicts the correctness of $\mathcal{A}$. 

We now prove that, assuming property 1 holds for $i$ and property 2 holds for $i-1$,  property 2 holds also for $i$. Consider the subgraph $G'$ of $G$ induced by all nodes belonging to paths of level strictly larger than $i$. 
By property $2$, all nodes connected to $P_i$ from lower layers are either rake nodes of layer $\ge i$ or compress nodes of layer $ \ge i-1$. Since the only nodes of degree $\le 2$ are nodes of some $P_1$, we get that:
\begin{itemize}
    \item Nodes of $P_i$ cannot be compress nodes of layer $i-1$;
    \item By \Cref{obs:rc-path}, nodes of $P_i$ are rake nodes of level $\ge i$ or compress nodes of level $\ge i$.
\end{itemize}
Note that the above statement is trivial for $i=1$.
Moreover, by the definition of $k$-hierarchical lower bound graph, all nodes of $G'$ are nodes of degree $\ge 3$ in $G$. Let $v$ be an arbitrary node in $V(G) \setminus V(G')$ that is connected to a node $u \in V(G')$. We prove that, either:
\begin{itemize}
    \item node $v$ is a rake node of level $\ge i+1$ or a compress node of level $\ge i$, or
    \item node $v$ is a rake node of level $i$ and the edge $\{u,v\}$ is oriented towards $v$. (We will later prove that this case cannot apply.)
\end{itemize}
Suppose node $v$ is not a rake node of level $\ge i+1$ nor a compress node of level $\ge i$. This implies that $v$ is a rake node of level exactly $i$. Let $P_i$ be the path containing $v$, and let $z$ be the compress node of layer $i$ belonging to $P_i$ that is the nearest to $v$, which, by property $1$ must exist. Let $w$ be the neighbor of $z$ that is nearer to $v$. Let $P$ be the subpath of $P_i$ that starts at $v$ and ends at $w$. If $w$ is a rake node of level $\ge i +1$, by \Cref{obs:no-2-outgoing} we get that the path $P$ is not rooted at $v$. Since $v$ cannot have two outgoing edges, we get that the edge $\{u,v\}$ is oriented towards $v$. Otherwise, if $w$ is a rake node of level exactly $i$, then the edge $\{w,z\}$ is oriented towards $z$. By \Cref{obs:no-2-outgoing}, the path $P$ is rooted on $w$. Again, since $v$ cannot have two outgoing edges, we get that the edge $\{u,v\}$ is oriented towards $v$.
Let us summarize what we have observed so far.
\begin{itemize}
    \item Since all nodes in $G'$ have degree $\ge 3$, there are no leaves in $G'$.
    \item Let $S$ be the set of nodes $v \in V(G) \setminus V(G')$ that have a neighbor $u \in V(G')$. Every $v \in S$ is an endpoint of some $P_i$ and is a rake node of level $\ge i$ or a compress node of level $\ge i$. If $v$ is a rake node of level exactly $i$, it must hold that the edge $\{u,v\}$ is oriented towards $v$. 
\end{itemize}
We prove that each node $v \in S$ is either a rake node of level $\ge i+1$ or a compress node of level $\ge i$, establishing property $2$.
Suppose for a contradiction that there exists a node $v \in S$ that is a rake node of level exactly $i$. By \Cref{obs:rc-path}, and the fact that $v$ cannot have two outgoing edges, there must exist a path $v_1,\ldots, v_l$ that starts from a node of degree $\le 2$, contains at least one node of $G'$, and ends at $v$, such that $t(v_j) \le t(v_{j+1})$ for all $j$ and such that if $v_j$ is a rake node, then it points to $v_{j+1}$. However, by the above observations, in order for such a path to start from a node of degree $\le 2$ (and hence a node of some $P_1$) and reach a node of $G'$, it must contain some rake node of level $\ge i+1$ or some compress node of level $\ge i$. Since $v$ is a rake node of level $i$, this contradicts the ordering on the path.
\end{proof}

\section{Polynomial Upper Bound \texorpdfstring{$N$}{N} Given}\label{sec:ub-polyn-given}
In this section, we consider the model in which nodes are provided with some integer $N$ and the promise that $n \leq N \leq n^c$ for some positive integer $c$. Nodes do not know the value of $n$, but they know $c$.

For some value of $\gamma$ to be specified later, we provide an upper bound of $O(\gamma)$ for computing a $(\gamma,\ell,k)$-decomposition. By \Cref{lem:all-problems}, such an upper bound applies also to all problems that have complexity $o(n^\frac{1}{k-1})$ in the standard LOCAL model.
The value of $\gamma$ will depend on the result of a nontrivial optimization problem. Hence, we now provide an informal statement, and we will provide the precise statement in \Cref{cor:optSolution}.

\begin{lemma}\label{lem:}
    For any positive integer $k$, there exists some constant $\alpha < \frac{1}{c}$, such that a $(O(n^{c\alpha}),\ell,k)$-decomposition can be computed in $O(n^{c\alpha})$ rounds. 
\end{lemma}

Our algorithm starts by doing $N^\alpha$ rakes and then a compress. The main idea is that we may do a lot of work upfront if our bound $N$ is bad. That is if $N$ is close to $n^c$, we might already spend a lot of time doing this first set of rakes. However, this also means that we make a lot of progress and so we might be able to be more aggressive afterwards. We push this idea to the limit, by starting with a fully parameterized algorithm and choosing the parameters based on an optimization problem.

\subsection{The Decomposition Algorithm}\label{sec:NAlgo}
Fix positive integers $k$ and $c$. The algorithm is parameterized by some parameters $\alpha_1, \ldots, \alpha_{k-1}$ that can be computed solely as a function of $k$ and $c$, and computes a $(\gamma,\ell,k)$-decomposition for $\gamma \in O(n^{c \alpha_1})$. Recall that an integer $N$ satisfying $n \le N \le n^c$ is provided to the nodes.
Recall that a rake operation is the removal of all nodes of degree $0$ or $1$, and that a compress operation consists of the removal of all connected components containing at least $\ell$ nodes and consisting of nodes of degree exactly 2.
The algorithm consists of the following $k$ phases. 
\begin{itemize}
    \item \textbf{Phase $1 \leq i < k$:} Perform $N^{\alpha_i}$ rakes and then a single compress.
    \item \textbf{Phase $k$:} Perform rakes until the graph becomes empty.
\end{itemize}

In order to analyze this algorithm, we use the following lemma, which has been proven in \cite{chang20}. This lemma upper bounds the number of nodes that remain after performing some number of rake operations followed by a single compress operation.
\begin{lemma}[\cite{chang20}]\label{lem:rcremaining}
Given a forest with $n$ nodes, by performing $x$ rakes and $1$ compress with minimum path length $\ell$, the number of remaining nodes is at most $\frac{\ell}{2 x} n$.
\end{lemma}

As a simple corollary of \Cref{lem:rcremaining}, we get bounds on the number of nodes that are still present after phase $i$.
\begin{corollary}\label{cor:remainAfterPhase}
    For any $1\leq i <k$, let $A_i := \sum_{j= 1}^i \alpha_j$. Then, after phase $i$ of the execution of the algorithm, at most $O\left(\frac{n}{N^{A_i}}\right)$ nodes remain.
\end{corollary}
By using \Cref{cor:remainAfterPhase},  we can derive an upper bound on the runtime of a given phase.
\begin{lemma}\label{lem:runtimeOfPhases}
    For any $1\leq i <k$, phase $i$ of the algorithm takes at most $O\left(\min\set{\frac{n}{N^{A_{i-1}}}, N^{\alpha_i}}\right)$ rounds.
    Furthermore, phase $k$ takes at most $O\left(\frac{n}{N^{A_{k-1}}}\right)$ rounds.
\end{lemma}
\begin{proof}
    By \Cref{cor:remainAfterPhase}, at the beginning of phase $2 \leq i \leq k$, only $O\left(\frac{n}{N^{A_{i-1}}}\right)$ nodes remain. Also, every phase, except phase $k$, can run for at most $O(N^{\alpha_i})$ rounds, because each rake operation, and the compress operation, can be performed in a constant number of rounds.
\end{proof}
The runtime of our algorithm depends on the parameters $\alpha_1,\ldots,\alpha_{k-1}$. In order to determine the correct choice for these parameters, we introduce and analyze an optimization problem in the next section.

\subsection{The Optimization Problem}\label{ssec:optimization-problem}
\Cref{lem:runtimeOfPhases} gives rise to the following optimization problem, where we want to optimize the parameters $\alpha_1, \ldots, \alpha_{k-1}$ to make the overall runtime as small as possible.

\paragraph{Parameters.}
The following parameters are constants of the problem instance:
\begin{itemize}
    \item The number of nodes $n$. This parameter is not known by the nodes, and hence the values of $\alpha_i$ cannot be computed as a function of it.\item The integer $k$. This parameter depends on the problem that the nodes need to solve, and hence it is known by the nodes. 
\end{itemize}
The following parameters are given adversarially:
\begin{itemize}
    \item The exponent $c$ and the integer $N$. These parameters are known by the nodes, and it is guaranteed that $n \le N \le n^c$.
 \end{itemize}   
We need to determine the values of the following parameters: 
\begin{itemize}
    \item The integers $\alpha_1,\ldots, \alpha_{k-1}$. Recall that these parameters govern the number of rakes performed during each phase. That is, at phase $i$, nodes perform $N^{\alpha_i}$ rakes.
\end{itemize}
In the following, recall that $A_i$ is defined as $A_i := \sum_{j= 1}^i \alpha_j$. Moreover, by $\bar{\alpha}$ we denote the vector $[\alpha_1,\ldots,\alpha_{k-1}]$.

\paragraph{Optimization Problem.}
Since the number of phases of our algorithm is $k$, and since $k$ is constant, the runtime of our algorithm is given by the runtime of its slowest phase. Hence, for $1 \le i \le k$, let $T_i$ be the runtime of the $i$th phase, as a function of all the other parameters. The goal is to minimize the largest $T_i$, for the worst-case choice of $n$ and $N$. Hence, we get that the overall runtime of the algorithm is bounded by the optimal solution to the following optimization problem.
\begin{align*}
\min_{\bar{\alpha}}\max_{N \in [n,n^c]}\left\{\begin{array}{lr}
    T_1 := N^{\alpha_1}\\
    T_2 := \min \left\{\frac{n}{N^{\alpha_1}}, N^{\alpha_2}\right\}\\
    \ldots \\
    T_i := \min \left\{\frac{n}{N^{A_{i-1}}}, N^{\alpha_i}\right\}\\
    \ldots\\
    T_k := n^{1-A_{k-1}}
    \end{array}\right\}
\end{align*}

Just from the statement alone we can immediately see that we get an algorithm with complexity at least polynomially faster than $O(n)$.

\begin{corollary}\label{cor:betterThanGlobal}
    Let $\bar{\alpha}$ be an optimal solution to the optimization problem then for $1 \leq i\leq k-1$ it holds that $\alpha_i < \frac{1}{c}$ and also all $T_i \in O(n^\varepsilon)$ for some $\varepsilon \leq \frac{c}{c+1}$. 
\end{corollary}
\begin{proof}
    We simply give a valid assignment to $\bar{\alpha}$, that is not optimal but still gives us a (rough) upperbound on the optimal solution.\\
    Let $\frac{1}{c+1} = \alpha_1 = \alpha_2 = \ldots = \alpha_{k-1}$. Then all of the $T_i$, except $T_k$ are bounded by 
    \[
    N^{\alpha_i} \leq n^{c \cdot \frac{1}{c+1}} = n^{\frac{c}{c+1}}
    \]
    For $T_k$ we have 
    \[
    T_k =  n^{1-A_{k-1}} =  n^{1-(k-1) \frac{1}{c+1}} <  n^{1-\frac{1}{c+1}} = n^{\frac{c}{c+1}}
    \]
    Since any optimal solution must be at least as good as this one, the statement follows.
\end{proof}

The rest of this section is devoted to finding a precise solution to the optimization problem. We try to reduce the complexity of the problem one step at a time, by first eliminating the minimization terms inside of the $T_i$ terms.

\begin{lemma}\label{lem:iZero}
 Let $i_0 := \left\lfloor\frac{1}{c \alpha_1}\right\rfloor$. Then, the optimal values of $\alpha_i$ satisfy $\alpha_1 =  \ldots = \alpha_{i_0}$. Moreover, as a result, it holds that $A_i = i \cdot \alpha_1$ for all $1\leq i \leq i_0$.
\end{lemma}
\begin{proof}
If we fix $\bar{\alpha}$, then the runtime $T_i$ of phase $i$ is maximized if $N=n^{b_i}$ for some worst case $1\leq b_i \leq c$. We solve for $b_i$, by setting the terms in the minimization equal to each other.
\begin{align*}
    T_i &:= \max_{n \leq N \leq n^c} \min \left\{\frac{n}{N^{A_{i-1}}}, N^{\alpha_i}\right\} = \max_{1 \leq b_i \leq c} \min \left\{\frac{n}{n^{b_i \cdot A_{i-1}}}, n^{b_i \cdot \alpha_i}\right\} \\
    &\iff \frac{n}{n^{b_i \cdot A_{i-1}}} = n^{b_i \cdot \alpha_i}\\
    &\iff 1- b_i \cdot A_{i-1} = b_i \cdot \alpha_i\\
    &\iff b_i = \frac{1}{A_{i-1} + \alpha_i} = \frac{1}{A_i}
\end{align*}
However, $b_i$ cannot be larger than $c$ and since all $\alpha_i$ are strictly positive, the $A_i$ terms are strictly increasing. As a result we might have to set $b_i = c < 1/A_i$ which leads to $\frac{n}{n^{b_i A_{i-1}}}> n^{b_i \alpha_i}$. So because $T_i$ is the minimum of these two terms we obtain the following:
\begin{align*}
    b_i &= \min \left\{ c, \frac{1}{A_i}\right\}\\
    T_i &= \begin{cases}
        n^{b_i\alpha_i} & \text{if } b_i = c\\
        \frac{n}{n^{b_i A_{i-1}}} = n^{b_i\alpha_i}, & \text{if } b_i = \frac{1}{A_i}\\
        \end{cases}
\end{align*}
In the former case, i.e., when $b_i = c$, the runtime $T_i$ becomes $n^{c\alpha_i}$. Hence, we get that for all $i$ such that $b_i = c$, the runtime $T_i$ is $n^{c \alpha_i}$ and hence, for all such $i$, the values of $\alpha_i$ are all the same. Now we only need to show that the $\alpha_i$ for which $b_i = c$ are exactly $\alpha_1  = \alpha_2 = \ldots = \alpha_{i_0}$.

Because the $A_i$ are strictly increasing by increasing $i$, there is some cutoff point at which $\frac{1}{A_i}$ becomes less than $c$. Let $i_0$ be that index, and more specifically, let $i_0$ be the index satisfying that $b_{i_0} = c$ and $b_{i_0 + 1} = \frac{1}{A_{i_0 + 1}}$. We get that $\alpha_1  = \alpha_2 = \ldots = \alpha_{i_0}$, and hence for all $i \leq i_0$ we get that $A_i = i \cdot \alpha_1$.
Moreover, since $b_i = \min \set{c, \frac{1}{A_i}}$, we get that $b_i = c \iff c \le \frac{1}{A_i} = \frac{1}{A_{i-1} + \alpha_i}$. We thus get that $b_i = c \iff c\alpha_1 \leq 1-c(i-1) \alpha_1$ (the exponents in the definition of $T_i$). By solving for $i$, we obtain the following:
\begin{align*}
    c\alpha_1 &\leq 1-c(i-1) \alpha_1\\
    \iff i &\leq \frac{1}{c\alpha_1}
\end{align*}
This implies that
\[
i_0 = \left\lfloor\frac{1}{c \alpha_1}\right\rfloor,
\]
since $i_0$ must be an integer.
\end{proof}

\begin{lemma}\label{lem:largeAlphas}
    Let $i_0 := \left\lfloor\frac{1}{c \alpha_1}\right\rfloor$. Then, for each $i_0< i \leq k-1$, the optimal value of $\alpha_i$ satisfies $\alpha_i = \frac{c\alpha_1}{1-c\alpha_1} A_{i-1}$.
    Moreover, for $i_0 < i \leq k-1$, it holds that $A_i = \frac{1}{1-c \alpha_1}A_{i-1}$.
\end{lemma}
\begin{proof}
    By \Cref{lem:iZero} and the arguments used in its proof, we know that, for all $1 \le i \le i_0$, the worst-case value of $T_i$ is obtained by using $b_i = c$, and by using such a value we obtain that all $T_i$, for $1 \le i \le i_0$, are equal to $T_1 = n^{c \alpha_1}$.
    Hence, in order to optimize the values of $T_i$ for $i_0 < i< k$, we set $T_i$ equal to $T_1$ and derive $\alpha_i$.     Recall that, for all $i_0 < i< k$, we proved that $T_i = n^{\alpha_i / A_i}$. By setting $T_1$ equal to $T_i$ for each $i_0 < i< k$, we thus get the following.
    \begin{alignat*}{2}
         &&n^{c \alpha_1} &= n^{\alpha_i / A_i}\\
         \iff&& c\alpha_1 &= \frac{\alpha_i}{A_i}\\ 
         \iff&& \alpha_i &= c\alpha_1(\alpha_i + A_{i-1})\\
         \iff&& \alpha_i &= \frac{c\alpha_1}{1-c\alpha_1} A_{i-1}
    \end{alignat*}
    From this it follows that, for all $i_0 < i \leq k-1$,
    \begin{align*}
        A_i - A_{i-1} = \alpha_i &= \frac{c\alpha_1}{1-c\alpha_1} A_{i-1}\\
        \iff A_i &= \frac{1}{1-c \alpha_1}A_{i-1}
    \end{align*}
\end{proof}

\begin{lemma}\label{lem:optSolution}
The optimal value of $\alpha_1$ satisfies $1 = \left(\frac{1}{1-c\alpha_1}\right)^{k - i_0} i_0 \alpha_1$.
\end{lemma}
\begin{proof}
    We use \Cref{lem:iZero,lem:largeAlphas} to express $A_{k-1}$ in terms of $i_0$ and $\alpha_1$.
    \begin{align*}
        A_{k-1} &= \left(\frac{1}{1-c\alpha_1}\right)^{k-1 - i_0} A_{i_0} \\&= \left(\frac{1}{1-c\alpha_1}\right)^{k-1 - i_0} i_0 \alpha_1
    \end{align*}
    By setting $T_1 = T_k$, we get that $n^{c\alpha_1} = n^{1- A_{k-1}}$. Hence, we obtain the following.
    \begin{alignat*}{2}
        && n^{c\alpha_1} &= n^{1- A_{k-1}} \\
        \iff&& c\alpha_1 &= 1-A_{k-1} \\
       \iff&& 1-c\alpha_1 &= A_{k-1} = \left(\frac{1}{1-c\alpha_1}\right)^{k-1 - i_0} i_0 \alpha_1\\
        \iff&& 1 &= \left(\frac{1}{1-c\alpha_1}\right)^{k - i_0} i_0 \alpha_1
    \end{alignat*}
\end{proof}

By combining \Cref{lem:iZero,lem:largeAlphas,lem:optSolution}, and by observing that in each phase the algorithm performs at most $O(n^{c\alpha_1})$ rakes, we obtain the following.
\begin{corollary}\label{cor:optSolution}
    The algorithm described in \Cref{sec:NAlgo} computes a $(\gamma, \ell, k)$-decomposition, for some $\gamma \in O(n^{c\alpha_1})$, in $O(n^{c\alpha_1})$ rounds, where $\alpha_1$ is the unique value satisfying $1 = \left(\frac{1}{1-c\alpha_1}\right)^{k - i_0} i_0 \alpha_1$, and $i_0 := \left\lfloor\frac{1}{c \alpha_1}\right\rfloor$.
\end{corollary}
\begin{proof}
    The algorithm described in \Cref{sec:NAlgo} satisfies the desired runtime when using the parameters provided by the optimization problem.\\
    What is left to show is that there exists a unique real solution. For this we will first prove, that the function 
	\[
		f(\alpha_1) = \left(\frac{1}{1-c\alpha_1}\right)^{k - i_0} i_0 \alpha_1 -1
	\]
	is continuous on the interval $(0,\frac{1}{c})$ and then argue, that it is monotonically increasing. The only part of $f$ that is not continuous, is $i_0$. Since $i_0 := \left\lfloor\frac{1}{c \alpha_1}\right\rfloor$ , $i_0$ makes jumps at $ \alpha_1 = \frac{1}{tc} =: \beta_t$, for all $t \in \N_{> 1}$. Hence, if we prove $f$ is continuous in these points, then $f$ is continuous on all of $(0, \frac{1}{c})$.
    Let us compute the left and right limits of $f$ at $\beta_t$.
    First, notice that the limit of $\frac{1}{1-c\alpha_1}$ exists at $\beta_t$ and is equal to $\frac{1}{1-c/(tc)}=\frac{t}{t-1}$.
    Therefore we have the following:
    \[
        \lim_{\alpha_1\to \beta_t^+} f(\alpha_1)= \left(\frac{t}{t-1}\right)^{k-(t-1)} (t-1) \beta_t - 1 = \frac{1}{c} \left(\frac{t}{t-1}\right)^{k-(t-1)}\frac{t-1}{t}-1 = \frac{1}{c} \left(\frac{t}{t-1}\right)^{k-t}-1
    \]
    and
    \[
        \lim_{\alpha_1\to \beta_t^-} f(\alpha_1)= \left(\frac{t}{t-1}\right)^{k-t} t \beta_t - 1 = \frac{1}{c} \left(\frac{t}{t-1}\right)^{k-t}-1.
    \]
    The two limits are equal, therefore $f$ is continuous at $\beta_t$ for every $t$, and hence $f$ is continuous on $(0, \frac{1}{c})$.
	
    Now, let us prove that $f$ is monotonically increasing.
    Since $f$ is continuous and differentiable on the intervals $(\beta_{t+1},\beta_t)$, we will just differentiate $f$ on these intervals and show that the derivative is nonnegative.
    Moreover, notice that on the interval $(\beta_{t+1},\beta_t)$, $i_0 = t$.
    Let $g_t(x) = \left(\frac{1}{1-cx}\right)^{k-t}tx - 1$.
    Then its derivative $g'$ is, by the standard rules of differentiation:
    \[
    g_t'(x) = t (1-cx)^{t-k} + tx (t-k)(-c)(1-cx)^{t-k-1} = t(1-cx)^{t-k-1}(1-cx+cx(k-t)).
    \]
    Let us show that $g_t'(x)>0$ for $x\in (\beta_{t+1},\beta_t)\cap (0,1/c)$ (remember that we require $\alpha_1 < 1/c$).
    As $x<1/c$, we have that $t(1-cx)^{t-k-1} > 0$.
    Moreover,
    \[
        1-cx+cx(k-t)> 0 \Leftrightarrow 1 > cx (1-k+t).
    \]
    We know that $cx< c/c = 1$ and $t\leq k$, i.e. $1-k+t\leq 1$.
    Therefore, $g'$ is nonnegative on $(\beta_{t+1},\beta_t)\cap (0,1/c)$ and only one solution to $f(\alpha_1)=0$ can exist.
    Therefore, $\alpha_1$ is unique. Is it know left to show that $\alpha_1$ exists. For this, we use the Intermediate Value Theorem.
    Let us compute the limits of $f(\alpha_1)$ towards $0$ and $1/c$.
    Let us first compute $\lim_{\alpha_1\to 0^+}f(\alpha_1)$.
    Let $i_0(\alpha_1) = \left\lceil\frac{1}{c\alpha_1}\right\rceil$.
    There is some function $\delta(\alpha_1)$ such that for all $\alpha_1$, $\delta(\alpha_1)\in[0,1)$ and such that $i_0(\alpha_1) = \frac{1}{c\alpha_1} - \delta(\alpha_1)$.
    First, 
    \[
        \lim_{\alpha_1\to 0^+} i_0(\alpha_1)\alpha_1
        = \lim_{\alpha_1\to 0^+} \left(\frac{1}{c\alpha_1}-\delta(\alpha_1)\right)\alpha_1
        = \lim_{\alpha_1\to 0^+} \frac{1}{c}-\delta(\alpha_1)\alpha_1
        = \frac{1}{c}.
    \]
    Now, we focus on computing $\lim_{\alpha_1\to 0^+} (1-c\alpha_1)^{i_0(\alpha_1)-k}$.
    Rewrite 
    \[
    (1-c\alpha_1)^{i_0(\alpha_1)-k} = \exp\left((i_0(\alpha_1)-k)\ln(1-c\alpha_1)\right).
    \]
    We have that
    \[
        (i_0(\alpha_1)-k)\ln(1-c\alpha_1)\sim_{\alpha_1\to 0^+} \left(\frac{1}{c\alpha_1}-\delta(\alpha_1)-k\right)(-c\alpha_1) \longrightarrow_{\alpha_1\to 0^+} -1 + \lim_{\alpha_1\to 0^+}(\delta(\alpha_1)+k)\alpha_1 = -1
    \]
    as $k$ and $\delta(\alpha_1)$ are bounded.
    Finally, we get that $\lim_{\alpha_1\to 0^+} (1-c\alpha_1)^{i_0(\alpha_1)-k} = \frac{1}{e}$ and that
    \[
        \lim_{\alpha_1\to 0^+}f(\alpha_1) = \frac{1}{ce}-1 < 0.
    \]
    Now, let us compute $ \lim_{\alpha_1\to {1/c}^-} f(\alpha_1)$.
    First, notice that $\lim_{\alpha_1\to {1/c}^-} i_0(\alpha_1) = 1$.
    Therefore, if it exists,
    \[
        \lim_{\alpha_1\to {1/c}^-}f(\alpha_1) = -1 + \frac{1}{c} \lim_{\alpha_1\to {1/c}^-} (1-c\alpha_1)^{1-k}.
    \]
    As $k> 1$ and $(1-c\alpha_1) \longrightarrow_{\alpha_1\to {1/c}^-} 0^+$, we get $\lim_{\alpha_1\to {1/c}^-} (1-c\alpha_1)^{1-k} = +\infty$, and
    \[
        \lim_{\alpha_1\to {1/c}^-}f(\alpha_1) = + \infty.
    \]
    Intermediate Value Theorem, as $0\in\left(\frac{1}{ce}-1,+\infty\right)$, and as $f$ is continuous, $\alpha_1$, which is the solution to $f(\alpha_1)=0$, exists.
    This finishes the proof.
\end{proof}

By invoking \Cref{lem:all-problems} we get that all LCLs can be solved in that round complexity, note that the $O(\log^* I)$ term comes from the compress precomputation (see \Cref{sec:decAndLandscape}).
\RcUpperBound*

With the optimal values determined, we prove some additional properties of the optimal values, that will be needed in the following sections.
\begin{lemma}\label{lem:Utilities}
    For $i_0 < i \leq k-1$, the following statements are true for any optimal solution to the optimization problem.
    \begin{align}
        \frac{\alpha_i}{A_i} &= \alpha_1 c \label{eq:setEq}\\
        \alpha_i &= \frac{1}{1-c\alpha_1} \alpha_{i-1} \label{eq:alphas}\\
        \alpha_{i_0} & < \alpha_{i_0 + 1} < \ldots < \alpha_{k-1} \label{eq:alphaIncrease}\\
        A_i &\geq \frac{1}{c} \label{eq:boundAi}
    \end{align}
\end{lemma}
\begin{proof}
    
    Fact~(\ref{eq:setEq}) follows from the prove of \Cref{lem:largeAlphas}, where we set $n^{c\alpha_1} = n^{\alpha_i/A_i}$
    \newline
    For Fact~(\ref{eq:alphas}): 
    \begin{align*}
        \alpha_i = \frac{c\alpha_1}{1-c\alpha_1} A_{i-1} = \frac{1}{1-c\alpha_1} \cdot \frac{\alpha_{i-1}}{A_{i-1}} \cdot A_{i-1} = \frac{1}{1-c\alpha_1} \alpha_{i-1}
    \end{align*}
    \newline
    Then Fact~(\ref{eq:alphaIncrease}) follows directly from Fact~(\ref{eq:alphas}), by observing that $c\alpha_1 < 1$
    \newline
    For Fact (\ref{eq:boundAi}) we start with (\ref{eq:setEq}) and get 
    \begin{align*}
         & \frac{\alpha_i}{A_i} = \alpha_1 c \\
        \iff & A_i = \frac{1}{c} \frac{\alpha_i}{\alpha_1}
    \end{align*}
    By \Cref{lem:iZero} and Fact~(\ref{eq:alphaIncrease}) it follows that $\alpha_i > \alpha_1$ and so (\ref{eq:boundAi}) follows.\\
    
\end{proof}

\section{Lower Bound for the Case in Which a Polynomial Upper Bound on \texorpdfstring{$n$}{n} is Given}\label{sec:LBpolyPromise}
Next we show that our algorithm from \Cref{sec:NAlgo} is actually optimal. This is interesting as it essentially implies that our algorithm exploits the given information as much as possible. So the tradeoff between upfront work and exploiting the knowledge obtained during earlier rounds is optimal. Note that we have essentially captured the inner workings of this tradeoff in the optimization problem. As such, we will also refer to it often when proving our lower bounds.

In this section, we prove lower bounds for the setting in which nodes are provided with some integer $N$ and the promise that:
\begin{itemize}
    \item $n \le N \le n^c$ for some positive integer $c$, and nodes do not know $n$ but they know $c$.
    \item IDs are from $\{1,\ldots,N\}$.
\end{itemize} While the lower bounds of \Cref{sec:lb-no-input} hold using Knuth's definition of $\Omega$, in this section we prove lower bounds using Hardy-Littlewood's definition of $\Omega$. For more information on those definitions, refer to \Cref{subsec:omega_definition}. We start by proving that such weaker statements are necessary, since no improved lower bounds can be obtained for Knuth's definition of $\Omega$.

\subsection{An Algorithm That Requires \texorpdfstring{$O(n^{1/k})$ Rounds}{O(n1/k) Rounds} for Infinitely Many Values of $n$}\label{ssec:ubknuth}
We prove that, for the setting considered in this section, for infinitely many values of $n$, it is possible to compute a $(\gamma,\ell,k)$-decomposition in $O(\gamma)$ rounds for $\gamma \in O(n^{1/k})$.
This implies that, for Knuth's definition of $\Omega$, the $(\gamma,\ell,k)$-decomposition problem is not in $\Omega(T)$, for any $T$ that is asymptotically strictly larger than $n^{1/k}$.  By \Cref{lem:all-problems}, the same statement will then hold for all LCLs that in the standard LOCAL model have complexity $O(n^{1/k})$, and in particular for $k$-hierarchical $2\frac{1}{2}$-coloring.

\begin{lemma}\label{lem:decomp-no-knuth}
    Let $\Pi$ be the problem of computing a $(\gamma(n),\ell,k)$-decomposition satisfying that for all $n$, $\gamma(n) = O(n)$, but for infinitely many values of $n$, $\gamma(n) = O(n^{1/k})$.
    For Knuth's definition of $\Omega$, the problem $\Pi$ is not in $\Omega(T)$, for any $T$ that is asymptotically strictly larger than $n^{1/k}$.
\end{lemma}
\begin{proof}
We prove that for infinitely many values of $n$, it is possible to compute a $(\gamma(n),\ell,k)$-decomposition in $O(\gamma(n))$ rounds, for $\gamma(n) \in O(n^{1/k})$. For all the other values of $n$, $\gamma(n)$ will be in $O(n)$.
Let $S = s_1, s_2, \ldots$ be the sequence defined as follows.
    \begin{align*}
        s_i &= 2 \text{ if }i = 1\\
        s_{i} &= s_{i-1}^{2c} \text{ otherwise }
    \end{align*}

The algorithm receives as input $N$, which is guaranteed to satisfy $n \le N \le n^c$. Recall that $c$ is known by the algorithm, and that all nodes receive the same value of $N$. The nodes compute $S^* = S \cap \{\lfloor N^{1/c} \rfloor,\ldots,N\}$. Observe that, by the definition of $S$, it either holds that $S^*$ is empty or that it contains a single element $s$. 

Let $X = N$ if $S^*$ is empty, and let $X = s$ otherwise. 
The algorithm, for $k-1$ times, performs $X^{1/k}$ rakes and a single compress. Then, it performs rakes until the graph becomes empty. Each of the $k-1$ phases costs $O(X^{1/k})$ rounds. After the $k-1$ phases, by \Cref{lem:rcremaining}, the number of remaining nodes is $O(\frac{n}{X^{(k-1)/k}})$. Hence, the runtime of the algorithm is 
$O\left(\max\set{\frac{n}{X^{(k-1)/k}}, X^{1/k}}\right)$.
While the runtime of the algorithm is clearly $O(n)$ and the algorithm clearly computes a $(O(n),\ell,k)$-decomposition, observe that, if $n = s_i$ for some $i$, then $X = n$, and hence for infinitely many values of $n$ it holds that $X = n$, the runtime is $O(n^{1/k})$, and the result is a $(O(n^{1/k}),\ell,k)$-decomposition.
\end{proof}
By combining \Cref{lem:decomp-no-knuth} with \Cref{lem:all-problems}, we obtain our theorem.

\NoKnuthLB*
Observe that \Cref{thm:no-kuth-LB} applies to $k$-hierarchical $2\frac{1}{2}$-coloring as well.

\subsection{Lower Bound for $k$-Hierarchical $2\frac{1}{2}$-coloring}\label{sec:Lowerbound}
We now prove that, according to the Hardy-Littlewood's definition of $\Omega$, the algorithm for $k$-hierarchical $2\frac{1}{2}$-coloring guaranteed to exist by \Cref{cor:optSolution} and \Cref{lem:all-problems} is tight. More specifically, we devote the rest of the section in proving the following theorem.

\HalfColLBPoly*

We start by proving that, similarly as in the proof of \Cref{thm:2half-no-bounds}, if all nodes below some layer $i$ output $D$, then also nodes at layer $i$ output $D$, unless they spend some large runtime.
\begin{lemma}\label{lem:onlyDecl}
Let $\mathcal{A}$ be an algorithm for $k$-hierarchical $2\frac{1}{2}$-coloring, let $G$ be a $k$-hierarchical lower bound graph with parameters $\ell_1,\ldots,\ell_k$ and an arbitrary ID assignment, and let $1 \le i \le k$ be an integer parameter.
Suppose $\mathcal{A}$, on $G$, outputs $D$ on all nodes of level $< i$. Then, either:
\begin{itemize}
    \item $\mathcal{A}$ outputs $D$ on all nodes of level $i$, or
    \item at least one node $z$ of level $i$ outputs $B$ or $W$, runs for strictly more than $\lceil \ell_i / 8 \rceil$ rounds, and is at distance strictly larger than $\lceil \ell_i / 8 \rceil$ from the endpoints of the path containing $z$.
\end{itemize}
\end{lemma}
\begin{proof}
Suppose that at least one node of level $i$ does not output $D$. Then, there must exist a path $P$ of level $i$ containing a node $w$ that does not output $D$. Since all nodes of levels $< i$ output $D$, the output of $w$ cannot be $E$, and hence it must be $B$ or $W$. By the constraints of the problem we get that $P$ must be properly $2$-colored. Suppose for a contradiction that all nodes of $P$ run for at most $T = \lceil \ell_i / 8 \rceil$ rounds. We modify $G$ to create a new instance $G'$ as follows.
Let $u$ and $v$ be two nodes of $P$ satisfying the following:
\begin{itemize}
    \item both $u$ and $v$ are at distance at least $T+1$ from the endpoints of $P$;
    \item $u$ and $v$ are at distance at least $2T+1$ from each other.
\end{itemize}
Let $z$ be an arbitrary node that lies between $u$ and $v$ in $P$, and such that it is at distance at least $T+1$ from both $u$ and $v$.
Let $z_1$ and $z_2$ be the two neighbors of $z$ in $P$. We create $G'$ as follows:
\begin{itemize}
    \item remove the edges $\{z,z_1\}$ and $\{z,z_2\}$;
    \item add the edge $\{z_1,z_2\}$;
    \item add an edge $\{x,z\}$ for an arbitrary node $x$ that is at distance at least $T+1$ from both $u$ and $v$.
\end{itemize}
We obtain that the parity of the distances between $u$ and $v$ is different in $G$ and in $G'$, but their view in $T$ rounds is the same. Hence, $\mathcal{A}$ must fail either in $G$ or in $G'$ in producing a $2$-coloring of the path containing them, reaching a contradiction.

\end{proof}

We now prove \Cref{thm:lb-known-N}. 
For a contradiction, assume that there exists a deterministic algorithm $\A$ with runtime $T(n) \in o(n^{c\alpha_1})$. As a result, for any $\varepsilon > 0$, there must exist an integer $n_0 = f(\varepsilon)$ such that, for all $n \geq n_0$, it holds that $T(n) \leq \varepsilon n^{c\alpha_1}$. 


Let $\beta$ be the constant $c_2$ guaranteed to exist by \Cref{obs:nodes-k-hierarchical}. Let $\varepsilon = \frac{1}{100\beta \sigma}$ for some constant $\sigma \ge 1$ to be fixed later, and let $n_0 = f(\varepsilon)$.
We first create an instance $G$ of the $k$-hierarchical lower bound graph with parameters $\ell_1 = \lceil n_1^{\alpha_1} \rceil, \ell_2 = \lceil n_1^{\alpha_2} \rceil, \ldots, \ell_{k-1} = \lceil n_1^{\alpha_{k-1}} \rceil, \ell_k = \lceil n_1^{1-A_{k-1}} \rceil$, where $n_1 = n_0^c$ and the values of $\alpha_i$ are given by \Cref{lem:iZero,lem:largeAlphas}. Let $n$ be the number of nodes of $G$. By \Cref{obs:nodes-k-hierarchical}, $n_1 \le n \le \beta n_1$.
We set $N = n$ and assign IDs $1, \ldots, n$ to the nodes of $G$ arbitrarily.

We prove by induction on $i$ that, on $G$, for all levels $1 \le i \le k-1$, all nodes of layer $\le i$ output $D$. We will later prove that this implies a contradiction for the nodes at level $k$. We start by considering the base case $i=1$.

\begin{claim}\label{lem:lbLevel1}
    All level $1$ nodes of $G$ output $D$.
\end{claim}
\begin{proof}
    Assume for a contradiction that at least one node $v$ of level $1$ does not output $D$. Let $P$ be the path containing $v$, and let $T = \lceil \ell_1 / 8 \rceil + 1$.
    By \Cref{lem:onlyDecl}, this implies that there is at least one node $v$ satisfying the following:
    \begin{itemize}
        \item Node $v$ is at distance at least $T$ from the endpoints of $P$;
        \item Node $v$ runs for at least $T$ rounds.
    \end{itemize} 
    We construct a new instance $G'$ of $\bar{n} = \lceil n^{1/c} \rceil$ nodes as follows:
    \begin{itemize}
        \item Start from the radius-$T$ neighborhood of $v$, which is a path $P'$ of $2T+1 < \frac{1}{2} n^{\alpha_1} \le \frac{1}{2} \hat{n}$ nodes; ($\alpha_1 < \frac{1}{c}$ by \Cref{cor:betterThanGlobal})
        \item Connect a path $P''$ of $\bar{n} - (2T+1)$ nodes to an arbitrary endpoint of $P'$;
        \item Assign to the nodes of $P''$ arbitrary IDs from $\{1,\ldots,n\}$ that are not used in $P'$.
    \end{itemize}
    We now run $\mathcal{A}$ on $G'$ by giving $N = n$ to the nodes. Observe that this is an allowed input, since:
    \begin{itemize}
        \item The instance $G'$ has $\bar{n}$ nodes, and it holds that $\bar{n} \le N = n \le \lceil n^{1/c}\rceil^c = \bar{n}^c$.
        \item IDs are in $\{1,\ldots, N\}$.
    \end{itemize}
    Moreover, observe that the radius-$T$ neighborhood of $v$ is the same in both $G$ and $G'$, and the given $N$ is the same in both instances. Hence, node $v$ runs for at least $T$ rounds also on $G'$.

    Observe that $\bar{n} = \lceil n^{1/c} \rceil \ge \lceil n_1^{1/c} \rceil = \lceil (n_0^c)^{1/c} \rceil = n_0$. Hence, by assumption, node $v$ must terminate in at most $T \le \frac{1}{100\beta \sigma} \bar{n}^{c \alpha_1}$ rounds, for $\sigma \ge 1$. We thus get that the runtime $T$ of $v$ satisfies:
    \begin{itemize}
        \item $T \ge \lceil \ell_1 / 8 \rceil +1 \ge \lceil n_1^{\alpha_1} / 8 \rceil +1 \ge \lceil \frac{1}{8\beta} n^{\alpha_1}  \rceil +1 \ge \frac{1}{8\beta} n^{\alpha_1}  +1 \ge \frac{1}{50\beta} n^{\alpha_1} +1$;
        \item $T \le \frac{1}{100\beta} \bar{n}^{c \alpha_1} = \frac{1}{100\beta} \lceil n^{1/c} \rceil^{c \alpha_1} \le \frac{1}{50\beta} ( n^{1/c} )^{c \alpha_1} = \frac{1}{50\beta} n^{\alpha_1}$.
    \end{itemize}
    Hence, we reach a contradiction.

\end{proof}

We now consider the inductive step. That is, assuming that all nodes of levels $< i$ output $D$, we prove that all nodes of level $i$ output $D$ as well.
\begin{claim}\label{lem:lbLeveli}
    Let $i$ be an integer satisfying $2\leq i \leq k-1$.
    Assume that, on $G$, all nodes of level $< i$ output $D$. Then,  all nodes of level $i$ output $D$.
\end{claim}
\begin{proof}
By assumption, all nodes of level $< i$ output $D$. Assume for a contradiction that at least one node $v$ of level $i$ does not output $D$. Let $P$ be the path containing $v$, and let $T = \lceil \ell_i / 8 \rceil + 1$. By \Cref{lem:onlyDecl}, this implies that there is at least one node $v$ of level $i$ satisfying the following:
\begin{itemize}
    \item Node $v$ is at distance at least $T$ from the endpoints of $P$;
    \item Node $v$ runs for at least $T$ rounds.
\end{itemize}
Similarly as in the proof of \Cref{lem:lbLevel1}, we will construct a new instance $G'$. However, this time, the value of $\bar{n}$ will depend on $i$. We start by proving an upper bound of $n^* = \min\{n,\sigma \lceil n^{A_i} \rceil\}$ on the number of nodes in the radius-$T$ neighborhood of $v$, where $\sigma \ge 1$ is a large enough constant (which depends on $k$) to be fixed later.
By \Cref{obs:nodes-k-hierarchical}, and the fact that $v$ is at distance at least $T$ from the endpoints of $P$, the nodes in the radius-$T$ neighborhood of $v$ are at most:
\[
(2T+1) \prod_{1 \le j \le i-1} \ell_j \le (2T+1)2^{k-1} n_1^{A_{i-1}} \le (2T+1)2^{k-1} n^{A_{i-1}} \le \sigma n^{A_i} \text{, for large enough }\sigma.
\]

In order to construct $G'$, we consider two separate cases, namely the case in which $i \le i_0$ and the case in which $i > i_0$.
\begin{itemize}
    \item Case $i \le i_0$. By \Cref{lem:iZero}, it holds that $A_i = i \cdot \alpha_1$ and $\alpha_i = \alpha_1$. We thus obtain the following.
    \begin{align*}
        n^* &\le \sigma \lceil n^{i \cdot \alpha_1} \rceil \leq \sigma \lceil n^{i_0 \cdot \alpha_i} \rceil < \sigma  \left\lceil n^{\left \lfloor \frac{1}{c\alpha_1} \right \rfloor \cdot \alpha_i} \right\rceil \\
        &< \sigma \lceil n^{ \frac{1}{c\alpha_1} \cdot \alpha_i} \rceil  = \sigma \lceil n^{\frac{1}{c}} \rceil
    \end{align*}
    We choose $\bar{n} = \min\{n,\sigma \lceil n^{1/c} \rceil\}$, and we obtain $G'$ by taking the radius-$T$ neighborhood of $v$ and adding nodes in order to get exactly $\bar{n}$ nodes, in  such a way that the radius-$T$ neighborhood of $v$ does not change (i.e., we connect a path to an arbitrary node at distance exactly $T$ from $v$, which exists in $P$). To the added nodes we assign unused IDs from $\{1,\ldots,n\}$.
    We now run $\mathcal{A}$ on $G'$ by giving $N = n$ to the nodes. Observe that this is an allowed input, since:
    \begin{itemize}
        \item The instance $G'$ has $\bar{n}$ nodes, and it holds that $\bar{n} \le N = n \le \min\{n,(\sigma \lceil n^{1/c}\rceil)^c\} \le \bar{n}^c$.
        \item IDs are in $\{1,\ldots, N\}$.
    \end{itemize}
    Moreover, observe that the radius-$T$ neighborhood of $v$ is the same in both $G$ and $G'$, and the given $N$ is the same in both instances. Hence, node $v$ runs for at least $T$ rounds also on $G'$.

    Observe that $\bar{n} \ge n_0$. Hence, by assumption, node $v$ must terminate in at most $T \le \frac{1}{100\beta \sigma} \bar{n}^{c \alpha_1}$ rounds. We thus get that the runtime $T$ of $v$ satisfies:
    \begin{itemize}
        \item $T \ge \lceil \ell_i / 8 \rceil +1 \ge \lceil n_1^{\alpha_i} / 8 \rceil +1 \ge \lceil \frac{1}{8\beta} n^{\alpha_i}  \rceil +1 \ge \frac{1}{8\beta} n^{\alpha_i}  +1 \ge \frac{1}{50\beta} n^{\alpha_1} +1$;
        \item $T \le \frac{1}{100\beta\sigma} \bar{n}^{c \alpha_1} \le \frac{1}{100\beta\sigma} (\sigma \lceil n^{1/c} \rceil)^{c \alpha_1} \le \frac{1}{50\beta}  ( n^{1/c} )^{c \alpha_1} = \frac{1}{50\beta} n^{\alpha_1}$.
    \end{itemize}
    Hence, we reach a contradiction.
    \item Case $i > i_0$. By \Cref{lem:Utilities}, we know that $A_i \ge 1/c$ and that $\frac{\alpha_i}{A_i} = \alpha_1 \cdot c$. We choose $\bar{n} = n^*$, and similarly as before we obtain $G'$ by taking the radius-$T$ neighborhood of $v$ and adding nodes in order to get exactly $\bar{n}$ nodes, using IDs from $\{1,\ldots,n\}$, in a way that satisfies that the radius-$T$ neighborhood of $v$ does not change. Again, we run $\mathcal{A}$ on $G'$ by giving $N = n$ to the nodes. Observe that this is an allowed input, since:
    \begin{itemize}
        \item The instance $G'$ has $\bar{n}$ nodes, and it holds that $\bar{n} \le N = n \le \min\{n,(\sigma \lceil n^{A_i}\rceil)^c\} \le \bar{n}^c$, where the second inequality holds because $A_i \ge 1/c$.
        \item IDs are in $\{1,\ldots, N\}$.
    \end{itemize}
    Moreover, observe that the radius-$T$ neighborhood of $v$ is the same in both $G$ and $G'$, and the given $N$ is the same in both instances. Hence, node $v$ runs for at least $T$ rounds also on $G'$.

    Observe that $\bar{n} = \min\{n,\sigma \lceil n^{A_i} \rceil\} \ge \min\{n,\lceil n^{1/c} \rceil\} \ge \min\{n,\lceil (n_0^c)^{1/c} \rceil\} \ge n_0$. Hence, by assumption, node $v$ must terminate in at most $T \le \frac{1}{100\beta \sigma} \bar{n}^{c \alpha_1}$ rounds. We thus get that the runtime $T$ of $v$ satisfies:
    \begin{itemize}
        \item $T \ge \lceil \ell_i / 8 \rceil +1 \ge \lceil n_1^{\alpha_i} / 8 \rceil +1 \ge \lceil \frac{1}{8\beta} n^{\alpha_i}  \rceil +1 \ge \frac{1}{8\beta} n^{\alpha_i}  +1 \ge \frac{1}{50\beta} n^{\alpha_i} +1$;
        \item $T \le \frac{1}{100\beta\sigma} \bar{n}^{c \alpha_1} = \frac{1}{100\beta\sigma} \bar{n}^{\frac{\alpha_i}{A_i}}
        \le \frac{1}{100\beta\sigma} (\sigma \lceil n^{A_i} \rceil)^{\frac{\alpha_i}{A_i}} \le  \frac{1}{50\beta} n^{\alpha_i}$.
    \end{itemize}
    Hence, we reach a contradiction.
\end{itemize}

\end{proof}

By combining \Cref{lem:lbLevel1} with \Cref{lem:lbLeveli} we obtain that, on $G$, all nodes in levels $1, \ldots, k-1$ output $D$.
Consider the nodes of $G$ at level $k$. They form a path $P$ of length $\ell_k$. Since $D$ is not allowed on nodes of level $k$, by \Cref{lem:onlyDecl} we obtain that at least one node $v$ in $P$ must spend at least $\lceil \ell_k / 8 \rceil$ rounds.
We thus get that the runtime $T$ of $v$ satisfies:
\begin{itemize}
    \item $T \ge \lceil \ell_1 / 8 \rceil +1 \ge  \lceil n_1^{1 - A_{k-1}} / 8 \rceil +1 = \lceil n_1^{c \alpha_1} / 8 \rceil + 1
    \ge \lceil \frac{1}{8\beta} n^{c \alpha_1}  \rceil +1 \ge \frac{1}{8\beta} n^{c \alpha_1}  +1 \ge \frac{1}{100\beta} n^{c \alpha_1} +1$, where $1 - A_{k-1} = \alpha_1 c$ is given by the optimization problem (\Cref{lem:optSolution});
    \item $T \le \frac{1}{100\beta\sigma } n^{c \alpha_1} \le \frac{1}{100\beta } n^{c \alpha_1} $.
\end{itemize}
Hence, we reach a contradiction. So the Theorem is proven.

\subsection{Lower Bound for $k$-Rake-and-Compress Decomposition}\label{sec:Lowerbound2}
We now prove that the lower bound of \Cref{thm:lb-known-N} holds for $k$-rake-and-compress decomposition as well.

\RcLBPoly*
We devote the rest of the section to proving \Cref{thm:lb-known-N-rc}. The proof of such theorem will borrow ideas from the proofs of \Cref{thm:decomposition-no-bounds,thm:lb-known-N}. We start by proving a lemma similar to \Cref{lem:onlyDecl}. 

\begin{lemma}\label{lem:only-rakes-is-hard}
    Let $i \ge 1$, and let $P_i$ be a path of level $i$ in a $k$-hierarchical lower bound graph. Assume that all nodes of lower layers connected to nodes of $P_i$ are rake nodes of level $\ge i$ or compress nodes of level $\ge i-1$. Then, either:
    \begin{itemize}
        \item there is at least one node of $P_i$ labeled compress node of level $\ge i$, or 
        \item at least one node $z$ of $P_i$ runs for strictly more than $\lceil \ell_i / 8 \rceil$ rounds, and is at distance strictly larger than $\lceil \ell_i / 8 \rceil$ from the endpoints of $P_i$.
    \end{itemize}
\end{lemma}
\begin{proof}
    The proof follows by applying the same ideas used in the proof of \Cref{thm:decomposition-no-bounds}.
    Suppose for a contradiction that there is no node labeled compress of level $\ge i$ and that all nodes of $P_i$ (and far enough from the endpoints) run for at most  $\lceil \ell_i / 8 \rceil $ rounds.
    By assumption, all nodes connected to $P_i$ from lower layers are rake nodes of layer $\ge i$ or compress nodes of layer $\ge i-1$. Observe that no node of $P_i$ can be a compress node of layer $i - 1$, since, in the subgraph induced by $P_i$ and all nodes connected to at least one node of $P_i$, nodes of $P_i$ have degree $\ge 3$. Hence, by \Cref{obs:rc-path}, all nodes of $P_i$ must be rake nodes of layer $\ge i$.
    By \Cref{obs:no-2-outgoing}, there must exist exactly one node $v$ of $P_i$ satisfying that all edges of $P_i$ are oriented towards $v$. 
    As argued in \Cref{thm:decomposition-no-bounds} and in \Cref{thm:decomposition-no-bounds}, if all nodes of $P_i$ (far enough from the endpoints) have a runtime that is sufficiently smaller compared with the length of $P_i$, then we can create a new instance that gives a contradiction with the correctness of the algorithm.
\end{proof}

We now prove \Cref{thm:lb-known-N-rc}.
For a contradiction, assume that there exists a deterministic algorithm $\A$ with runtime $T(n) \in o(n^{c\alpha_1})$. As a result, for any $\varepsilon > 0$, there must exist an integer $n_0 = f(\varepsilon)$ such that, for all $n \geq n_0$, it holds that $T(n) \leq \varepsilon n^{c\alpha_1}$. 

Let $\beta$ be the constant $c_2$ guaranteed to exist by \Cref{obs:nodes-k-hierarchical}. Let $\varepsilon = \frac{1}{100\beta \sigma}$ for some constant $\sigma \ge 1$ to be fixed later, and let $n_0 = f(\varepsilon)$.
We first create an instance $G$ of the $k$-hierarchical lower bound graph with parameters $\ell_1 = \lceil n_1^{\alpha_1} \rceil, \ell_2 = \lceil n_1^{\alpha_2} \rceil, \ldots, \ell_{k-1} = \lceil n_1^{\alpha_{k-1}} \rceil, \ell_k = \lceil n_1^{1-A_{k-1}} \rceil$, where $n_1 = n_0^c$ and the values of $\alpha_i$ are given by \Cref{lem:iZero,lem:largeAlphas}. Let $n$ be the number of nodes of $G$. By \Cref{obs:nodes-k-hierarchical}, $n_1 \le n \le \beta n_1$.
We set $N = n$ and assign IDs $1, \ldots, n$ to the nodes of $G$ arbitrarily.

We prove by induction on $i$ that, on $G$, for all levels $1 \le i \le k-1$, for all $P_i$, if $v$ is the endpoint of $P_i$ that is connected to a node of $P_{i+1}$, then $v$ is a rake node of level $\ge i+1$ or a compress node of level $\ge i$. We will later show that this gives a contradiction on $P_k$. We start by considering the base case $i=1$.
\begin{claim}\label{claim:lb2-base}
    For all $P_1$, if $v$ is the endpoint of $P_1$ that is connected to a node of $P_{2}$, then $v$ is a rake node of level $\ge 2$ or a compress node of level $\ge 1$.
\end{claim}
\begin{proof}
    We start by proving that $P_1$ contains at least one node labeled compress node of level $\ge 1$.
    Assume for a contradiction that $P_1$ contains only rake nodes. By \Cref{lem:only-rakes-is-hard}, this implies that there is at least one node $z$ of $P_1$ satisfying the following:
    \begin{itemize}
        \item Node $z$ is at distance at least $T$ from the endpoints of $P_1$;
        \item Node $z$ runs for at least $T$ rounds.
    \end{itemize}
    From these assumptions, we can reach a contradiction in the exact same way as in the proof of \Cref{lem:lbLevel1}. Hence, $P_1$ contains at least one node labeled compress node of level $\ge 1$. As shown in the proof of \Cref{thm:decomposition-no-bounds}, this implies that $v$ is a rake node of level $\ge 2$ or a compress node of level $\ge 1$.
\end{proof}
We now consider the inductive step.
\begin{claim}\label{claim:lb2-ind}
    Let $1 \le i \le k-1$. Assume that, for all $j < i$, it holds that, for all $P_j$, if $v$ is the endpoint of $P_j$ that is connected to a node of $P_{j+1}$, then $v$ is a rake node of level $\ge j+1$ or a compress node of level $\ge j$. Then, for all $P_i$, if $v$ is the endpoint of $P_i$ that is connected to a node of $P_{i+1}$, then $v$ is a rake node of level $\ge i+1$ or a compress node of level $\ge i$.
\end{claim}
\begin{proof}
    We start by proving that $P_i$ contains at least one node labeled compress node of level $\ge i$.
    Similarly as in the proof of \Cref{thm:decomposition-no-bounds}, by the assumptions we obtain that each node of $P_i$ is either a rake node of level $\ge i$ or a compress node of level $\ge i$. If $P_i$ does not contain any compress node of layer $\ge i$, then $P_i$ contains only rake nodes. By \Cref{lem:only-rakes-is-hard}, this implies that there is at least one node $z$ of $P_i$ satisfying the following:
    \begin{itemize}
        \item Node $z$ is at distance at least $T$ from the endpoints of $P_i$;
        \item Node $z$ runs for at least $T$ rounds.
    \end{itemize}
    From these assumptions, we can reach a contradiction in the exact same way as in the proof of \Cref{lem:lbLeveli}. Hence, $P_i$ contains at least one node labeled compress node of level $\ge i$. As shown in the proof of \Cref{thm:decomposition-no-bounds}, this implies that $v$ is a rake node of level $\ge i+1$ or a compress node of level $\ge i$.
\end{proof}
By applying \Cref{claim:lb2-base} and \Cref{claim:lb2-ind} inductively, we obtain that all nodes of lower layers connected to nodes of $P_k$ are rake nodes of level $\ge k$ or compress nodes of level $\ge k-1$, and hence by \Cref{lem:only-rakes-is-hard} and by the assumption on the runtime we obtain that there is at least one node of $P_k$ labeled compress node of level $\ge k$, which is a contradiction, since in a $k$-rake-and-compress decomposition nodes cannot be labeled compress node of level $\ge k$.

\section{Only Bounded IDs}\label{sec:promiseIds}
Consider the model where nodes do not know $n$, but are given the promise that all IDs are integers between $1$ and $n^c$ for some known $c$. We stress that nodes do not know anything else about $n$. They are only given a unique ID and a number $c$, together with the promise that the IDs are in $\{1,\ldots,n^c\}$.

Since this model is more restricted than the model of \Cref{sec:LBpolyPromise}, the lower bounds of that section still apply.

We show that the guarantees of this setting are good enough to get the same complexities as in our lower bounds, by giving an algorithm for $k$-hierarchical $2\frac{1}{2}$-coloring, that matches the complexity of \Cref{thm:lb-known-N}.

\IdAlgo*

So it seems that, for LCLs, it is already good enough to know a bound on the IDs. In fact, on a high level, our strategy is locally estimating $n$, based on the observed IDs. In the first phase, nodes will only continue their execution if they have observed a large enough ID. Afterwards, nodes adjust their behavior based on the results of this first phase.

\subsection{The Algorithm}
The algorithm consists of rules that every node checks in every round. These first two rules are generic and handle the remainder and label $E$.
\begin{enumerate}[label=(\alph*)]
    \item All nodes in the remainder immediately output $D$.
    \item A level $1 <j\le k$ node becomes active once all of its lower-level neighbors have decided on an output. Before that, it is inactive.
    \item When a node $v$ becomes active and one of its lower-level neighbors has as output one of $\set{B,W,E}$, then $v$ immediately outputs $E$.\label{rule:excempt}
\end{enumerate}
Next, is our rule for outputting consistent 2 colorings using labels $B,W$. Only active nodes will participate and only continuous paths of active nodes will be considered. To that end, all nodes will keep track of the maximal continuous level $j$ path that they can see.

\begin{definition}[maximal active level $j$ subpath of $v$]
Every active level $1 < j \le k$ node keeps track of $P_v$, the maximal active level $j$ subpath of $v$. Initially $P_v = (v)$, and then in every round $P_v$ tries to add the (at most) 2 level $j$ nodes that are adjacent to $P_v$. It only adds a node $u$ to $P_v$, if $u$ is also active and did not output a label yet.
\end{definition}

\begin{enumerate}[label=(\alph*)]
    \setcounter{enumi}{3}
    \item \label{rule:colorPaths}For any active node $v$ of some level $1\le j\le k-1$, if $P_v$ is adjacent to a level $j$ node that has output $D$, then $v$ immediately outputs $D$. \\
    If on both sides of $P_v$, $P_v$ either ends in a node of degree 1, or that side is adjacent to a node that outputs $E$, then we will aim to 2-color $P_v$. For this, $v$ waits for an additional $|P_v|$ rounds. After waiting an additional $|P_v|$ rounds, all nodes in $P_v$ can see all of $P_v$ and output a consistent 2-coloring with colors $B,W$.
\end{enumerate}
In very long paths, this will result in nodes exploring their path for a very long time. To still get a fast algorithm, we have some nodes decide to output $D$ earlier.\\
For each $1\le j<k$ we define a condition, such that if this condition is verified for some level $j$ node $v$, then $v$ immediately outputs $D$. These conditions will depend on the results from previous levels and so we keep track of the number of observed nodes.

\begin{definition}[size of a decline]
    For every node $v$, we store the size $n^{(j)}_v$ of the largest level $1 \le j <k$ decline that this node has seen so far. Initially, this value is 0 for all nodes. When a node is assigned the output $D$, our algorithm updates this value. The details of this are given below.
\end{definition}
The idea behind this is that if the size of the previous decline was large, then these nodes must have had a good reason to run for a long time. So as a result, also in the next phase we will be able to spend a lot of time.\\

Before stating the conditions, we fix the behavior of nodes that are adjacent to another node that outputs $D$, as this is important to make sure the size of decline values propagates properly.
\begin{enumerate}[label=(\alph*)]
    \setcounter{enumi}{4}
    \item When a level $j$ node $v$ outputs $D$ as a result of Rule~\ref{rule:colorPaths}, then we set $n^{(j)}_v = n^{(j)}_u$, where $u$ is the node that did output $D$.
\end{enumerate}
We are now ready to give the decline conditions.
\paragraph{Decline Conditions:}
These conditions will fix the behavior of nodes that are in very long paths, where nodes cannot see the endpoints.\\
Only the condition for level-1 nodes will depend on the actual IDs that can be seen in the graph. For all future levels, we will use the computed sizes of the smaller level declines.
\begin{enumerate}[label=(\arabic*)]
    \item \textbf{Condition for level 1:} Let $\maxIDi$ be the maximum ID that $v$ has seen until round $i$, then $v$ will output $D$, if
    \[
    i > (\maxIDi)^{\alpha_1}.
    \]\label{cond:Ids}\\
    We then set $n_v^{(1)} = i$.
    \item \textbf{Condition for level \boldmath{$1<j<k$}:} Let $P_v$ be the maximal active subpath for $v$. For every node $u\in P_v$ and $1 \le i<j$, let $n^{(i)}_u$ be the maximum value of $n_{w}^{(i)}$ for any level $j-1$ neighbor $w$ of $u$. Then $v$ will output $D$, if
    \[
    |P_v|^{(A_j/ \alpha_j)} > \sum_{u \in P_v} n^{(j-1)}_u.
    \]
    Additionally we set 
    \begin{align*}
    n^{(j)}_v &= \sum_{u \in P_v} n_u^{(j-1)}. 
    \end{align*}
    \label{cond:ratioNodes}
\end{enumerate}
This finishes the description of our algorithm.\\

\noindent Note that, because of \Cref{lem:iZero}, for $1 < j \le i_0$ the second condition reduces to 
\[
|P_v|^{(A_j/ \alpha_j)} = |P_v|^{(j \cdot \alpha_1)/ \alpha_1} = |P_v|^j > \sum_{u \in P_v} n_u^{(j-1)}.
\]
And for $i_0<j<k$, because of \Cref{lem:largeAlphas} the condition reduces to 
\[
|P_v|^{(A_j/ \alpha_j)} = |P_v|^{c\alpha_1} > \sum_{u \in P_v} n_u^{(j-1)}.
\]

We start by showing that our algorithm actually produces a correct solution.

\begin{lemma}\label{lem:IdAlgoCorrect}
    The algorithm produces a correct solution to the $k$-hierarchical $2\frac{1}{2}$-coloring problem.
\end{lemma}
\begin{proof}
    We first argue that all nodes eventually output a label and then argue that the labeling satisfies the constraints.
    
    All nodes in the remainder immediately output a label and for any node with a level, it will either decide to output a label because of one of the other rules, or Rule~\ref{rule:colorPaths} will eventually produce a 2 coloring of all paths. Therefore, all nodes do terminate.\\
    We go through all of the conditions for a solution to be correct in order.

    \begin{itemize}
        \item All nodes in the remainder give the only valid output $D$.
        \item None of our rules allow for a node in level $k$ to output $D$, so all nodes in level $k$ output one of $\set{B,W,E}$. On the other hand, in the algorithm, nodes only output $E$ if they have a lower level neighbor that outputs one of $\set{B,W,E}$. Since level-1 nodes do not have lower level neighbors, they never output $E$.
        \item If a node $v$ of any level $j$ outputs one of $B,W$, then all nodes of some level $j$ subpath $P_v$ containing $v$ output a consistent 2-coloring. So if $v$ is inside the path, its neighbors neither output the same color, nor $D$. If instead $v$ is an endpoint of $P_v$, then either $v$ has only one same level neighbor, which is then in $P_v$, or the neighbor not in $P_v$ outputs $E$. Both of which are fine.
        \item The nodes that output $E$ do so exactly based on the rule for a correct solution, so also all $E$ outputs are correct.
    \end{itemize}
\end{proof}
With the correctness proven, what remains to show is that our algorithm achieves the desired complexity. This will be significantly more work.

\subsection{Analyzing the Algorithm}
We will from now on refer to this algorithm as $\A$. To analyse $\A$ we make a connection between the execution of $\A$ and some graph gadgets. Whenever $\A$ makes some nodes output $D$, then such a gadget must be present. By then upper bounding the number of such gadgets, we will see that our algorithm makes sufficient progress, fast enough.

The main aim of these gadgets is to capture the answer to the following question. How must a path look like for some node $v$ to decide to output $D$? We want that if $\A$ is run on such a gadget, then one of its nodes will output $D$. We capture this property by using a simulation argument.

\begin{definition}[Algorithm Simulation]
    For any graph $G=(V,E)$ with some ID-Assignment $\phi:V \rightarrow \N$ to the nodes. The algorithm's output on the pair $(G, \phi)$ is the (partial) output assignment $\sigma: V \rightarrow \set{B,W,D,E}$, that is obtained when running $\A$ on $G$ with ID-assignment $\phi$.
\end{definition}

We show how we use this simulation idea in the following construction.

\begin{definition}
    A level-1 decline gadget $\mathcal{D}= (id_1, \ldots, id_l)$ of length $l$ is an ordered list of IDs, such that if we simulate the algorithm on a length $l$ path $P = (v_1, \ldots, v_l)$ with ID assignments $(id_1, \ldots, id_l)$, all nodes will output label $D$.\\
    We call the path $P$ of $l$ nodes with ID assignment $(id_1, \ldots, id_l)$ the realization $R_\mathcal{D}$ of $\mathcal{D}$.
\end{definition}

Next we show that this definition is not just arbitrary, but rather that whenever some level-1 node $v$ actually decides to output $D$, it is precisely because there is a decline gadget.

\begin{lemma}\label{lem:lvl1DG}
    Consider any execution of the algorithm on a graph $G$. If a level-1 node $v$ outputs $D$, then $v$ is in a path $P$ of $G$ such that the ID assignment of $P$ constitutes a level-1 decline gadget.
\end{lemma}
\begin{proof}
    Let $P$ be the maximal path of level-1 nodes that contains $v$ and $i$ be the round in which $v$ outputs $D$. Let $P_v$ be the subpath that $v$ has seen until round $i$. Then the Ids assigned to the nodes of $P_v$ constitutes a valid level-1 decline Gadget. To see this, it is sufficient to verify that Condition~\ref{cond:Ids} will also be satisfied after $i$ rounds when running the algorithm on an isolated copy of $P_v$. This is true, since in an isolated copy $v$ will see a subset of the IDs it could see in $G$ and so if Condition~\ref{cond:Ids} was satisfied in $G$ then it is also satisfied in $P_v$. Since we have a level-1 path, none of the nodes in $P_v$ ever output $E$. So Rule~\ref{rule:colorPaths} can never apply, since one node between the two endpoints is already labeled $D$. Therefore, all nodes in the simulation will eventually output $D$.
\end{proof}

Since the way $\A$ behaves in the graphs gives us decline gadgets, we can also argue about the behavior of $\A$, by giving an upper bound on the number of decline gadgets. Intuitively, we need to make sure that there are not too many small paths that decide to decline. If this were the case, then our algorithm would not make progress sufficiently fast. To show that this is not the case and our algorithm is well (enough) behaved, we give a bound on the number of decline gadgets that can exist, if we limit the number of available IDs. As a result, we also bound the number of short declining level-1 paths.

\begin{lemma}\label{lem:upDG1}
    For any positive integer $I$, let $\DS_1(I)$ be the maximum number of ID-disjoint level $1$ decline gadgets that exist using IDs $1, 2, \ldots ,I$.
    Then there exists some constant $C$, such that for any $I$,
    $\DS_1(I) \leq C \cdot I^{1-\alpha_1}$.
\end{lemma}
\begin{proof}
    Suppose we are using IDs $1, 2, \ldots ,I$.
    We will count the number of short ID assignments that can result in an output label $D$. The longest such path can be of length at most 
    \[
    i \leq (\maxIDi)^{\alpha_1} \leq I^{\alpha_1}.
    \]
    Let us define $\Lmax = I^{\alpha_1}$. We will first derive an upper bound on the number $p_j$ of paths of length between $L_j = \frac{\Lmax}{2^j}$ and $L_{j-1} - 1$ , for any $0 \leq j \leq \alpha_1\log I$. Since those paths have length less than $L_{j-1}$ , the maximum ID that can be used is
    \[
    L_{j-1}^{(1/\alpha_1)} = \left(\frac{\Lmax}{2^{j-1}}\right)^{(1/\alpha_1)} = \frac{I}{2^{(j-1)/\alpha_1}}.
    \]
    Now assuming that all of these IDs are actually at our disposal and using that each path has length at least $L_j$ we get
    \begin{align*}
    p_j &\leq \frac{\text{\#Ids}}{\text{\#Ids per DG}} = \frac{L_{j-1}^{(1/\alpha_1)}}{L_j} = \left(\frac{I}{2^{(j-1)/\alpha_1}}\right)/\left(\frac{\Lmax}{2^j}\right) = I^{1-\alpha_1} \cdot \frac{2^j}{2^{(j-1)/\alpha_1}}\\ &= I^{1-\alpha_1} \cdot 2 \cdot \frac{2^{(j-1)}}{2^{(j-1)/\alpha_1}} = I^{1-\alpha_1} \cdot 2 \cdot \left(2^{1-(1/\alpha_1)}\right)^{j-1}.
    \end{align*}
    Since $\alpha_1 < 1$, we get that $1-(1/\alpha_1) < 0$. The sum of all $p_j$ can be computed as follows, because the geometric series converges.
    \[
    \sum_{j=0}^{\alpha_1\log(I)} p_j \leq 2 I^{1-\alpha_1} \cdot \sum_{j=0}^{\alpha_1\log(I)} \left(2^{1-(1/\alpha_1)}\right)^{j-1} \in O(I^{1-\alpha_1}).
    \]
\end{proof}

If two nodes $u$ and $v$ both output $D$ and are not on the same path $P$, then by \Cref{lem:lvl1DG} there must also be two separate level-1 decline gadgets. The number of such decline gadgets is upper bounded by \Cref{lem:upDG1}. Note that because Condition~\ref{cond:Ids} produces smaller gadgets when using smaller IDs, we get the maximum number of such gadgets, by using the smallest possible IDs ($1, \ldots, n$).
\begin{corollary}\label{cor:fewLvl1Paths}
    In any execution of the algorithm on an $n$ node graph $G$, there are at most $O(n^{1-\alpha_1})$ decline gadgets and hence also at most $O(n^{1-\alpha_1})$ disjoint level-1 paths that output $D$ as a result of Condition~\ref{cond:Ids}.
\end{corollary}

We will now need to generalize the same notion beyond level-1 decline gadgets.
\begin{definition}[Decline Gadget]
    For all $1<j<k$, a level $j$ decline gadget $D= (P,\phi)$ of length $l$, is a tuple consisting of the following.
    \begin{itemize}
        \item $P = (d_1, d_2, \ldots, d_l)$ is a tuple of level $j-1$ decline gadgets.
        \item $\phi$ is an assignment of IDs to the path $P_R = (v_1, v_2, \ldots, v_{l})$.
    \end{itemize}
    The realization of $D$ is a graph $R_D$ with an ID assignment $\Phi$. It is obtained by starting with $P_R$, with ID assignment $\phi$, and attaching to every node at position $i$ the realisation of decline gadget $d_i$ (which already have IDs assigned).\\
    Then, $D$ is a level $i$ decline gadget if by simulating the algorithm on $(R_D, \Phi)$, all nodes in $P_R$ output $D$.
\end{definition}

Clearly, because of the way we defined Condition~\ref{cond:ratioNodes}, the size of these gadgets will be important.
\begin{definition}[size of a decline gadget]\label{def:cardinalityOfGadgets}
    For any level $1 \le j <k$ decline gadget $D=(P, \phi)$ we define $|D| = |R_D|$, where $R_D$ is the realisation of $D$.
\end{definition}
Note that this is also exactly the value assigned to $n_v^{(j)}$ for all nodes on the main path of the gadget.

We now want to prove a similar result as before, where if some level $j$ node outputs $D$ during the execution of $\A$, then it is because there is some level $j$ decline gadget. To make our arguments a bit cleaner we define the notion of \emph{below}.

\begin{definition}[below]
    For any level $j$ node $v$ in $G$, we say that another node $u$ is \emph{below} $v$ if and only if there exists a decreasing path of lower level nodes that connects $u$ and $v$. Formally there exists $P_u = (v = v_0, v_1, \ldots, v_l = u)$, such that the level of $v_i$ is greater than or equal to the level of $v_{i+1}$, for all $1\le i < l$.\\
    We say $u$ is below some level $j$ path $P$, if there exist some node $v\in P$, such that $u$ is below $v$.
\end{definition}

Now we prove the same connection as in \Cref{lem:lvl1DG}, namely that if a level $j$ node outputs $D$, then it is precisely because there is a level $j$ decline gadget.
\begin{lemma}\label{lem:AlgoImpliesGadgets}
    For all $1\le j<k$, in any execution of the algorithm on a graph $G$, let $P$ be a maximal level $j$ path, such that all nodes of $P$ output $D$. Then there exists a set of nodes $A$ that are in or below $P$, and a level $j$ decline Gadget $D$, such that $|A| = R_D$ and such that the IDs used in $A$ and $R_D$ are the same.
\end{lemma}
\begin{proof}
    We prove the lemma by induction on $j$. The base case is already handled in \Cref{lem:lvl1DG}, so suppose the lemma statement is true for $j-1$, and let us show that it is true for level $j$.\\
    Let $P$ be a maximal path of nodes that output $D$ in round $t$. This means that for some node $v$ and $P_v \subset P$ the maximal active subpath in some round $t$, Condition~\ref{cond:ratioNodes} holds.
    \begin{enumerate}[label=\ref{cond:ratioNodes}]
        \item Let $n^{(j-1)}_u$ be the maximum value of $n_{w}^{(j-1)}$ for any level $j-1$ neighbor $w$ of $u$. Then $v$ will output $D$ if
    \[
    |P_v|^{(A_j/ \alpha_j)} > \sum_{u \in P_v} n^{(j-1)}_u.
    \]
    \end{enumerate}
    Every node $u \in P_v$ has a level $j-1$ neighbor (otherwise $u$ would be level $j-1$), and since $u$ did not output $E$, that level $j-1$ neighbor cannot output any of $\set{B,W,E}$ -- otherwise $u$ would immediately have output $E$. So, this level $j-1$ neighbor has $D$ as output. By induction hypothesis, we get a level $j-1$ decline gadget $d_u$, together with a set $A_u$ for every node in $P_v = (u_1, \ldots, u_{|P_v|})$.\\

    We then construct $D=((d_1, \ldots, d_{|P_v|}), \phi)$, where $\phi$ just assigns the IDs of $P_v$.
    Similarly 
    \[
    A = P_v \cup \bigcup_{u \in P_v} A_u.
    \]
    Clearly, $|A| = |R_D|$, so what is left to show is that $D$ is indeed a decline gadget.\\
    Consider the simulation of $\A$ on $D$.
    Since all of the $d_i$'s are proper level $j-1$ decline gadgets, all of the level $j-1$ neighbors that are adjacent to the main path $P_R$ output $D$.
    As a result, none of the nodes in $P_R$ output $E$ and all of them eventually become active. 
    Note that for the path $P_R$ it is true that 
    \[
    |P_R|^{(A_j/ \alpha_j)} = |P_v|^{(A_j/ \alpha_j)} > \sum_{u \in P_v} n^{(j-1)}_u \ge \sum_{u \in P_R} n^{(j-1)}_u,
    \]
    since $P_R$ has the exact same length as $P_v$ and the values $n^{(j-1)}_u$, are taken as the maximum values over all neighbors. Note that for any $u \in P_R$, only one of its level $j-1$ neighbors is present in $R_D$. Since $n^{(j-1)}_u$ is defined as the maximum value over all neighbors, these values in $P_R$ will be at most as large as the values in $P_v$.\\
    
    Therefore, after at most $t$ rounds of the simulation for $v$, $P_v = P_R$ and Condition~\ref{cond:ratioNodes} is satisfied, as argued above. So $v$ does output $D$. As a result, all of the other nodes also eventually output $D$ as a result of Rule~\ref{rule:colorPaths}, as desired.
\end{proof}

It is now clear that the behavior of our algorithm is closely intertwined with the existence of decline gadgets. Again, we have to argue and get an upper bound on the number of decline gadgets.
This time however, the decline gadgets will be of higher levels.
For this we will have to build larger decline gadgets from smaller ones.
Let us first introduce the notion of minimal decline gadgets. 

\begin{definition}[minimal Decline Gadget]
    We say that a Decline Gadget $D = (P, \phi)$ is minimal if and only if there does not exist a subset $P' \subset P$, such that $D'= (P', \phi')$ is also a Decline Gadget. Here $\phi'$ is the restriction of $\phi$ to $P'_R$. 
\end{definition}

Clearly, any bound on the number of minimal decline gadgets implies a bound on the number of decline gadgets. Given a set $\mathcal{D}$ of decline gadgets, we can transform it into a set $\mathcal{D}'$ of minimal decline gadgets, without decreasing the cardinality. To do this, just take any non-minimal gadget $D=(P, \phi)$ and turn it into a minimal one, by taking a subset $P'$ of $P$, such that $D'= (P', \phi')$ is a minimal decline gadget. 

Consider any level $j$ decline gadget $D = (P, \phi)$.
Using \Cref{def:cardinalityOfGadgets} about the size of a decline gadget, we get that in the simulation of $\A$ on $R_D$, Condition~\ref{cond:ratioNodes} reduces to 
\[
|P|^{(A_j/\alpha_j)} > \sum_{d \in P} |d|.
\]
Indeed, the values $n_v^{(j-1)}$ are exactly equal to $|d|$.

We will now start proving the bound on the number of higher level decline gadgets.
Note that we only aim for an asymptotic bound.
In order to keep our proofs simple, we introduce a series of constants $\kappa_1, \kappa_2, \ldots$, each hiding as many constant terms as possible, and in doing so, reducing the number of terms we have to keep track of.

\Cref{claim:boundGadgets} is the main reason why we want to argue about minimal decline gadgets.

\begin{claim}\label{claim:boundGadgets}
    Let $D = (P, \phi)$ be a minimal level $j$-decline gadget of length $L$. Then there exists some constant $\kappa_2>0$, such that for all $d \in P$ it holds that $|d| \leq \kappa_2 \cdot L^{(A_j/\alpha_j)-1}$.
\end{claim}
\begin{proof}
    Since $D$ is a valid decline gadget, we get
    \[
        L^{(A_j/\alpha_j)} = |P_R|^{(A_j/\alpha_j)} > \sum_{d \in P} |d|.
    \]
    Suppose that $D$ is still a decline gadget after having one of the $d$ removed from $P$. Then the path would get 1 node shorter and the new bound would be:
    \[
    (L-1)^{(A_j/\alpha_j)} > \sum_{d' \in P\setminus\{d\}} |d'|.
    \]
    By a first order Taylor expansion with respect to $L$, for some second order error term $R_2(x) = O(x^2)$ we get
    \begin{align*}
    (L-1)^{(A_j/\alpha_j)} &= L^{A_j/\alpha_j} - \frac{A_j}{\alpha_j} L^{(A_j/\alpha_j) -1} + (L-1)^{(A_j/\alpha_j)}\cdot R_2(1/L).
    \end{align*}
    So there is some constant $\kappa_2>0$, that hides the constants ($A_j$ and $\alpha_j$ are constant) and the error term, such that
    \[
        L^{A_j/\alpha_j} - \frac{A_j}{\alpha_j} L^{(A_j/\alpha_j) -1} + (L-1)^{(A_j/\alpha_j)}\cdot R_2(1/L) \ge L^{A_j/\alpha_j} - \kappa_2 L^{(A_j/\alpha_j) -1}.
    \]
    Now suppose that for some $d \in P$, $|d| > \kappa_2 L^{(A_j/\alpha_j) -1}$. We get
    \begin{align*}
        (L-1)^{(A_j/\alpha_j)} &\ge L^{A_j/\alpha_j} - \kappa_2 L^{(A_j/\alpha_j) -1} > \left(\sum_{d' \in P} |d'| \right) - |d| = \sum_{d' \in P\setminus\{d\}} |d'|.
    \end{align*}
    So $D' = (P\setminus \{d\}, \phi')$ is valid and hence $D$ not minimal. So, for all $d\in P$ it must be that $|d| \le \kappa_2 L^{(A_j/\alpha_j) -1}$.
\end{proof}

The last Claim we need for our bound on the size of a decline gadget correlates the size of a decline gadget with its length. Our lemma argues about decline gadgets of a certain length, so in order to use \Cref{claim:boundGadgets} which talks about the size of gadgets, we need some glue: \Cref{claim:sizeImpliesLength}.
Yet again it introduces some constant -- this time hidden in the big O notation.

\begin{claim}\label{claim:sizeImpliesLength}
    Any minimal level $j$ decline gadget of size at most $S$ has length at most $O(S^{(\alpha_j/A_j)})$.
\end{claim}
\begin{proof}
    Let $D = (P, \phi)$ be a minimal level $j$ decline gadget and let $\kappa_2$ be as in \Cref{claim:boundGadgets}. If 
    \[
    (L-1)^{(A_j/\alpha_j)} > S - \kappa_2 L^{(A_j/\alpha_j) -1},
    \]
    then by the same argument as in \Cref{claim:boundGadgets}, $D$ cannot be minimal.
    So we get 
    \begin{align*}
    &S - \kappa_2 L^{(A_j/\alpha_j) -1} \ge (L-1)^{A_j/\alpha_j} \ge  L^{A_j/\alpha_j} - \kappa_2 L^{(A_j/\alpha_j) -1}\\
    &\Rightarrow S \ge L^{A_j/\alpha_j}\\
    &\Rightarrow S^{\alpha_j/A_j} \ge L \in O(S^{\alpha_j/A_j}).
    \end{align*}
\end{proof}

With these results in place, we can prove our bound on the number of higher level decline gadgets.
\begin{lemma}\label{lem:lengthBoundsNumber}
    For any length $L \in \N$, the number of minimal level $j$ decline gadgets $\DS_j(L)$ of length at most $L$ is bounded by $O(L^{(1-A_j)/\alpha_j})$, where $A_j = \sum_{i=1}^j \alpha_i$.
\end{lemma}
\begin{proof}
    We prove the lemma by induction over $j$. Consider the base case.
    If the length of a level-1 decline gadget is at most $L$, then because of Condition~\ref{cond:Ids}, the largest ID that can be used for such a gadget is $L^{1/\alpha_1}$.
    So by \Cref{lem:upDG1}, we get 
    \[
    \DS_1(L^{1/\alpha_1}) \in O(L^{(1/\alpha_1)(1-\alpha_1)}) = O(L^{(1-A_1)/\alpha_1}).
    \]
    Now, let $1<j<k$ and suppose that the statement of the lemma is true for $j-1$.
    
    Similarly as in the proof of \Cref{lem:upDG1}, we compute a geometric sum over the lengths. For all $i$, let $L_i = \frac{L}{2^i}$. 
    First, we derive a bound on $n_i$, the number of decline gadgets of length at least $L_{i}$ and length at most $L_{i+1}$.
    Any level $j$ decline gadget $D = (P, \phi)$ of length at least $L_i$ consists of at least $L_i$ level $j-1$ decline gadgets. By \Cref{claim:boundGadgets} and since the longest gadget has length $L_{j+1}$, all of these $j-1$ decline gadgets have size at most $\kappa_2' L_{i+1}^{(A_j/\alpha_j - 1)}$ for some constant $\kappa_2'>0$. We will continue with another constant $\kappa_2 >0$ that additionally hides the +1 in the subscript.
    \[
    \kappa_2' L_{i+1}^{(A_j/\alpha_j - 1)} = \kappa_2' (L_{i}/2)^{(A_j/\alpha_j - 1)} = \frac{\kappa_2'}{2^{(A_j/\alpha_j - 1)}} L_{i}^{(A_j/\alpha_j - 1)} = \kappa_2 L_i^{(A_j/\alpha_j - 1)}.
    \]
    
    So, what is the maximum number of minimal level $j-1$ decline gadgets of size at most $\kappa_2 L_i^{(A_j/\alpha_j - 1)}$? By \Cref{claim:sizeImpliesLength}, there exists a constant $\kappa_3>0$, such that any level $j-1$ decline gadget of size at most $\kappa_2 L_i^{(A_j/\alpha_j - 1)}$ has length at most 
    \[
    \kappa_3 \left(\kappa_2 L_i^{(A_j/\alpha_j - 1)}\right)^{(\alpha_{j-1}/A_{j-1})} = \kappa_4 L_i^{(A_j/\alpha_j -1)(\alpha_{j-1}/A_{j-1})}.
    \]
    For a suitable constant $\kappa_4$.\\
    By applying our induction hypothesis, and choosing suitable $\kappa_5,\kappa_6>0$, we get that the number of such level $j-1$ decline gadgets is upper bounded by 
    \[
        \kappa_5 \left(\kappa_4 L_i^{(A_j/\alpha_j - 1)(\alpha_{j-1}/A_{j-1})} \right)^{(1-A_{j-1})/\alpha_{j-1}} = \kappa_6 L_i^{(A_j/\alpha_j -1)(\alpha_{j-1}/A_{j-1})(1-A_{j-1})/\alpha_{j-1}}.
    \]
    Before finally resolving the big exponent, we lastly bound the number of level $j$ decline gadget of length between $L_i$ and $L_{i+1}$, by observing that each of them requires at least $L_{i+1}$ of such level $j-1$ gadgets. Therefore, the total number of such level $j$ decline gadgets is bounded by
    \[
    \frac{\kappa_6 L_i^{(A_j/\alpha_j -1)(\alpha_{j-1}/A_{j-1})(1-A_{j-1})/\alpha_{j-1}}}{L_{i+1}} = 2 \kappa_6 L_i^{(A_j/\alpha_j -1)(\alpha_{j-1}/A_{j-1})(1-A_{j-1})(1/\alpha_{j-1}) - 1}
    \]
    Now we will take care of resolving the exponent $(A_j/\alpha_j -1)(\alpha_{j-1}/A_{j-1})(1/\alpha_{j-1})(1-A_{j-1}) - 1$.
    \paragraph{Case $j\le i_0$:} Then by \Cref{lem:iZero} it holds that $(A_j/\alpha_j) = j$, $A_{j-1}/\alpha_{j-1} = j-1$ and $\alpha_j = \alpha_{j-1} = \alpha_1$. So we get
    \begin{align*}
        &{(A_j/\alpha_j -1)(\alpha_{j-1}/A_{j-1})(1/\alpha_{j-1})(1-A_{j-1}) - 1} = (j-1)\left(\frac{1}{j-1}\right) \left(\frac{1}{\alpha_{j-1}} - \frac{A_{j-1}}{\alpha_{j-1}}\right) -1\\
        &=\frac{1}{\alpha_{j-1}} - (j-1) - 1 = \frac{1}{\alpha_{j}} - j = \frac{1}{\alpha_j} - \frac{A_j}{\alpha_j} = \frac{1}{\alpha_j}(1-A_j)
    \end{align*}
    as desired.
    \paragraph{Case $i_0 < j$:} Then by \Cref{lem:Utilities} the following hold
    \begin{align*}
        &\frac{\alpha_j}{A_j} = \alpha_1 c,\\
        &\alpha_j = \frac{1}{1-c\alpha_1}\alpha_{j-1},\\
        &\frac{\alpha_{j-1}}{A_{j-1}} \le\alpha_1 c.
    \end{align*}
    In the last inequality we don't get an exact equality, as in the case of $j-1 = i_0$ we have
    \[
    \frac{\alpha_{i_0}}{A_{i_0}} = \frac{\alpha_1}{i_0 \cdot \alpha_1} = \lfloor c\alpha_1\rfloor \le c\alpha_1.
    \]
    Using these we obtain
    \begin{align*}
        (A_j/ & \alpha_j-1)(\alpha_{j-1}/A_{j-1})(1/\alpha_{j-1})(1-A_{j-1}) - 1\\
        &\le \left(\frac{1}{c\alpha_1} -1\right) (c\alpha_1)(1/\alpha_{j-1})(1-A_{j-1}) -1\\
        &=(1-c\alpha_1)(1/\alpha_{j-1})(1-A_{j-1}) - 1 \\
        &= (1/\alpha_j)(1-A_{j-1}) - \frac{\alpha_j}{\alpha_j} \\
        &= \frac{1}{\alpha_j} - \frac{A_{j-1}}{\alpha_{j}} - \frac{\alpha_j}{\alpha_j} \\
        &= \frac{1}{\alpha_j} - \frac{A_{j}}{\alpha_{j}} = \left(\frac{1}{\alpha_j}\right)(1-A_j).
    \end{align*}
    So in either case it holds that 
    \[
    2 \kappa_6 L_i^{(A_j/\alpha_j -1)(\alpha_{j-1}/A_{j-1})(1/\alpha_{j-1})(1-A_{j-1}) - 1} \le 2\kappa_6 L_i^{(1/\alpha_j)(1-A_j)}.
    \]
    The desired bound then follows from computing a geometric sum over all $i$ in the same way as in the proof of \Cref{lem:upDG1} (remember that $L_i = L/2^i$).
\end{proof}

\begin{lemma}\label{lem:boundLen}
    For any $1 \le j <k$, the maximum length of a level $j$-decline gadget is at most $O(n^{c \alpha_j})$ and for any $v$, the maximum value of $n^{(j)}_v$ is at most $n^{cA_j}$.
\end{lemma}
\begin{proof}
    For the base case, since the largest ID is $n^c$, Condition~\ref{cond:Ids} is true for all nodes after round $n^{c\alpha_1}$. So, the maximum value set for $n_v^{(1)}$ on any node is at most $n^{c\alpha_1}$.
    
    Suppose now that $j>1$. Let $D$ be a decline gadget $(P, \phi)$, of length $L$.
    For the inductive step, we use the fact that Condition~\ref{cond:ratioNodes} is satisfied together with the induction hypothesis to get 
    \begin{align*}
    L^{A_j/\alpha_j} > \sum_{d \in P} |d| \ge \sum_{d \in P} O(n^{cA_{j-1}})  = O(L \cdot n^{cA_{j-1}}).
    \end{align*}
    So we get that Condition~\ref{cond:ratioNodes} is always satisfied when
    \begin{align*}
    &L^{(A_j/\alpha_j) -1} \in O(n^{cA_{j-1}}) \\
    &\Rightarrow L \in O\left(n^{\frac{cA_{j-1}}{(A_j/\alpha_j) -1}}\right)
    \end{align*}
    and the exponent is
    \begin{align*}
        \frac{cA_{j-1}}{(A_j/\alpha_j) -1} = \frac{c A_{j-1} \alpha_j}{A_j - \alpha_j} = \frac{c A_{j-1}\alpha_j}{A_{j-1}} = c \alpha_j
    \end{align*}
    as desired. 
    The bound on the size then immediately follows
    \[
    n^{(j)}_v  = |D| = L + L^{(A_j/\alpha_j)} \in O(L^{(A_j/\alpha_j)}) = O\left(\left(n^{c\alpha_j}\right) ^{(A_j/\alpha_j)}\right) = O(n^{cA_j}).
    \]
\end{proof}

Using exactly the same argument, we get the following bounds when restricting to only using IDs $1 ,\ldots, n$. Note that by using the smallest possible set of $n$ IDs, we get the largest number of level-1 decline gadgets and they also are the smallest possible, so we also get the largest number of higher level decline gadgets.

\begin{lemma}\label{lem:fewGadgets}
    The maximum number of decline gadgets possible using only $n$ IDs, is obtained by using IDs $1,  \ldots, n$.
    When using only IDs $1,2,\ldots, n$, for any $1 \le j <k$, the maximum length of a level $j$-decline gadget is at most $O(n^{\alpha_j})$ and for any $v$, the maximum value of $n^{(j)}_v$ is at most $n^{A_j}$.
\end{lemma}

We now use these bounds on the number of decline gadgets to argue on the progress our algorithm makes.
Here, we earn the reward for all our technical work we did before, as we just need to invoke our lemmas.

\begin{lemma}\label{lem:boundAdjacent}
    For all $1 < j \le k$, during the execution of $\A$, the number of level $j$ nodes that are adjacent to some level $j-1$ node that outputs $D$ is bounded by $O(n^{1-A_{j-1}})$.
\end{lemma}
\begin{proof}
    By \Cref{lem:AlgoImpliesGadgets}, for each level $j-1$ neighbor $u$ of some level $j$ node $v$, it holds that if $u$ outputs $D$, then there is some level $j-1$ decline gadget $D_u$. Since $D_u$ consists of a main path $P_{D_u}$, which has only 2 endpoints (one of which is $u$), only two nodes not in $D_u$ can be adjacent to $D_u$. One of those is $v$.\\
    Therefore, the number of nodes adjacent to a level $j-1$ decline gadget is bounded by 2 times the number of level $j-1$ decline gadgets $\DS_{j-1}(L)$. By \Cref{lem:fewGadgets}, the largest length $L$ of a level $j-1$ decline gadget is at most $n^{\alpha_{j-1}}$, so by \Cref{lem:lengthBoundsNumber}, the number of level $j$ nodes adjacent to a level $j-1$ node that outputs $D$ in any execution is bounded by 
    \[
    2 \cdot (n^{\alpha_{j-1}})^{(1-A_{j-1})/\alpha_{j-1}} = 2 \cdot n^{(1-A_{j-1})},
    \]
    as desired.
\end{proof}

The last step is to argue that all of this actually happens fast enough. However, in the definition of maximal active level $j$ subpath of $v$, we potentially wait for a long time for nodes to become active. The following lemma shows that we will never wait for too long.

\begin{lemma}\label{lem:nodesBecomeActive}
    All level $1 \le j \le k$ nodes become active after at most $(j-1) \cdot n^{c\alpha_1}$ rounds.
\end{lemma}
\begin{proof}
    We prove the statement by induction over $j$. The base case where all level-1 nodes are active immediately is easy.\\
    Let $1<j\le k$. By the induction hypothesis, all level $j-1$ nodes became active after at most round number $(j-2) n^{c\alpha_1}$. We argue that all of them choose an output after at most an additional $n^{c\alpha_1}$ rounds. Consider some level $j-1$ node $v$, which is in some maximal path $P$ of level $j-1$ nodes. Since all level $j-1$ nodes are active, $P_v$ will grow by at least 1 in every round after round $(j-2) n^{c\alpha_1}$. If $|P| < n^{c\alpha_1}$, then either $v$ outputs a label for some other reason, or eventually Rule~\ref{rule:colorPaths} will apply, because $v$ will see the endpoints of $P$. If instead $|P| > n^{c\alpha_1}$, then if $v$ does not output another label we get that $|P_v| > n^{c\alpha_1} \ge n^{c\alpha_{j-1}}$ (by \Cref{lem:iZero,lem:Utilities}). In the proof of \Cref{lem:boundLen}, we argue that for nodes of level $j-1$, Condition~\ref{cond:ratioNodes} is always satisfied if $|P_v| \ge n^{c\alpha_{j-1}}$. Therefore, $v$ will output $D$. Since $v$ was an arbitrary level $j-1$ node, after at most $(j-2) n^{c\alpha_1} + n^{c\alpha_1} = (j-1) n^{c\alpha_1} $ rounds, all level $j-1$ nodes have chosen an output, this means that all level $j$ nodes must now be active.
\end{proof}
 
We can finally prove \Cref{thm:2.5ColIdsUB}.

\begin{proof}[Proof of \Cref{thm:2.5ColIdsUB}]
    By \Cref{lem:nodesBecomeActive}, after at most $(k-1) n^{c\alpha_1}$ many rounds, all level $k$ nodes become active. This means that all nodes of lower levels (and all nodes in the remainder) have chosen their outputs. What is left to show, is that these nodes do not have to spend too much time.
    
    By \Cref{lem:boundAdjacent} the number of level $k$ nodes adjacent to level $k-1$ nodes that output $D$ is at most $O(n^{(1-A_{k-1})})$. All other level $k$ nodes must be adjacent to lower level neighbors that output one of $\set{B,W,E}$ and hence by Rule~\ref{rule:excempt} these nodes output $E$, the moment they became active, so before or in round $(k-1) n^{c\alpha_1}$.
    
    Therefore, the number of level $k$ nodes that did not yet output a label in round $(k-1) n^{c\alpha_1}$ is at most $O(n^{(1-A_{k-1})})$. So the longest path of level $k$ nodes that did not yet output a label has length at most $O(n^{(1-A_{k-1})})$.
    
    Rule~\ref{rule:colorPaths} will 2-color all of these level $k$ nodes after another at most $O(n^{(1-A_{k-1})})$ rounds.
    
    By \Cref{lem:optSolution}, we get that $n^{(1-A_{k-1})} \in O(n^{c\alpha_1})$. So all nodes decided on an output after at most 
    \[
    (k-1) n^{c\alpha_1} +  O(n^{c\alpha_1}) =  O(n^{c\alpha_1})
    \]
    rounds.
    
    By \Cref{lem:IdAlgoCorrect} the output is correct and hence the theorem is proven.
\end{proof}

\section{Randomized LOCAL Model}\label{sec:rand}
The randomized complexity of $2$-hierarchical $2\frac{1}{2}$-coloring without knowledge of $n$ is, surprisingly, $\Theta\bigl(\frac{n}{\log n}\bigr)$.
\subsection{Upper Bound}
\begin{theorem}\label{thm:randomized_k_2_upper}
    There is an algorithm that solves $2$-hierarchical $2\frac{1}{2}$-coloring with locality $O\bigl(\frac{n}{\log n}\bigr)$ without knowledge of $n$.
\end{theorem}
\begin{proof}
    We give the following algorithm:
    Every node in a level-$1$ path will mark itself with probability $\frac{1}{2}$.
    The path will be declined if at least one node is marked.
    Given a level-$1$ path of length $l$, the probability that it is declined is therefore $1-2^{-l}$. Further, with high probability, each level-$1$ path is either of length $O(\log n)$ or every node on the path sees a marked node on the path within distance $O(\log n)$. The level-$1$ paths can therefore be labeled within $O(\log n)$ rounds.
    
    To bound the round complexity of handling level-$2$ nodes, let us now compute the probability $p$ that the algorithm needs to $2$-color some fixed level-$2$ path of some length $k$.
    Suppose that level-$1$ paths of length $l_1,l_2,\dots,l_k$ are attached to this level-$2$ path.
    We get
    \[
        p   \leq \prod_{i=1}^k (1-2^{-l_i})
            \leq \prod_{i=1}^k \exp\left({-2^{-l_i}}\right)
            = \exp\left({-\sum_{i=1}^k \frac{1}{2^{l_i}}}\right)
    \]
    For $a+b$ constant and $a,b\in \N$, one can verify that $1/2^a + 1/2^b$ is minimal whenever $a=b$, if possible, or $|a-b|=1$ otherwise.
    Therefore, because $\sum_{i=1}^k l_i \leq n$ and by letting $k=2n/(\log n)$, we get that:
    \[
        p \leq \exp\left({-\sum_{i=1}^k \frac{1}{2^{n/k}}}\right)
          = \exp\left({-\frac{2\sqrt{n}}{\log n}}\right)
    \]
    By union bound, we thus get that with high probability, there is no level-$2$ path of length $\omega\bigl(\frac{n}{\log n}\bigr)$ that $\A$ needs to $2$-color. The level-$2$ nodes can therefore be labeled in $O\bigl(\frac{n}{\log n}\bigr)$ rounds.
    This completes the proof.
\end{proof}

\subsection{Lower Bound}
\begin{theorem}\label{thm:randomized_k_2_lower}
    $2$-hierarchical $2\frac{1}{2}$-coloring requires $\Omega\bigl(\frac{n}{\log n}\bigr)$ locality to be solved by a randomized algorithm without knowledge of $n$.
\end{theorem}
\begin{proof}
    Let $\A$ be a randomized algorithm solving $2\frac{1}{2}$-coloring without knowledge of $n$, in time $o\bigl(\frac{n}{\log n}\bigr)$.
    Let $\P$ be an infinite path on $\Z$ and with edges $\{j,j+1\}$ for all $j\in\Z$.
    Let $T'_i$ be the random variable measuring the time for the node with ID $i$ to stop on $\P$.
    Note that because $\A$ does not know $n$, this $T'_i$ will be independent of $n$.
    If, for all $n_0$, there exists $i$ such that $\Pr(T'_i\geq n_0) \geq 1/2$, then one can take arbitrarily long subpaths of $\P$ on which the expected round complexity of $\A$ is $\Omega(n)$ and thus too large.
    Otherwise, there is some $n_0$ such that for every $i$, $\Pr(T'_i\leq n_0)\geq 1/2$.
    
    Now, we can build a large level-1 path $P$ by joining up $\Omega(\log n)$ disjoint subpaths of length at least $2n_0+1$ of $\P$. The center of each subpath will terminate within $n_0$ rounds with probability at least $1/2$ and therefore, the whole level-1 path $P$ cannot be $2$-colored with high probability. Recall that to solve $2\frac{1}{2}$-coloring, a level-$1$ path either needs to be properly $2$-colored or all nodes on the path need to output $D$ (decline). If two nodes at distance more than $2n_0$ both terminate within time $n_0$, it is not possible to $2$-color the path (we do not guarantee that the joined short segments of $\P$ together form a long segment of $\P$. With high probability, all nodes on the level-$1$ path therefore have to output $D$.

    If we build $\Theta(n/\log n)$ disjoint such level-$1$ paths of length $\Theta(\log n)$ each and join the end nodes on one side of those paths to a level-$2$ path, we create a level-$2$ path of length $\Omega(n/\log n)$ for which each node is a adjacent to a level-$1$ path that outputs decline. The definition of $2\frac{1}{2}$-coloring then forces this path to be $2$-colored, which takes $\Omega(n/\log n)$ rounds even with randomization.
\end{proof}

\subsection{\texorpdfstring{$3$}{3}-hierarchical \texorpdfstring{$2\frac{1}{2}$}{2.5}-coloring}

\begin{definition}[Interesting and friendly nodes]\label{def:friendly}
    Consider some graph $G$ and let each node be assigned the same level as in the $3$-hierarchical $2\frac{1}{2}$-coloring problem.
    Assume that a partial output for $3$-hierarchical $2\frac{1}{2}$-coloring is given that assigns an output label to exactly those nodes that are in level $1$.
    We call a node of level $2$ \emph{interesting} if all its neighbors of level $1$ output $D$.
    Similarly for a partial output for $3$-hierarchical $2\frac{1}{2}$-coloring given to all nodes in levels $1$ and $2$, we call a node of level $3$ \emph{interesting} if all its neighbors of level $1$ and $2$ output $D$.

    For any positive integer $i$ and any level-$2$ node $v \in V(G)$, let $Q_i(v)$ be the set of all level-$1$ nodes that can be reached from $v$ via a path of length at most $i$ that contains only level-$1$ nodes except for $v$.
    Let $f:\mathbb{R}_{\geq 1} \rightarrow \mathbb{R}_{\geq 0}$ be a monotonically increasing function.
    We call a path $P$ consisting of $i$ level-$2$ nodes \emph{$f$-friendly} if
    $|\bigcup_{v \in V(P)} Q_i(v)| < i \cdot f(i)$, and \emph{$f$-unfriendly} otherwise.

\end{definition}

\begin{theorem}\label{thm:randomized_k_3_upper}
    Let $g:\mathbb{R}_{\geq 1} \rightarrow \mathbb{R}_{\geq 0}$ be a monotonically increasing function satisfying
    $g(g(x)) = \log x$ for any real number $x \geq 2$ (cf.\ \Cref{subsec:halflog})
    Then there exists a randomized Las Vegas algorithm that w.h.p.\ solves $3$-hierarchical $2\frac{1}{2}$-coloring with locality $O(n/g(n))$ without knowledge of $n$.
\end{theorem}
\begin{proof}
    Define a function $f:\mathbb{R}_{\geq 1} \rightarrow \mathbb{R}_{\geq 0}$ by setting $f(x) := 1/3 \cdot g(x)$, for each real number $x \geq 1$.
    Consider the following algorithm $\fA$, that we provide from a global perspective.
    We will show later how to implement it in the distributed setting.
    
    Each node starts by computing its level.
    Then, each level-$1$ node marks itself with probability $1/2$.
    Consider some arbitrary node $v$.
    
    If $v$ is of level $1$, do the following.
    Let $P_v$ denote the level-$1$ path containing $v$.
    If $P_v$ contains a marked node, then $v$ outputs $D$; if $P_v$ does not contain a marked node, then the nodes of $P_v$ output a correct $2$-coloring.

    If $v$ is of level $2$, do the following.
    If $v$ is not interesting, then $v$ outputs $E$.
    If $v$ is interesting, then let $P'_v$ denote the maximal level-$2$ path containing $v$ that consists of interesting nodes.
    If $P'_v$ contains an $f$-friendly subpath, then $v$ outputs $D$.
    If $P'_v$ does not contain an $f$-friendly subpath, then the nodes of $P'_v$ output a correct $2$-coloring.
    
    If $v$ is of level $3$, do the following.
    If $v$ is not interesting, then $v$ outputs $E$.
    If $v$ is interesting, then the nodes on the maximal path of interesting nodes containing $v$ output a correct $2$-coloring.

    If $v$ is of level $4$ (i.e., the remaining nodes), the definition of $2\frac{1}{2}$-coloring in \Cref{sec:2.5Col} does not make any requirements. All nodes of level $4$ can therefore output $D$.
    
    For the distributed implementation, each node simply does the following: in each round it sends around all the information it has received so far (and in the very first round the information it has initially) until it has received sufficient information itself to determine its output, upon which it sends its output to its neighbors and terminates.
    More precisely, if a node is supposed to output $D$ or $E$, it will do so as soon as it has gathered sufficient information to determine that it satisfies one of the respective criteria for outputting $D$ or $E$ outlined in the algorithm description.
    Similarly, if a node is supposed to output a color, then it does so as soon as it has gathered sufficient information to determine that it satisfies one of the respective criteria \emph{and} either complete knowledge of its connected component of nodes that output a color or knowledge of a path to a node in this component and of the color that this node outputs.
    We remark that for each connected component of nodes to be colored, a $2$-coloring can be fixed via the random bits of the nodes (with probability $1$), which is why ``seeing'' the whole component suffices for a node to determine its color.

    In the following, we argue that it holds w.h.p.\ that each node $v$ has received the aforementioned sufficient information after $O(n/f(n))$ rounds.
    To this end, we consider the different types of nodes one by one.
    Note that each node can determine its level in a constant number of rounds, determining the output of level-$4$ nodes (which is always $D$) and deciding whether level-$1$ a node marks itself does not require any communication, which is why we ignore these two steps in the following. 

    If $v$ is a level-$1$ node, we consider two cases.
    If $v$ is in a level-$1$ path of length at most $5 \log n$, then it terminates after at most $5 \log n$ rounds.
    If it is in a level-$1$ path of length greater than $5 \log n$, then it has a marked node in distance at most $5 \log n$ with probability at least $1 - 1/2^{5 \log n} = 1 - 1/n^5$, which implies that it terminates after at most $5 \log n$ round with probability at least $1 - 1/n^5$.

    Next, consider the case that $v$ is a level-$2$ node.
    From the discussion of the level-$1$ nodes, we obtain by a union bound that with probability at least $1 - 1/n^4$, each level-$2$ node knows after $5 \log n + 1$ rounds whether it is interesting or not.
    Hence, if $v$ is not interesting it terminates with probability at least $1 - 1/n^4$ after $5 \log n + 1$ rounds.
    Now consider the case that $v$ is interesting.
    From the definition of an $f$-unfriendly path, we know that no $f$-unfriendly path in the input tree $G$ can be longer than $i$, where $i$ is the real number satisfying $i \cdot f(i) = n$.
    The function $f(n)$ is growing slowly enough such that $f(n) = O(f(n/f(n)))$.
    We therefore obtain that the longest $f$-unfriendly path in $G$ is of length $O(n/f(n))$.
    Hence, after each level-$2$ node learned the output of its level-$1$ neighbors, node $v$ can determine in $O(n/f(n))$ rounds whether the path $P'_v$ from the algorithm description contains an $f$-friendly subpath.
    Note that we use here that for a level-$2$ node to determine whether a level-$2$ path of length $j$ it is contained in is friendly takes only $O(j)$ rounds (provided that each level-$2$ nodes knows the outputs of its level-$1$ neighbors).
    In conclusion, $v$ terminates with probability at least $1 - 1/n^4$ after $O(n/f(n) + 5 \log n + 1) = O(n/f(n))$ rounds.

    Next, consider the case that $v$ is a level-$3$ node.
    From the discussion of the lower-level nodes, we obtain by a union bound that with probability at least $1 - 1/n^3$, each level-$3$ node knows after $O(n/f(n))$ rounds whether it is interesting or not.
    If $v$ is not interesting, it terminates with probability at least $1 - 1/n^3$ after $O(n/f(n))$ rounds.
    Now consider the case that $v$ is interesting.

    Consider an arbitrary (not necessarily maximal) level-$3$ path $P''$, and let $v'$ be an arbitrary node on $P''$.
    Since $v'$ is a level-$3$ node it must have at least one adjacent level-$2$ node.
    Let $u$ be such a node and let $P''_u$ be the maximal level-$2$ path containing $u$ (in which $u$ is necessarily a node of degree $1$).   
    Let $u'$ be the node on $P''_u$ that is closest to $u$ with the property that the subpath $P^*$ of $P''_u$ from $u$ to $u'$ contains an $f$-friendly subpath.
    Let $j$ be the smallest positive integer such that $P^*$ contains an $f$-friendly subpath consisting of precisely $j$ nodes.
    In the following, we prove an upper bound for the probability of $v'$ being an interesting node (as a function of $j$).

    Let $\hat{P} = (u_1, \dots, u_j)$ be an $f$-friendly subpath of $P^*$ consisting of precisely $j$ nodes.
    For each $1 \leq r \leq j$, let $w_r$ be a level-$1$ node adjacent to $u_r$ (which exists by the definition of a level-$2$ node) and let $P_r$ be the maximal level-$1$ path containing $w_r$ and consisting of at most $j$ nodes.
    For each $1 \leq r \leq j$, the probability that $u_r$ is interesting is at most $1-1/2^{|V(P_r)|}$.
    This implies that the probability that \emph{all} $u_r$ are interesting is at most $\prod_{1 \leq r \leq j} (1-1/2^{|V(P_r)|})$ (due to the independence of the events that an individual $u_r$ is interesting).
    Moreover, we know that
    \begin{equation}\label{eq:friendlysum}
        \sum_{1 \leq r \leq j} |V(P_r)| \leq j \cdot f(j)
    \end{equation}
    since $P^*$ is $f$-friendly.
    Observe that $\prod_{1 \leq r \leq j} (1-1/2^{|V(P_r)|})$ is maximized under~\Cref{eq:friendlysum} if $|V(P_r)| = f(j)$ for each $1 \leq r \leq j$.
    To see this, assume that this is not the case.
    Then in the maximizing choice for the $|V(P_r)|$, there must be two members $|V(P_{r_1})|$ and $|V(P_{r_2})|$ satisfying $|V(P_{r_1})| + 1 \leq f(j) \leq |V(P_{r_2})| - 1$.
    But then $\prod_{1 \leq r \leq j} (1-1/2^{|V(P_r)|})$ can be increased further by replacing $|V(P_{r_1})|$ by $|V(P_{r_1})| + 1$ and $|V(P_{r_2})|$ by $|V(P_{r_1})| - 1$ since
    \begin{align*}
        \left(1-1/2^{|V(P_{r_1})|}\right)\left(1-1/2^{|V(P_{r_2})|}\right) &< 1-1/2^{|V(P_{r_1})|}\\
        &< \left(1-1/2^{|V(P_{r_1})| + 1}\right)^2\\
        &\leq \left(1-1/2^{|V(P_{r_1})| + 1}\right)\left(1-1/2^{|V(P_{r_2})| - 1}\right),
    \end{align*}
    yielding a contradiction.

    Hence, the probability that all $u_r$ are interesting is at most $\left(1-1/2^{f(j)}\right)^{j}$.
    Observe that, by the definition of $u'$ (and the design of $\fA$), node $u$ outputs $D$ only if all nodes on $P^*$ output $D$, and the latter condition can only be satisfied if all $u_r$ output $D$, which in turn can only happen if all $u_r$ are interesting.
    Thus, $u$ outputs $D$ with probability at most $\left(1-1/2^{f(j)}\right)^{j}$.
    It follows that $v'$ is interesting with probability at most $\left(1-1/2^{f(j)}\right)^{j}$.
    In particular, if $f(j) < 1/2 \cdot \log n$, then $v'$ is interesting with probability at most $1-1/\sqrt{n}$.

    Now, let $n$ be sufficiently large, and consider an arbitrary path $P$ of level-$3$ nodes of length precisely $2n/f(n)$.
    If there are more than $4 \log n \cdot \sqrt{n}$ nodes $v'$ on $P$ for which the respective $j$ in the above calculations satisfies $f(j) < 1/2 \cdot \log n$, then with probability at least $1 - 1/n^4$, path $P$ contains a node that is not interesting.
    
    Hence, consider the case that there are at most $4 \log n \cdot \sqrt{n}$ such nodes $v'$ on $P$.
    It follows that there are at least $2n/f(n) - 4 \log n \cdot \sqrt{n} \geq n/f(n)$ nodes on $P$ for which the respective $j$ in the above calculations satisfies $f(j) \geq 1/2 \cdot \log n$.    
    Let $V''$ denote the set of these nodes, and consider some node $v'' \in V''$.
    Let $u''$ be a level-$2$ neighbor of $v''$.
    From the above discussion (and definitions) and the definition of an $f$-unfriendly path, we obtain that, if $v''$ is interesting, then the number of nodes that can be reached from $v''$ via edge $\{ v'', u'' \}$ is at least $(j' - 1) \cdot f(j' - 1)$, where $j'$ is the smallest positive integer $j'$ satisfying $f(j') \geq 1/2 \cdot \log n$.
    This implies that the number of nodes that can be reached from $v''$ via edge $\{ v'', u'' \}$ is at least $f^{-1}(1/2 \cdot \log n) - 1$ (as $f(j'  - 1) \geq 1$ for sufficiently large $n$).

    Therefore, if all nodes in $V''$ are interesting, we obtain that the input tree contains at least $f^{-1}(1/2 \cdot \log n) \cdot n/f(n)$ nodes, which implies $f^{-1}(1/2 \cdot \log n) \leq f(n)$ (as the input tree contains $n$ nodes), which in turn implies $1/2 \cdot \log n \leq f(f(n)) \leq f(g(n)) = 1/3 \cdot g(g(n)) = 1/3 \cdot \log n$, yielding a contradiction.
    Hence, in the considered case, there must be at least one node in $V''$ that is not interesting, concluding the consideration of that case.

    We obtain that the probability that $P$ contains a node that is not interesting is at least $1 - 1/n^4$.
    Since the input tree contains at most $n$ level-$3$ paths of length precisely $2n/f(n)$, it follows by a union bound that the probability that there exists a (not necessarily maximal) level-$3$ path of length precisely $2n/f(n)$ that consists only of interesting nodes is at most $1/n^3$.
    
    Now let us come back to our consideration of the interesting level-$3$ node $v$.
    By the above discussion, we conclude that, with probability at least $1 - 2/n^3 \geq 1 - 1/n^2$ (for sufficiently large $n$), node $v$ terminates after $O(n/f(n))$ rounds.
    
    Hence, each node of the input tree terminates with probability at least $1 - 1/n^2$ after $O(n/f(n))$ rounds, which implies that, with probability at least $1 - 1/n$, Algorithm $\fA$ terminates after $O(n/f(n))$ rounds.
    Since $g(n) = \Theta(f(n))$, the lemma statement follows.\
\end{proof}

\begin{theorem}\label{thm:randomized_k_3_lower}
    Let $g:\mathbb{R}_{\geq 1} \rightarrow \mathbb{R}_{\geq 0}$ be a monotonically increasing function satisfying
    $g(g(x)) = \log x$ for any real number $x \geq 2$. 
    Then any a randomized Las Vegas algorithm that w.h.p.\ solves $3$-hierarchical $2\frac{1}{2}$-coloring requires locality $\Omega(n/(\log n \cdot g(n))$ without knowledge of $n$.
\end{theorem}
\begin{proof}
    Let $\fA$ be an arbitrary randomized Las Vegas algorithm that w.h.p.\ solves $3$-hierarchical $2\frac{1}{2}$-coloring (without knowledge of $n$).
    Consider first the case that there exists no positive integer $c$ such that the probability that, at the middle node of a path of length $6c$, Algorithm $\fA$ terminates after at most $c$ rounds and outputs $D$ is at least $1/2$.
    We claim that then, for any sufficiently large $c$, the probability that Algorithm $\fA$ does not terminate after $c$ rounds on a path $P$ of length $6c$ is at least $1/4$.

    To prove our claim, assume for a contradiction that $\fA$ terminates on $P$ after $c$ rounds with probability at least $3/4$.
    Consider two nodes $v$ and $w$ on $P$ that have distance $2c + 1$ from each other and greater than $c$ from both endpoints of the path. 
    The probability that $v$ terminates after at most $c$ rounds and outputs a color is at least $3/4 - 1/2 = 1/4$ (as $v$'s $c$-hop view is isomorphic to the $c$-hop view of the middle node of $P$), and the same holds for $w$.
    Since the $c$-hop neighborhoods of $v$ and $w$ are isomorphic and non-overlapping, this implies that there is some color that either node outputs with probability at least $1/8$.
    As in any correct $2$-coloring, $v$ and $w$ must output different colors (as they have odd distance), it follows that the output of $\fA$ is incorrect with nonzero constant probability, yielding a contradiction (for sufficiently large $c$).
    Hence, the claim is true, which yields an $\Omega(n)$-round lower bound for $\fA$.
    
    Now consider the complementary case, and let $c$ be the smallest positive integer such that the probability that $\fA$ outputs $D$ at the middle node of a path of length $6c$ is at least $1/2$.
    Fix some arbitrary integer $i > 300c$.
    For any positive integer $j$, let $G_j$ denote the graph consisting of a path $P_j$ of length $6j$ together with a path of length $i$ attached at every node of $P_j$.
    Set $n_j := |V(G_j)|$.
    It follows that $n_j = (i + 2)(6j + 1)$.
    We consider two subcases.

    First, consider the case that, when executing $\fA$ on $G_j$, there exists no positive integer $j$ such that the probability that the middle node of $P_{j}$ terminates after at most $j$ rounds and outputs $D$ is at least $1/2$.
    Consider some $j$ such that $30c \log n_{j} \leq i \leq 30c(\log n_j + 1)$ (which exists by the definition of $n_j$ and implies $n_j > 300$).
    Consider $G_{j}$.    
    By splitting each of the $6j + 1$ paths attached to the nodes of $P_j$ into $5\log n_j$ chunks of length at least $6c$ each, we see that, for each such path, the probability that it outputs $D$ is at least $1 - 1/n_j^5$.
    Hence, the probability that no node on $P_j$ outputs $E$ is at least $1 - 1/n_j^4$.
    By an argument analogous to above, it follows that the probability that the middle node of $P_{j}$ terminates after $j$ rounds and outputs a color is smaller than $1/8 + 1/n_j^4$ (where we use that $n_j > 300$).
    Thus, the probability that the middle node of $P_j$ terminates after at most $j$ rounds is at most $1/2 + 1/8 + 2/n_j^4$, which implies that, on $G_j$, the probability that $\fA$ does not terminate after at most $j$ rounds is greater than $1/n_j$.
    Note that, if we initially fix $i$ to be sufficiently large, then we have $j > n_j/(300\log n_j)$.
    
    Now, consider the complementary subcase, and fix $j$ to be the smallest positive integer such that, when executing $\fA$ on $G_j$, the probability that the middle node of $P_{j}$ terminates after at most $j$ rounds and outputs $D$ is at least $1/2$.
    Moreover, let $f:\mathbb{R}_{\geq 1} \rightarrow \mathbb{R}_{\geq 0}$ be the function satisfying $f(x) := 30c \cdot g(x)$, for each real number $x \geq 1$.
    Again, we consider two subcases (of this subcase).

    Consider first the case that $j$ is such that $i < f(n_j)$.
    By the definition of $j$, for each $1 \leq j' < j$, the probability that the middle node of $P_{j'}$ terminates after at most $j'$ rounds and outputs $D$ is smaller than $1/2$.
    Now, with an argumentation analogous to above (and using the fact that $f$ grows faster than the $\log$-function), we obtain that there is some $j'$ satisfying $30c \log n_{j'} \leq i < f(n_{j'})$, such that the probability that, on $G_{j'}$, Algorithm $\fA$ terminates after at most $n_{j'}/(300g(n_{j'}))$ rounds is greater than $1/n_{j'}$ (if $i$ is chosen sufficiently large initially).

    Finally, consider the complementary subcase, i.e., that $j$ is such that $i \geq f(n_j)$.
    For any positive integer $\ell$, let $G_{\ell, \ell'}$ be the graph obtained as follows: start with a path $P'_{\ell}$ of length $\ell$, then attach at each node of $P'_{\ell}$ a path of length  $\ell'$, and finally attach at each node of each such path of length $\ell'$ a path of length $i$.
    Set $n_{\ell,\ell'} := |V(G_{\ell,\ell'})|$.
    It follows that $n_{\ell,\ell'} = (\ell + 1)((i + 2)(\ell' + 1) + 1)$.
    Now let $\ell, \ell'$ be such that
    \begin{enumerate}
        \item $i \geq 30c \log n_{\ell,\ell'}$ and
        \item $\ell' \geq 18j \log n_{\ell, \ell'}$.
    \end{enumerate}
    More precisely, let $\ell, \ell'$ be so that they maximize $n_{\ell,\ell'}$ with the mentioned properties.
    Observe that (due to the fact that, for the $\log$-function, a multiplicative change in the argument only results in an additive change of the function value) for sufficiently large $i$, there exists some constant $c'$ that is independent of the chosen $i$ (only requiring that $i$ is sufficiently large and such that $j$ exists and satisfies $i \geq f(n_j)$) such that
    \begin{enumerate}
        \item $i \leq c' \cdot 30c \log n_{\ell,\ell'}$ and
        \item $\ell' \leq c' \cdot 18j \log n_{\ell, \ell'}$.
    \end{enumerate}
    Since we also have $i \geq f(n_j)$, i.e., $n_j \leq f^{-1}(i)$, we obtain
    \[
        \ell \geq n_{\ell,\ell'}/(300c' \cdot 18 \log n_{\ell,\ell'} \cdot f^{-1}(30c \log n_{\ell,\ell'})).
    \]
    Since $f(f(n_{\ell,\ell'})) \geq f(g(n_{\ell,\ell'})) = 30c \cdot g(g(n_{\ell,\ell'})) = 30c \cdot \log n_{\ell,\ell'}$, we obtain $f(n_{\ell,\ell'}) \geq f^{-1}(30c \cdot \log n_{\ell,\ell'})$.
    Hence, we obtain
    \[
        \ell \geq n_{\ell,\ell'}/(100c' \cdot 18 \log n_{\ell,\ell'} \cdot f(n_{\ell,\ell'})).
    \]
    Observe that, with an argumentation analogous to above (and applying two more union bounds), the probability that all nodes on the attached paths of length $\ell'$ output $D$ is at least $1 - 1/n^2$.
    Since, in this case, the nodes on the path of length $\ell$ need to output a proper $2$-coloring, we obtain that the probability that $\fA$ terminates after at most $\ell/6$ rounds is smaller than $1 - 1/n$.
    This concludes the individual considerations of all of the above cases.
    
    Using that $g(n) = \Theta(f(n))$, we obtain that there is some positive universal constant $\beta$, such that, for any sufficiently large $i$, there is some graph with $n > i$ nodes such that the probability that $\fA$ terminates on this graph after at most $\beta \cdot n/(\log n \cdot g(n))$ rounds is smaller than $1 - 1/n$.
    This implies that there are infinitely many $n$ such that there exists a graph $G$ with $n$ nodes such that $\fA$ does not terminate after at most $\beta \cdot n/(\log n \cdot g(n))$ rounds with probability larger that $1/n$.
    This implies the claimed lower bound of $\Omega(n/(\log n \cdot g(n)))$ rounds.
\end{proof}

\end{document}